%% file: ms.tex
\documentclass[]{fundam}
\usepackage{graphicx}
\usepackage{tikz}

\usepackage{amsmath,tabu}
\usepackage{amssymb,amsfonts,textcomp}
\usepackage{stmaryrd}
\usepackage{booktabs}
\usepackage[ruled, linesnumbered, vlined]{algorithm2e} 

\usepackage{subcaption}
\usetikzlibrary{automata, positioning, arrows}

\def\rem{$\mathit{REM}$}
\def\sra{$\mathit{SRA}$}
\def\srem{$\mathit{SREM}$}
\def\ssra{$\mathit{sSRA}$}
\def\ssrem{$\mathit{sSREM}$}
\def\wsrem{$\mathit{wSREM}$}
\def\gsra{$\mathit{gSRA}$}
\def\dsra{$\mathit{dSRA}$}
\def\nsra{$\mathit{nSRA}$}
\def\usra{$\mathit{uSRA}$}
\def\pst{$\mathit{PST}$}

\newcommand*{\FI}{}%

\begin{document}

%
\setcounter{page}{1}
\publyear{2021}
\papernumber{0001}
\volume{178}
\issue{1}
%

\title{Symbolic Register Automata for Complex Event Recognition and Forecasting}


\author{Elias Alevizos\\
Department of Informatics,\\ 
National and Kapodistrian University of Athens, Greece\\ 
Institute of Informatics \& Telecommunications,\\ 
National Center for Scientific Research ``Demokritos'', Greece\\
ilalev{@}di.uoa.gr,alevizos.elias{@}iit.demokritos.gr
\and 
Alexander Artikis\\
Department of Maritime Studies, \\
University of Piraeus, Greece \\ 
Institute of Informatics \& Telecommunications,\\ 
National Center for Scientific Research ``Demokritos'', Greece \\
a.artikis{@}unipi.gr
\and 
Georgios Paliouras\\
Institute of Informatics \& Telecommunications,\\ 
National Center for Scientific Research ``Demokritos'', Greece \\
paliourg{@}iit.demokritos.gr} 

\maketitle

\runninghead{E. Alevizos, A. Artikis, G. Paliouras}{Symbolic Automata with Memory for CER/F}

\begin{abstract}
We propose an automaton model which is a combination of symbolic and register automata, i.e.,
we enrich symbolic automata with memory.
We call such automata Symbolic Register Automata (\sra).
\sra\ extend the expressive power of symbolic automata,
by allowing Boolean formulas to be applied not only to the last element read from the input string,
but to multiple elements, 
stored in their registers. 
\sra\ also extend register automata, 
by allowing arbitrary Boolean formulas, 
besides equality predicates.
We study the closure properties of \sra\ under union, intersection, concatenation, Kleene closure, complement and determinization and show that \sra,
contrary to symbolic automata,
are not in general closed under complement and they are not determinizable.
However,
they are closed under these operations when a window operator, 
quintessential in Complex Event Recognition, 
is used.
We show how \sra\ can be used in Complex Event Recognition in order to detect patterns upon streams of events,
using our framework that provides declarative and compositional semantics, 
and that allows for a systematic treatment of such automata.
We also show how the behavior of \sra,
as they consume streams of events, 
can be given a probabilistic description with the help of prediction suffix trees.
This allows us to go one step beyond Complex Event Recognition to Complex Event Forecasting,
where,
besides detecting complex patterns,
we can also efficiently forecast their occurrence. 
\end{abstract}


\input{intro}

\input{related}

\input{srem}

\input{sra}

\input{sra_properties}

\input{streaming_sra}

\input{complexity}
\input{markov}

\input{outro}

\begin{acknowledgements}
This work has received funding from the EU Horizon 2020 program INFORE under grant agreement No 825070 and from the European Horizon 2020 program ARIADNE under grant agreement No 871464.
\end{acknowledgements}

\bibliographystyle{fundam}
\bibliography{refs}

\input{appendix}


\end{document}

%% file: intro.tex
\section{Introduction}
\label{sec:intro}

A Complex Event Recognition (CER) system takes as input a stream of events,
along with a set of patterns,
defining relations among the input events,
and detects instances of pattern satisfaction,
thus producing an output stream of complex events \cite{DBLP:journals/vldb/GiatrakosAADG20,DBLP:books/daglib/0017658,DBLP:journals/csur/CugolaM12}.
Typically, an event has the structure of a tuple of values which might be numerical or categorical.
Since time is of critical importance for CEP,
a temporal formalism is used in order to define the patterns to be detected.
Such a pattern imposes temporal (and possibly atemporal) constraints on the input events,
which, if satisfied, lead to the detection of a complex event.
Atemporal constraints may be ``local'',
applying only to the last event read from the input stream.
For example, 
in streams from temperature sensors,
the constraint that the temperature of the last event is higher than some constant threshold would constitute such a local constraint.
Alternatively, these constraints might involve multiple events of the pattern,
e.g., 
the constraint that the temperature of the last event is higher than that of the previous event.
The input to a CER system thus consists of two main components: 
a stream of events, 
also called simple derived events (SDEs);
and a set of patterns that define relations among the SDEs.
Instances of pattern satisfaction are called Complex Events (CEs).
The output of the system is another stream, 
composed of the detected CEs.
CEs must often be detected with very low latency,
which, in certain cases, 
may even be in the order of a few milliseconds \cite{DBLP:books/daglib/0017658,DBLP:books/daglib/0024062,hedtstuck_complex_2017}.

Automata are of particular interest for the field of CER,
because they provide a natural way of handling sequences.
As a result, 
the usual operators of regular expressions, 
concatenation, union and Kleene-star,
have often been given an implicit temporal interpretation in CER.
For example, 
the concatenation of two events is said to occur whenever the second event is read by an automaton after the first one,
i.e.,
whenever the timestamp of the second event is greater than the timestamp of the first 
(assuming the input events are temporally ordered).
On the other hand,
atemporal constraints are not easy to define using classical automata,
since they either work without memory or, 
even if they do include a memory structure,
e.g., as with push-down automata,
they can only work with a finite alphabet of input symbols.
For this reason,
the CER community has proposed several extensions of classical automata.
These extended automata have the ability to store input events and later retrieve them in order to evaluate whether a constraint is satisfied \cite{DBLP:conf/cidr/DemersGPRSW07,DBLP:conf/sigmod/AgrawalDGI08,DBLP:journals/csur/CugolaM12}.
They resemble both register automata \cite{DBLP:journals/tcs/KaminskiF94},
through their ability to store events,
and symbolic automata \cite{DBLP:conf/cav/DAntoniV17},
through the use of predicates on their transitions.
They differ from symbolic automata in that predicates apply to multiple events, 
retrieved from the memory structure that holds previous events.
They differ from register automata in that predicates may be more complex than that of (in)equality.

One issue with these automata is that their properties have not been systematically investigated,
as is the case with models derived directly from the field of languages and automata
(see \cite{DBLP:conf/icdt/GrezRU19} for a discussion about the weaknesses of automaton models in CER). 
Moreover, they sometimes need to impose restrictions on the use of regular expression operators in a pattern, 
e.g., nesting of Kleene-star operators is not allowed.
A recently proposed formal framework for CER attempts to address these issues \cite{DBLP:conf/icdt/GrezRU19}.
Its advantage is that it provides a logic for CER patterns, 
with simple denotational and compositional semantics, 
but without imposing severe restrictions on the use of operators.
An automaton model is also proposed which may be conceived as a variation of symbolic transducers \cite{DBLP:conf/cav/DAntoniV17}. 
However, this automaton model can only handle ``local'' constraints,
i.e.,
the formulas on their transitions are unary and thus are applied only to the last event read.

We propose an automaton model that is a combination of symbolic and register automata.
It has the ability to store events and its transitions have guards in the form of $n$-ary conditions.
These conditions may be applied both to the last event and to past events that have been stored. 
Conditions on multiple events are crucial in CER because they allow us to express many patterns of interest,
e.g., an increasing trend in the speed of a vehicle.
We call such automata \emph{Symbolic Register Automata} (\sra).
\sra\ extend the expressive power of symbolic and register automata, 
by allowing for more complex patterns to be defined and detected on a stream of events.
We also present a language with which we can define patterns for complex events that can then be translated to \sra.
We call such patterns \emph{Symbolic Regular Expressions with Memory} (\srem),
as an extension of the work presented in \cite{DBLP:journals/jcss/LibkinTV15},
where \emph{Regular Expression with Memory} (\rem) are defined and investigated.
\rem\ are extensions of classical regular expressions with which we allow some of the terminal symbols of an expression to be stored and later be compared for (in)equality.
\srem\ allow for more complex conditions to be used,
besides those of (in)equality.

We then show how \srem\ and \sra\ may be used in order to perform Complex Event Forecasting (CEF).
Our solution allows a user to define a pattern for a complex event in the form of a \srem.
It then constructs a probabilistic model for such a pattern in order to forecast,
on the basis of an event stream,
if and when a complex event is expected to occur.
We use prediction suffix trees \cite{DBLP:journals/ml/RonST96,DBLP:conf/nips/RonST93} to learn a probabilistic model for the pattern and the \sra\ corresponding to this pattern.    
We have already presented how symbolic automata without registers may be combined with prediction suffix trees for the purpose of CEF
\cite{DBLP:journals/vldb/AlevizosAP21}.
We show here when and how symbolic automata with memory can be combined with prediction suffix trees for the same purpose.
Prediction suffix trees fall under the class of the so-called variable-order Markov models.
They are Markov models whose order (how deep into the past they can look for dependencies) can be increased beyond what is computationally possible with full-order models. 
They can do this by avoiding a full enumeration of every possible dependency and focusing only on ``meaningful'' dependencies.
Efficient and early CEF would thus allow analysts to take proactive action when critical situations are expected to happen,
e.g., to alert maritime authorities for the possible collision of vessels at sea. 

The contributions of the paper may be summarized as follows:
\begin{itemize}
	\item We present a language for CER, 
	Symbolic Regular Expressions with Memory (\srem).
	\item We present a computational model for patterns written in \srem, 
	Symbolic Register Automata (\sra),
	whose main feature is that it allows for relating multiple events in a pattern.
	Constraints with multiple events are essential in CER,
	since they are required in order to capture many patterns of interest,
	e.g., an increasing or decreasing trend in stock prices.
	\item We show that \sra\ and \srem\ are equivalent,
	i.e., they accept the same set of languages.
	\item We study the closure properties of \sra\ (and \srem). 
	We show that, 
	in the general case, 
	they are closed under the most usual operators (union, intersection, concatenation and Kleene-star), 
	but not under complement and determinization.
	Failure of closure under complement implies that negation cannot be arbitrarily (i.e., in a compositional manner) used in CER patterns.
	The negative result about determinization implies that certain techniques requiring deterministic automata,
	like the ones we will describe later for event forecasting,
	are not applicable.
	\item We show that, 
	by using windows, 
	\sra\ are able to retain their nice closure properties,
	i.e., they remain closed under complement and determinization.
	Windows are an indispensable operator in CER because, among others, they limit the search space when attempting to find matches for a pattern.
	\item We show how \sra\ with windows can be combined with Prediction Suffix Trees in order to perform CEF,
	thus extending our previous work from symbolic automata to symbolic register automata \cite{DBLP:journals/vldb/AlevizosAP21}. 
\end{itemize}
All proofs and complete algorithms may be found in the Appendix,
Section \ref{section:appendix}. 
Please, note that the results of this paper are presented with CER in mind.
However, we need to stress that they are not restricted to CER.
They are general results,
applicable to any strings and not just to streams of events.
In fact,
one may treat CER as a special case of string processing.
Thus, our contributions lie both in the more specific field of CER and in the more general one of formal languages and automata theory.

\begin{example}
\label{example:sensors}
\begin{table}[t]
\centering
\caption{Example stream.}
\begin{tabular}{cccccccc} 
\toprule
type & T & T & T & H & H & T & ... \\ 
\midrule
id & 1 & 1 & 2 & 1 & 1 & 2 & ... \\
\midrule
value & 22 & 24 & 32 & 70 & 68 & 33 & ... \\
\midrule
index & 1 & 2 & 3 & 4 & 5 & 6 & ... \\
\bottomrule
\end{tabular}
\label{table:example_stream}
\end{table}
We now introduce an example which will be used throughout the paper to provide intuition
(borrowed from \cite{DBLP:conf/icdt/GrezRU19}).
The example is that of a set of sensors taking temperature and humidity measurements,
monitoring an area for the possible eruption of fires.
A stream is a sequence of input events,
where each such event is a tuple of the form $(\mathit{type},\mathit{id},\mathit{value})$.
The first attribute ($\mathit{type}$) is the type of measurement: 
$H$ for humidity and $T$ for temperature.
The second one ($\mathit{id}$) is an integer identifier, 
unique for each sensor.
It has a finite set of possible values.
Finally, the third one ($\mathit{value}$) is the real-valued measurement from a possibly infinite set of values.
Table \ref{table:example_stream} shows an example of such a stream.
We assume that events are temporally ordered and their order is implicitly provided through the index. $\diamond$
\end{example}

The structure of the paper is as follows.
In Section \ref{sec:related} we present the state-of-the-art in automata theory with respect to automaton models that can store elements and models that can use more complex conditions on their transitions than simple equality.
In Section \ref{sec:grammar} we discuss extensively the grammar and the semantics of \srem.
Next,
in Section \ref{sec:sra} we define \sra\ and the languages that they recognize.
In Section \ref{sec:sra_properties} we show that \sra\ and \srem\ are equivalent.
We show that \sra\ and \srem\ are closed under union, intersection, concatenation and Kleene-star,
but not under complement and determinization.
We also define windowed \srem\ and \sra\ and show that windows make \sra\ and \srem\ closed under complement and determinization.
In Section \ref{sec:streaming_sra},
we discuss how \sra\ can be used for CER,
whereas Section \ref{sec:complexity} briefly discusses some complexity issues.
Subsequently,
Section \ref{sec:markov} discusses how \sra\ can be given a probabilistic description through the use of variable-order Markov models in order to perform CEF.
We conclude with Section \ref{sec:outro} where we summarize our contributions and discuss possible avenues for future work.

%% file: related.tex
\section{Related Work}
\label{sec:related}

Because of their ability to naturally handle sequences of characters,
automata have been extensively adopted in CER,
where they are adapted in order to handle streams composed of tuples.
Typical cases of CER systems that employ automata are
the Chronicle Recognition System \cite{DBLP:conf/kr/Ghallab96,DBLP:conf/ijcai/DoussonM07},
Cayuga \cite{DBLP:conf/cidr/DemersGPRSW07},
TESLA \cite{DBLP:conf/debs/CugolaM10} and
SASE \cite{DBLP:conf/sigmod/AgrawalDGI08,DBLP:conf/sigmod/ZhangDI14}.
There also exist systems that do not employ automata as their computational model,
e.g., 
there are logic-based systems \cite{DBLP:journals/tkde/ArtikisSP15} or systems that use trees \cite{DBLP:conf/sigmod/MeiM09},
but the standard operators of concatenation, union and Kleene-star are quite common and they may be considered as a reasonable set of core operators for CER.
For an overview of CER languages, 
see \cite{DBLP:journals/vldb/GiatrakosAADG20},
and for a general review of CER systems, 
see \cite{DBLP:journals/csur/CugolaM12}.

However, 
current CER systems do not have the full expressive power of regular expressions,
e.g.,
SASE does not allow for nesting Kleene-star operators. 
Moreover, 
due to the various approaches implementing the basic operators and extensions in their own way,
there is a lack of a common ground that could act as a basis for systematically understanding the properties of these automaton models.
The abundance of different CER systems,
employing various computational models and using various formalisms 
has recently led to some attempts at providing a unifying framework 
\cite{DBLP:conf/icdt/GrezRU19,DBLP:journals/corr/Halle17}.
Specifically, 
in \cite{DBLP:conf/icdt/GrezRU19},
a set of core CER operators is identified,
a formal framework is proposed that provides denotational semantics for CER patterns, 
and a computational model is described for capturing such patterns.

Outside the field of CER,
research on automata has evolved towards various directions.
Besides the well-known push-down automata that can store elements from a finite set to a stack,
there have appeared other automaton models with memory,
such as register automata, 
pebble automata and 
data automata \cite{DBLP:journals/tcs/KaminskiF94,DBLP:journals/tocl/NevenSV04,DBLP:journals/tocl/BojanczykDMSS11}.
For a review, 
see \cite{DBLP:conf/csl/Segoufin06}.
Such models are especially useful when the input alphabet cannot be assumed to be finite,
as is often the case with CER.
Register automata (initially called finite-memory automata) constitute one of the earliest such proposals \cite{DBLP:journals/tcs/KaminskiF94}.
At each transition,
a register automaton may choose to store its current input 
(more precisely, the current input's data payload)
to one of a finite set of registers.
A transition is followed if the current input is equal with the contents of some register.
With register automata,
it is possible to recognize strings constructed from an infinite alphabet,
through the use of (in)equality comparisons among the data carried by the current input and the data stored in the registers.
However,
register automata do not always have nice closure properties,
e.g.,
they are not closed under determinization
For an extensive study of register automata, 
see \cite{DBLP:journals/jcss/LibkinTV15,DBLP:conf/lpar/LibkinV12}.
We build on the framework presented in \cite{DBLP:journals/jcss/LibkinTV15,DBLP:conf/lpar/LibkinV12} in order to construct register automata with the ability to handle ``arbitrary'' structures,
besides those containing only (in)equality relations.

Another model that is of interest for CER is the symbolic automaton,
which allows CER patterns to apply constraints on the attributes of events.
Automata that have predicates on their transitions were already proposed in \cite{DBLP:journals/grammars/NoordG01}.
This initial idea has recently been expanded and more fully investigated in symbolic automata 
\cite{DBLP:conf/lpar/VeanesBM10,DBLP:conf/wia/Veanes13,DBLP:conf/cav/DAntoniV17}.
In this automaton model,
transitions are equipped with formulas constructed from a Boolean algebra.
A transition is followed
if its formula,
applied to the current input,
evaluates to true.
The work presented in \cite{DBLP:conf/icdt/GrezRU19,DBLP:conf/icdt/GrezRUV20} may also be categorized under this class of ``unary'' symbolic automata 
(or transducers, to be more precise).
Contrary to register automata,
symbolic automata have nice closure properties,
but their formulas are unary and thus can only be applied to a single element from the input string.

This is one limitation that we address in this paper.
We propose an automaton model, 
called \textit{Symbolic Register Automata} (\sra), 
whose transitions can apply $n$-ary formulas/conditions (with $n{>}1$) on multiple elements.
\sra\ are thus more expressive than symbolic and register automata,
thus being suitable for practical CER applications,
while, at the same time,
their properties can be systematically investigated,
as in standard automata theory.
In fact, our model subsumes these two automaton models as special cases. 

We also show how this new automaton model can be given a probabilistic description in order to perform forecasting,
i.e., predict the occurrence of a complex event before it is actually detected by the automaton.
However, 
forecasting has not received much attention in the field of CER,
despite the fact that it is an active research topic in various related research areas,
such as time-series forecasting \cite{montgomery2015introduction},
sequence prediction \cite{DBLP:journals/jair/BegleiterEY04,DBLP:journals/ml/RonST96,DBLP:journals/tcom/ClearyW84,DBLP:journals/tit/WillemsST95}, 
temporal mining \cite{DBLP:conf/icdm/VilaltaM02,DBLP:conf/kdd/LaxmanTW08,DBLP:journals/eswa/ZhouCG15,DBLP:journals/vldb/ChoWYZC11} and
event sequence prediction and point-of-interest recommendations through neural networks \cite{DBLP:conf/ijcai/LiDL18,DBLP:conf/ijcai/ChangPPKK18}.
These methods are powerful in predicting the next numerical value(s) in a time-series or the next input event(s) in a sequence of events,
but they suffer from limitations that render them unsuitable for CEF.
In CEF we are interested in both numerical and categorical values, 
related through complex patterns and involving multiple variables.
Such patterns require a language to be defined,
much like SQL in databases.
Our goal is to forecast the occurrence of such complex events defined via patterns and not input events.
Input event forecasting is actually not very useful for CER,
since the majority of input events are ignored, 
without contributing to the detection of complex events.
The number of complex events is typically orders of magnitude lower than that of input events.

Some conceptual proposals have acknowledged the need for CEF though \cite{DBLP:conf/bci/FulopBTDVF12,DBLP:conf/debs/EngelE11,DBLP:conf/edoc/ChristKK16}.
In what follows,
we briefly present the relatively few previous concrete attempts at CEF.
The first such attempt at CEF was presented in \cite{DBLP:conf/debs/MuthusamyLJ10},
where a variant of regular expressions and automata was used to define complex event patterns,
along with Markov chains.
Each automaton state was mapped to a Markov chain state.
Symbolic automata and Markov chains were again used in \cite{DBLP:conf/debs/AlevizosAP17,DBLP:conf/lpar/AlevizosAP18}.
The problem with these approaches is that they are essentially unable to encode higher-order dependencies,
since high-order Markov chains may lead to a combinatorial explosion of the number of states. 
In \cite{DBLP:conf/colcom/PandeyNC11},
complex events were defined through transitions systems and Hidden Markov Models (HMM) were used to construct a probabilistic model.
The observable variable of the HMM corresponded to the states of the transition system.
HMMs are in general more powerful than Markov chains,
but,
in practice, the may be hard to train (\cite{DBLP:journals/jair/BegleiterEY04,DBLP:journals/ml/AbeW92}) and require elaborate domain modeling,
since mapping a pattern to a HMM is not straightforward.
In contrast, 
our approach constructs seamlessly a probabilistic model from a given CE pattern (declaratively defined).
Knowledge graphs were used in in \cite{DBLP:journals/pvldb/LiGC20} to encode events and their timing relationships. 
Stochastic gradient descent was employed to learn the weights of the graph's edges that determine how important an event is with respect to another target event.
However, this approach falls in the category of input event forecasting, 
as it does not target complex events.

%% file: srem.tex
\section{A Grammar for Symbolic Regular Expressions with Memory}
\label{sec:grammar}

Before presenting \sra,
we first present a high-level formalism for defining CER patterns.
We extend the work presented in \cite{DBLP:journals/jcss/LibkinTV15},
where the notion of regular expressions with memory (\rem) was introduced.
These regular expressions can store some terminal symbols in order to compare them later against a new input element for (in)equality.
One important limitation of \rem\ with respect to CER is that they can handle only (in)equality relations.
In this section,
we extend \rem\ so as to endow them with the capacity to use relations from ``arbitrary'' structures.
We call these extended \rem\ \emph{Symbolic Regular Expressions with Memory} (\srem).

First, 
in Section \ref{sec:formulas_models} we repeat some basic definitions from logic theory.
We also describe how we can adapt them and simplify them to suit our needs.
Next,
in Section \ref{sec:conditions} we precisely define the notion of conditions.
In \srem,
conditions will act in a manner equivalent to that of terminal symbols in classical regular expressions.
The difference is of course that conditions are essentially logic formulas that can reference both the current element read from a string/stream and possibly some past elements.
Finally,
in Section \ref{sec:srem} we provide a precise definition for \srem\ and their semantics.

\subsection{Formulas and Models}
\label{sec:formulas_models}

In this section,
we follow the notation and notions presented in \cite{hedman2004first}.
The first notion that we need is that of a $\mathcal{V}$-structure.
A $\mathcal{V}$-structure essentially describes a domain along with the operations that can be performed on the elements of this domain and their interpretation. 
\begin{definition}[$\mathcal{V}$-structure \cite{hedman2004first}]
A vocabulary $V$ is a set of function, relation and constant symbols.
A $\mathcal{V}$-structure is an underlying set $\mathcal{U}$, called a universe, and an interpretation of $\mathcal{V}$.
An interpretation assigns an element of $\mathcal{U}$ to each constant in $\mathcal{V}$, a function from $\mathcal{U}^{n}$ to $\mathcal{U}$ to each $n$-ary function in $\mathcal{V}$ and a subset of $\mathcal{U}^{n}$ to each $n$-ary relation in $\mathcal{V}$.
$\blacktriangleleft$
\end{definition}

\begin{example}
Using Example \ref{example:sensors},
we can define the following vocabulary
\begin{equation*}
\mathcal{V} = \{R,c_{1},c_{2},c_{3},c_{4},c_{5},c_{6}\}
\end{equation*}
and the universe 
\begin{equation*}
\mathcal{U} = \{(T,1,22),(T,1,24),(T,2,32),(H,1,70),(H,1,68),(T,2,33)\}
\end{equation*}
We can also define an interpretation of $V$ by assigning each $c_{i}$ to an element of $\mathcal{U}$,
e.g., $c_{1}$ to $(T,1,22)$, $c_{2}$ to $(T,1,24)$, etc.
$R$ may also be interpreted as $R(x,y) := x.\mathit{id} = y.\mathit{id}$,
i.e., this binary relation contains all pairs of $\mathcal{U}$ which have the same $\mathit{id}$.
For example, $((T,1,22),(H,1,70)) \in R$ and  $((T,1,22),(T,2,33)) \notin R$.
If there are more (even infinite) tuples in a stream/string,
then we would also need more constants (even infinite). $\diamond$
\end{example}

We extend the terminology from classical regular expressions to define characters, strings and languages.
Elements of $\mathcal{U}$ are called \emph{characters} and finite sequences of characters are called \emph{strings}. 
A set of strings $\mathcal{L}$ constructed from elements of $\mathcal{U}$ ($\mathcal{L} \subseteq \mathcal{U}^{*}$, 
where $^{*}$ denotes Kleene-star) is called a language over $\mathcal{U}$.
We can also define streams as follows.
A stream $S$ is an infinite sequence $S=t_{1},t_{2},\cdots$, 
where each $t_{i}$ is a character ($t_{i} \in \mathcal{U}$).
By $S_{1..k}$ we denote the sub-string of $S$ composed of the first $k$ elements of $S$.
$S_{m..k}$ denotes the slice of $S$ starting from the $m^{th}$ and ending at the $k^{th}$ element.

We now define the syntax and semantics of formulas that can be constructed from the constants, relations and functions of a $\mathcal{V}$-structure.
We begin with the definition of terms.
\begin{definition}[Term \cite{hedman2004first}]
A term is defined inductively as follows:
\begin{itemize}
	\item Every constant is a term.
	\item If $f$ is an $m$-ary function and $t_{1}, \cdots, t_{m}$ are terms,
	then $f(t_{1}, \cdots, t_{m})$ is also a term. $\blacktriangleleft$
\end{itemize}
\end{definition}

Using terms, relations and the usual Boolean constructs of conjunction, disjunction and negation,
we can define formulas.
\begin{definition}[Formula \cite{hedman2004first}]
\label{definition:formula}
Let $t_{i}$ be terms.
A formula is defined as follows:
\begin{itemize}
	\item If $P$ is an $n$-ary relation, then $P(t_{1}, \cdots, t_{n})$ is a formula (an atomic formula).
	\item If $\phi$ is a formula, $\neg \phi$ is also a formula.
	\item If $\phi_{1}$ and $\phi_{2}$ are formulas, $\phi_{1} \wedge \phi_{2}$ is also a formula.
	\item If $\phi_{1}$ and $\phi_{2}$ are formulas, $\phi_{1} \vee \phi_{2}$ is also a formula. $\blacktriangleleft$
\end{itemize}
\end{definition}
\begin{definition}[$\mathcal{V}$-formula \cite{hedman2004first}]
If $\mathcal{V}$ is a vocabulary, 
then a formula in which every function, relation and constant is in $\mathcal{V}$ is called a $\mathcal{V}$-formula.
$\blacktriangleleft$
\end{definition}

\begin{example}
Continuing with our example,
$R(c_{1},c_{4})$ is an atomic $\mathcal{V}$-formula. 
$R(c_{1},c_{4}) \wedge \neg R(c_{1},c_{3})$ is also a (complex) $\mathcal{V}$-formula,
where $\mathcal{V} = \{R,c_{1},c_{2},c_{3},c_{4},c_{5},c_{6}\}$. $\diamond$
\end{example}

Notice that in typical definitions of terms and formulas
(as found in \cite{hedman2004first})
variables are also present.
A variable is also a term.
Variables are also used in existential and universal quantifiers to construct formulas.
In our case, 
we will not be using variables in the above sense
(instead, as explained below, we will use variables to refer to registers).
Thus, existential and universal formulas will not be used.
In principle, 
they could be used, 
but their use would be counter-intuitive.
At every new event,
we need to check whether this event satisfies some properties,
possibly in relation to previous events.
A universal or existential formula would need to check every event
(variables would refer to events),
both past and future,
to see if all of them or at least one of them (from the universe $\mathcal{U}$) satisfy a given property.
Since we will not be using variables,
there is also no notion of free variables in formulas
(variables occurring in formulas that are not quantified).
Thus, every formula is also a sentence,
since sentences are formulas without free variables.
In what follows,
we will thus not differentiate between formulas and sentences.

We can now define the semantics of a formula with respect to a $\mathcal{V}$-structure.
\begin{definition}[Model of $\mathcal{V}$-formulas \cite{hedman2004first}]
\label{definition:models_formulas}
Let $\mathcal{M}$ be a $\mathcal{V}$-structure and $\phi$ a $\mathcal{V}$-formula.
We define $\mathcal{M} \models \phi$ ($\mathcal{M}$ models $\phi$) as follows:
\begin{itemize}
	\item If $\phi$ is atomic, i.e. $\phi = P(t_{1}, \cdots, t_{m})$, then $\mathcal{M} \models P(t_{1}, \cdots, t_{m})$ iff the tuple $(a_{1}, \cdots, a_{m})$ is in the subset of $\mathcal{U}^{m}$ assigned to $P$,
	where $a_{i}$ are the elements of $\mathcal{U}$ assigned to the terms $t_{i}$.
	\item If $\phi := \neg \psi$, then $\mathcal{M} \models \phi$ iff $\mathcal{M} \nvDash \psi$. 
	\item If $\phi := \phi_{1} \wedge \phi_{2}$, then $\mathcal{M} \models \phi$ iff $\mathcal{M} \models \phi_{1}$ and $\mathcal{M} \models \phi_{2}$. 
	\item If $\phi := \phi_{1} \vee \phi_{2}$, then $\mathcal{M} \models \phi$ iff $\mathcal{M} \models \phi_{1}$ or $\mathcal{M} \models \phi_{2}$. $\blacktriangleleft$
\end{itemize}
\end{definition}

\begin{example}
If $\mathcal{M}$ is the $\mathcal{V}$-structure of our example,
then $\mathcal{M} \models R(c_{1},c_{4})$,
since $c_{1} \rightarrow (T,1,22)$, $c_{1} \rightarrow (H,1,70)$ and $((T,1,22),(H,1,70)) \in R$.
We can also see that $\mathcal{M} \models R(c_{1},c_{4}) \wedge \neg R(c_{1},c_{3})$,
since $c_{3} \rightarrow (T,2,32)$ and $((T,1,22),(T,2,32)) \notin R$. $\diamond$
\end{example}

\subsection{Conditions}
\label{sec:conditions}

Based on the above definitions, 
we will now define conditions over registers. 
These will essentially be the $n$-ary guards on the transitions of \sra.

\begin{definition}[Condition]
\label{definition:condition}
Let $\mathcal{M}$ be a $\mathcal{V}$-structure always equipped with the unary relation $\top$ for which it holds that $u \in \top$, ${\forall u \in \mathcal{U}}$,
i.e., this relation holds for all elements of the universe $\mathcal{U}$.
Let $R = \{r_{1}, \cdots, r_{k}\}$ be variables denoting the registers and $\sim$ a special variable denoting an automaton's head which reads new elements. 
The ``contents'' of the head  always correspond to the most recent element.
We call them register variables.
A condition is essentially a $\mathcal{V}$-formula,
as defined above (Definition \ref{definition:formula}),
where,
instead of terms,
we use register variables.
A condition is defined by the following grammar:
\begin{itemize}
	\item $\top$ is a condition.
	\item $P(r_{1}, \cdots, r_{n})$, where $r_{i} \in R \cup \{ \sim \}$ and $P$ an $n$-ary relation, is a condition.
	\item $\neg \phi$ is a condition, if $\phi$ is a condition.
	\item $\phi_{1} \wedge \phi_{2}$ is a condition if $\phi_{1}$ and $\phi_{2}$ are conditions.
	\item $\phi_{1} \vee \phi_{2}$ is a condition if $\phi_{1}$ and $\phi_{2}$ are conditions. $\blacktriangleleft$
\end{itemize}
\end{definition}

Since terms now refer to registers,
we need a way to access the contents of these registers.
We will assume that each register has the capacity to store exactly one element from $\mathcal{U}$.
The notion of valuations provides us with a way to access the contents of registers.
\begin{definition}[Valuation]
A valuation on $R = \{r_{1}, \cdots, r_{k}\}$ is a partial function $v: R \hookrightarrow \mathcal{U}$.
The set of all valuations on $R$ is denoted by $F(r_{1}, \cdots, r_{k})$.
$v[r_{i} \leftarrow u]$ denotes the valuation where we replace the content of $r_{i}$ with a new element $u$:
\begin{equation}
v'(r_{j}) = v[r_{i} \leftarrow u] = 
  \begin{cases}
    u & \quad \text{if } r_{j} = r_{i}   \\
    v(r_{j}) & \quad \text{otherwise} \\
  \end{cases}
\end{equation}
$v[W \leftarrow u]$, 
where $W \subseteq R$,
denotes the valuation obtained by replacing the contents of all registers in $W$ with $u$.
We say that a valuation $v$ is compatible with a condition $\phi$ if, 
for every register variable $r_{i}$ that appears in $\phi$,
$v(r_{i})$ is defined.
$\blacktriangleleft$
\end{definition}
A valuation $v$ is essentially a function with which we can retrieve the contents of any register.
We will also use the notation $v(r_{i}) = \sharp$ to denote the fact that register $r_{i}$ is empty,
i.e., we extend the range of $v$ to $\mathcal{U} \cup \{ \sharp \}$.
We also extend the domain of $v$ to $R \cup \{ \sim \}$.
By $v(\sim)$ we will denote the ``contents'' of the automaton's head,
i.e., the last element read from the string.

We can now define the semantics of conditions,
similarly to the way we defined models of $\mathcal{V}$-formulas in Definition \ref{definition:models_formulas}.
The difference is that the arguments to relations are no longer elements assigned to terms but elements stored in registers,
as retrieved by a given valuation.
\begin{definition}[Semantics of conditions]
\label{definition:condition_semantics}
Let $\mathcal{M}$ be a $\mathcal{V}$-structure, 
$u \in \mathcal{U}$ an element of the universe of $\mathcal{M}$ 
and $v \in F(r_{1}, \cdots, r_{k})$ a valuation.
We say that a condition $\phi$ is satisfied by $(u,v)$,
denoted by $(u,v) \models \phi$, 
iff one of the following holds:
\begin{itemize}
	\item $\phi := \top$, i.e., $(u,v) \models \top$ for every element and valuation.
	\item $\phi := P(x_{1}, \cdots, x_{n})$, $x_{i} \in R \cup \{ \sim \}$, $v(x_{i})$ is defined for all $x_{i}$ and $u \in P(v(x_{1}), \cdots, v(x_{n}))$.
	\item $\phi := \neg \psi$ and $(u,v) \nvDash \psi$.
	\item $\phi := \phi_{1} \wedge \phi_{2}$, $(u,v) \models \phi_{1}$ and $(u,v) \models \phi_{2}$. 
	\item $\phi := \phi_{1} \vee \phi_{2}$, $(u,v) \models \phi_{1}$ or $(u,v) \models \phi_{2}$. $\blacktriangleleft$
\end{itemize}
\end{definition}

\subsection{Symbolic Regular Expressions with Memory}
\label{sec:srem}

We are now in a position to define Symbolic Regular Expressions with Memory \srem.
We achieve this by combining conditions via the standard regular operators.
Conditions act as terminal ``symbols'',
as the base case from which we construct more complex expressions.
\begin{definition}[Symbolic regular expression with memory (\srem)]
\label{definition:srem}
A symbolic regular expression with memory over a $\mathcal{V}$-structure $\mathcal{M}$ and a set of register variables $R = \{r_{1}, \cdots, r_{k}\}$ is inductively defined as follows:
\begin{enumerate}
	\item $\epsilon$ and $\emptyset$ are \srem.
	\item If $\phi$ is a condition (as in Definition \ref{definition:condition}), then $\phi$ is a \srem.
	\item If $\phi$ is a condition, then $\phi \downarrow r_{i}$ is a \srem. 
	This is the case where we need to store the current element read from the automaton's head to register $r_{i}$.
	\item If $e_{1}$ and $e_{2}$ are \srem, then $e_{1} + e_{2}$ is also a \srem. This corresponds to disjunction. 
	\item If $e_{1}$ and $e_{2}$ are \srem, then $e_{1} \cdot e_{2}$ is also a \srem. This corresponds to concatenation. 
	\item If $e$ is a \srem, then $e^{*}$ is also a \srem. This corresponds to Kleene-star. $\blacktriangleleft$
\end{enumerate} 
\end{definition}

In order to define the semantics of \srem,
we need to define precisely how the contents of the registers may change.
We thus need to define how a \srem,
starting from a given valuation $v$ and reading a given string $S$,
reaches another valuation $v'$.
\begin{definition}[Semantics of \srem]
\label{definition:srem_semantics}
Let $e$ be a \srem\ over a $\mathcal{V}$-structure $\mathcal{M}$ and a set of register variables $R = \{r_{1}, \cdots, r_{k}\}$, 
$S$ a string constructed from elements of the universe of $\mathcal{M}$ 
and $v,v' \in F(r_{1}, \cdots, r_{k})$.
We define the relation $(e,S,v) \vdash v'$ as follows (a textual explanation is provided after the formal definition):
\begin{enumerate}
	\item $(\epsilon,S,v) \vdash v'$ iff $S = \epsilon$ and $v=v'$.
	\item $(\phi, S, v) \vdash v'$ iff $\phi \neq \epsilon$, $S=u$, $(u,v) \models \phi$ and $v' = v$. 
	\item $(\phi \downarrow r_{i}, S, v) \vdash v'$ iff $S=u$, $(u,v) \models \phi$ and $v' = v[r_{i} \leftarrow u]$. 
	\item $(e_{1} \cdot e_{2}, S, v) \vdash v'$ iff $S=S_{1} \cdot S_{2}$: $(e_{1},S_{1},v) \vdash v''$ and $(e_{2},S_{2},v'') \vdash v'$.
	\item $(e_{1} + e_{2}, S, v) \vdash v'$ iff $(e_{1},S,v) \vdash v'$ or $(e_{2},S,v) \vdash v'$.
	\item $(e^{*}, S, v) \vdash v'$ iff 
	\begin{equation*}
  \begin{cases}
    S=\epsilon\ \text{and } v'=v & \quad \text{or }   \\
    S=S_{1} \cdot S_{2}: (e,S_{1},v) \vdash v''\ \text{and }  (e^{*},S_{1},v'') \vdash v' & \quad \text{} \\
  \end{cases}
	\end{equation*} $\blacktriangleleft$
\end{enumerate}
\end{definition}
In the first case,
we have an $\epsilon$ \srem.
It may reach another valuation only if it reads an $\epsilon$ string and this new valuation is the same as the initial one, i.e., the registers do not change.
In the second case where we have a condition $\phi \neq \epsilon$,
we move to a new valuation only if the condition is satisfied with the current element and the given register contents.
Again, the registers do not change.
The third case is similar to the second,
with the important difference that the register $r_{i}$ needs to change and to store the current element.
For the fourth case (concatenation),
we need to be able to break the initial string into two sub-strings such that the first one reaches a certain valuation and the second one can start from this new valuation and reach another one.
The fifth case is a disjunction.
Finally, 
the sixth case implies that we must be able to break the initial string into multiple sub-strings such that each one of these substring can reach a valuation and the next one can start from this valuation and reach another one.

Based on the above definition,
we may now define the language that a \srem\ accepts (as in \cite{DBLP:journals/jcss/LibkinTV15}).
The language of a \srem\ contains all the strings with which we can reach a valuation,
starting from the empty valuation,
where all registers are empty.
\begin{definition}[Language accepted by a \srem]
\label{definition:language_srem}
We say that $(e,S,v)$ infers $v'$ if $(e,S,v) \vdash v'$.
We say that $e$ induces $v$ on a string $S$ if $(e,S,\sharp) \vdash v$,
where $\sharp$ denotes the valuation in which no $v(r_{i})$ is defined,
i.e., all registers are empty.
The language accepted by a \srem\ $e$ is defined as $\mathcal{L}(e) = \{S \mid (e,S,\sharp) \vdash v\}$ for some valuation $v$.
$\blacktriangleleft$
\end{definition}

\begin{example}
As an example,
consider the following \srem\
\begin{equation}
\label{srem:t_seq_h_filter_eq_id}
e_{1} := (\mathit{TypeIsT}(\sim) \downarrow r_{1}) \cdot (\top)^{*} \cdot (\mathit{TypeIsH}(\sim) \wedge \mathit{EqualId}(\sim,r_{1}))
\end{equation}
where we assume that a) $\mathit{TypeIsT}(x) := x.\mathit{type} = T$, b) $\mathit{TypeIsH}(x) := x.\mathit{type} = H$ and c) $\mathit{EqualId}(x,y) := x.\mathit{id} = y.\mathit{id}$.
If we feed the string/stream of Table \ref{table:example_stream} to $e_{1}$,
then we will have the following.
We will initially read the first element $(T,1,22)$.
Since its type is $T$,
we will move on and store $(T,1,22)$ to register $r_{1}$,
i.e., we will move from the empty valuation where $v(r_{1}) = \sharp$ to $v'$,
where $v'(r_{1}) = (T,1,22)$.
Then, the sub-expression $(\top)^{*}$ lets us skip any number of elements.
We can thus skip the second and third elements without changing the register contents.
Now, upon reading the fourth element $(H,1,70)$,
there are two options.
Either skip it again to read the fifth element or try to move on by checking the sub-expression $(\mathit{TypeIsH}(\sim) \wedge \mathit{EqualId}(\sim,r_{1}))$.
This condition is actually satisfied,
since the type of this element is indeed $H$ and its $\mathit{id}$ is equal to the $\mathit{id}$ of the element store in $r_{1}$.
Thus, $S_{1..4}$ is indeed accepted by $e_{1}$.
With a similar reasoning we can see that the same is also true for $S_{1..5}$. $\diamond$
\end{example}

%% file: sra.tex
\section{Symbolic Register Automata}
\label{sec:sra}

In order to capture \srem, 
we propose Symbolic Register Automata (\sra), 
an automaton model equipped  with memory and logical conditions on its transitions.
The basic idea is the following.
We add a set of registers $R$ to an automaton in order to be able to store elements from the string/stream that will be used later in $n$-ary conditions. 
Each register can store at most one element.
In order to evaluate whether to follow a transition or not,
each transition is equipped with a guard, 
in the form of a condition.
If the condition evaluates to true, 
then the transition is followed.
Since a condition might be $n$-ary, 
with $n{>}1$,
the values passed to its arguments during evaluation may be either the current element
or the contents of some registers,
i.e.,
some past elements.
In other words, 
the transition is also equipped with a \textit{register selection},
i.e., a tuple of registers.
Before evaluation, 
the automaton reads the contents of those registers,
passes them as arguments to the condition and the condition is evaluated.
Additionally, if, during a run of the automaton, 
a transition is followed,
then the transition has the option to write the element that triggered it
to some of the automaton's registers.
These are called its \textit{write registers},
i.e.,
the registers whose contents may be changed by the transition.
We also allow for $\epsilon$-transitions, 
as in classical automata,
i.e., 
transitions that are followed without consuming any elements and without altering the contents of the registers.

We now formally define \sra.
To aid understanding,
we present three separate definitions:
one for the automaton itself,
one for its configurations and one for its runs.

\begin{definition}[Symbolic Register Automaton]
\label{definition:sra}
A symbolic register automaton (\sra) with $k$ registers over a $\mathcal{V}$-structure $\mathcal{M}$ is a tuple ($Q$, $q_{s}$, $Q_{f}$, $R$, $\Delta$)
where 
\begin{itemize}
	\item $Q$ is a finite set of states,
	\item $q_{s} \in Q$ the start state, 
	\item $Q_{f}\subseteq Q$ the set of final states, 
	\item $R = (r_{1}, \cdots, r_{k})$ a finite set of registers and
	\item $\Delta$ the set of transitions.
\end{itemize}
A transition $\delta \in \Delta$ is a tuple $(q,\phi,W,q')$, 
also written as $q,\phi \downarrow W  \rightarrow q'$,
where
\begin{itemize}
	\item $q,q' \in Q$, 
	\item $\phi$ is a condition, as defined in Definition \ref{definition:condition} or $\phi = \epsilon$ and
	\item $W \in 2^{R}$ are the write registers. $\blacktriangleleft$
\end{itemize}
\end{definition}

We will use the dot notation to refer to elements of tuples.
For example, 
if $A$ is a \sra, 
then $A.Q$ is the set of its states.
For a transition $\delta$,
we will also use the notation $\delta.\mathit{source}$ and $\delta.\mathit{target}$ to refer to its source and target states respectively.

\begin{figure}[t]
\begin{centering}
\includegraphics[width=0.55\linewidth]{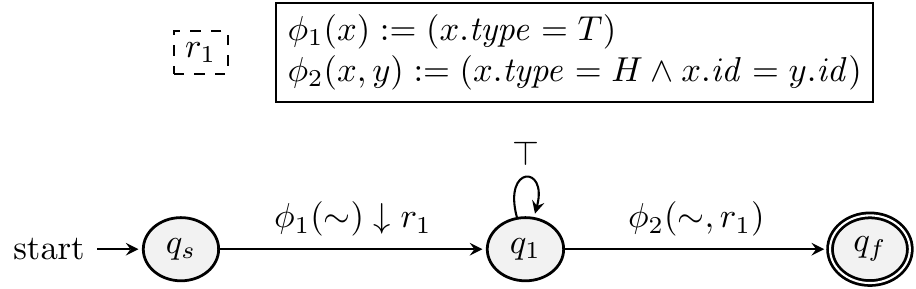}
\caption{\sra\ corresponding to Expression \eqref{srem:t_seq_h_filter_eq_id}.}
\label{fig:example1}
\end{centering}
\end{figure}

\begin{example}
As an example,
consider the \sra\ of Figure \ref{fig:example1}.
Each transition is represented as $\phi \downarrow W$,
where $\phi$ is its condition and $W$ its set of write registers 
(or simply $r_{i}$ if only a single register is written).
$W$ may also be an empty set,
implying that no register is written.
In this case, 
we avoid writing $W$ on the transition
(see, for example, the transition from $q_{1}$ to $q_{f}$ in Figure \ref{fig:example1}).
The definitions for the conditions of the transitions are presented in a separate box,
above the \sra.
Note that the arguments of the conditions correspond to registers, through the register selection.
Take the transition from $q_{s}$ to $q_{1}$ as an example.
It takes the last event consumed from the stream ($\sim$) and
passes it as argument to the unary formula $\phi_{1}$.
If $\phi_{1}$ evaluates to true,
it writes this last event to register $r_{1}$,
displayed as a dashed square in Figure \ref{fig:example1}.
On the other hand, 
the transition from $q_{1}$ to $q_{f}$ uses both the current event and the event stored in $r_{1}$ ($(\sim,r_{1})$) and passes them to the binary formula $\phi_{2}$.
The condition $\top$ (in the self-loop of $q_{1}$) is a unary condition that always evaluates to true and allows us to skip any number of events.
The \sra\ of Figure \ref{fig:example1} captures \srem\ \eqref{srem:t_seq_h_filter_eq_id}. $\diamond$
\end{example}

We can describe formally the rules for the behavior of a \sra\ through the notion of configuration:
\begin{definition}[Configuration of \sra]
\label{definition:configuration}
Assume a string
$S=t_{1},t_{2},\cdots,t_{l}$ 
and a \sra\ $A$ consuming $S$.
A configuration of $A$ is a triple
$c=[j,q,v] \in \mathbb{N} \times Q \times F(r_{1}, \cdots, r_{k})$, 
where
\begin{itemize}
	\item $j$ is the index of the next event/character to be consumed,
	\item $q$ is the current state of $A$ and
	\item $v$ the current valuation, i.e., the current contents of $A$'s registers.
\end{itemize}
We say that $c'=[j',q',v']$ is a \emph{successor} of $c$ iff one of the following holds:
\begin{itemize}
	\item $\exists \delta: \delta.\mathit{source} = q,\ \delta.\mathit{target}=q',\ \delta.\phi = \epsilon,\ j'=j,\ v'=v$, i.e., if this is an $\epsilon$ transition, we move to the target state without changing the index or the registers' contents.
	\item $\exists \delta: \delta.\mathit{source} = q,\ \delta.\mathit{target}=q',\ \delta.W = \emptyset,\ (t_{j},v) \models \delta.\phi,\ j'=j+1,\ v'=v$, i.e., if the condition is satisfied according to the current event and the registers' contents and there are no write registers, we move to the target state, we increase the index by 1 and we leave the registers untouched.
	\item $\exists \delta: \delta.\mathit{source} = q,\ \delta.\mathit{target}=q',\ \delta.W \neq \emptyset,\ (t_{j},v) \models \delta.\phi,\ j'=j+1,\ v'=v[W \leftarrow t_{j}]$, i.e., if the condition is satisfied according to the current event and the registers' contents and there are write registers, we move to the target state, we increase the index by 1 and we replace the contents of all write registers (all $r_{i} \in W$) with the current element from the string. $\blacktriangleleft$
\end{itemize}
\end{definition}
We denote a succession by $[j,q,v] \rightarrow [j',q',v']$,
or $[j,q,v] \overset{\delta}{\rightarrow} [j',q',v']$ if we need to refer to the transition as well.
For the initial configuration,
before any elements have been consumed,
we assume that
$j=1$, $q=q_{s}$ and $v(r_{i}) = \sharp,\ \forall r_{i} \in R$.
In order to move to a successor configuration,
we need a transition whose condition evaluates to true,
when applied to $\sim$, 
if it is unary, 
or to $\sim$ and the contents of its register selection, 
if it is $n$-ary.
If this is the case, 
we move one position ahead in the stream and update the contents of this transition's write registers,
if any, 
with the event that was read. 
If the transition is an $\epsilon$-transition, 
we do not move the stream pointer and do not update the registers,
but only move to the next state.

The actual behavior of a \sra\ upon reading a stream is captured by the notion of the run:
\begin{definition}[Run of \sra\ over string/stream]
A run $\varrho$ of a \sra\ $A$ over a stream $S=t_{1},\cdots,t_{n}$ is a sequence of successor configurations
$[1,q_{1},v_{1}] \overset{\delta_{1}}{\rightarrow} [2,q_{2},v_{2}] \overset{\delta_{2}}{\rightarrow} \cdots \overset{\delta_{n}}{\rightarrow} [n+1,q_{n+1},v_{n+1}]$.
A run is called accepting iff $q_{n+1} \in A.Q_{f}$.
$\blacktriangleleft$
\end{definition}

\begin{example}
A run of the \sra\ of Figure \ref{fig:example1}, 
while consuming the first four events from the stream of Table \ref{table:example_stream}, 
is the following:
\begin{equation}
\label{run:example}
[1,q_{s},\sharp] \overset{\delta_{s,1}}{\rightarrow} [2,q_{1},(T,1,22)] \overset{\delta_{1,1}}{\rightarrow} [3,q_{1},(T,1,22)] \overset{\delta_{1,1}}{\rightarrow}
[4,q_{1},(T,1,22)] \overset{\delta_{1,f}}{\rightarrow} [5,q_{f},(T,1,22)]
\end{equation}
Transition subscripts in this example refer to states of the \sra,
e.g.,
$\delta_{s,s}$ is the transition from the start state to itself,
$\delta_{s,1}$ is the transition from the start state to $q_{1}$, etc.
See also Figure \ref{fig:example_run}.
Run \eqref{run:example} is not the only run,
since the \sra\ could have followed other transitions with the same input,
e.g.,
moving directly from $q_{s}$ to $q_{1}$.
Another possible (and non-accepting) run would be the one where the \sra\ always remains in $q_{1}$ after its first transition.
\begin{figure}
\centering
\begin{subfigure}[t]{0.49\textwidth}
	\includegraphics[width=0.99\textwidth]{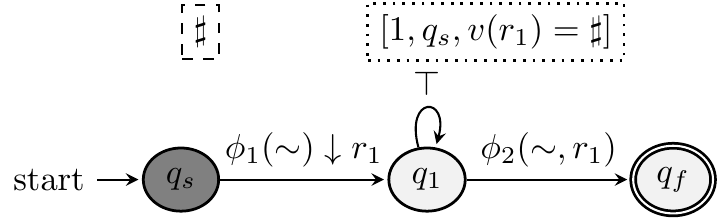}
	\caption{Initial configuration.}
\end{subfigure}
\begin{subfigure}[t]{0.49\textwidth}
	\includegraphics[width=0.99\textwidth]{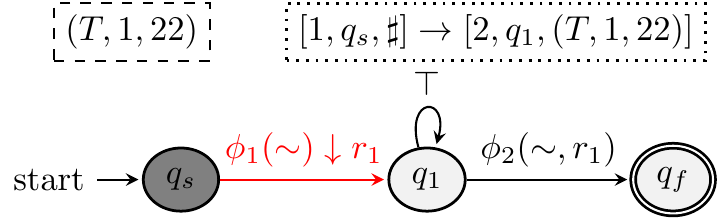}
	\caption{Configuration after reading $t_{1}$.}
\end{subfigure}\\
\begin{subfigure}[t]{0.49\textwidth}
	\includegraphics[width=0.99\textwidth]{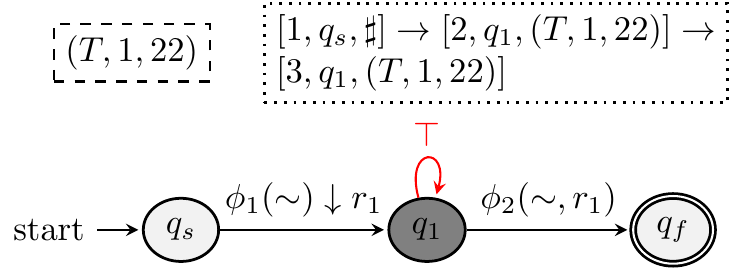}
	\caption{Configuration after reading $t_{2}$.}
\end{subfigure}
\begin{subfigure}[t]{0.49\textwidth}
	\includegraphics[width=0.99\textwidth]{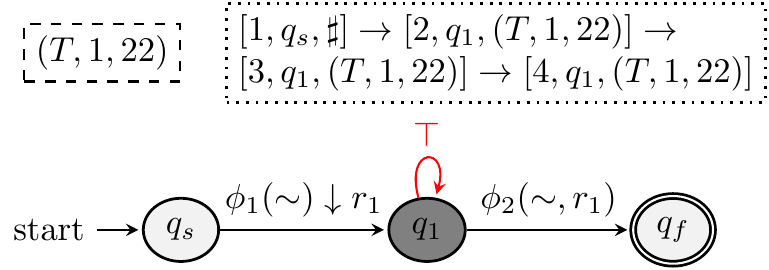}
	\caption{Configuration after reading $t_{3}$.}
\end{subfigure}\\
\begin{subfigure}[t]{0.5\textwidth}
	\includegraphics[width=0.99\textwidth]{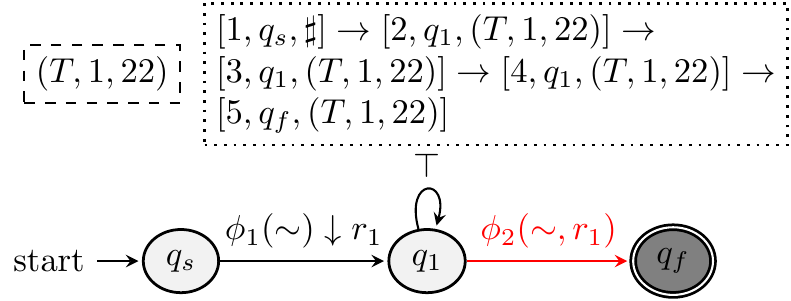}
	\caption{Configuration after reading $t_{4}$.}
\end{subfigure}
\caption{A run of the \sra\ of Figure \ref{fig:example1}, while consuming the first four events from the stream of Table \ref{table:example_stream}. Triggered transitions are shown in red and the current state of the \sra\ in dark gray.
The dashed box represents a register.
The contents of the register at each configuration are shown inside the dashed box.
Inside the dotted boxes, the run is shown.}
\label{fig:example_run}
\end{figure}
$\diamond$
\end{example}

Finally,
we can define the language of a \sra\ as the set of strings for which the \sra\ has an accepting run,
starting from an empty configuration.
\begin{definition}[Language recognized by \sra]
\label{definition:sra_language}
We say that a \sra\ $A$ accepts a string $S$ iff there exists an accepting run $\varrho=[1,q_{1},v_{1}] \overset{\delta_{1}}{\rightarrow} [2,q_{2},v_{2}] \overset{\delta_{2}}{\rightarrow} \cdots \overset{\delta_{n}}{\rightarrow} [n+1,q_{n+1},v_{n+1}]$ of $A$ over $S$,
where $q_{1} = A.q_{s}$ and $v_{1} = \sharp$.
The set of all strings accepted by $A$ is called the language recognized by $A$ and is denoted by $\mathcal{L}(A)$.
$\blacktriangleleft$
\end{definition}

%% file: sra_properties.tex
\section{Properties of Symbolic Register Automata}
\label{sec:sra_properties}
 
We now study the properties of \sra.
First, we prove the equivalence of \sra\ and \srem.
We then show that \sra\ and \srem\ are closed under union, intersection, concatenation and Kleene-start but not under complement and determinization.
We can thus construct \srem\ and \sra\ by using arbitrarily (in whatever order and depth is required) the four basic operators of union, intersection, concatenation and Kleene-star.
However, 
the negative result about complement suggests that the use of \emph{negation} in CER patterns cannot be equally arbitrary.
Moreover, 
deterministic \sra\ cannot be used in cases where this might be required,
as in CEF.
We will discuss in Section \ref{sec:markov} why determinization is important for CEF.
If, however, we use an extra window operator,
effectively limiting the length of strings accepted by a \sra,
we can then show that closure under complement and determinization is also possible. 

\subsection{Equivalence of \srem\ and \sra}

We first prove that,
for every \srem\ there exists an equivalent \sra.
The proof is constructive,
similar to that for classical automata.
For the inverse direction,
i.e. converting a \sra\ to an equivalent \srem,
we use the notion of generalized \sra.
These are \sra\ which have complete \srem\ on their transitions.
By incrementally removing states from the \sra,
we are finally left with two states and the \srem\ which connects them is the \srem\ we are looking for.

We now show how,
for each \srem\,
we can construct an equivalent \sra.
Equivalence between an expression $e$ and a \sra\ $A$ means that they recognize the same language,
i.e., $\mathcal{L}(e) = \mathcal{L}(A)$.
See Definitions \ref{definition:language_srem} and \ref{definition:sra_language}.
\begin{theorem}
\label{theorem:srem2sra}
For every \srem\ $e$ there exists an equivalent \sra\ $A$, i.e., a \sra\ such that $\mathcal{L}(e) = \mathcal{L}(A)$.
\end{theorem}
\begin{proof}
The complete \sra\ construction process and proof may be found in Appendix \ref{sec:proof:srem2sra}.
\end{proof}

\begin{figure}[!ht]
\centering
\begin{subfigure}[t]{0.76\textwidth}
	\includegraphics[width=0.99\textwidth]{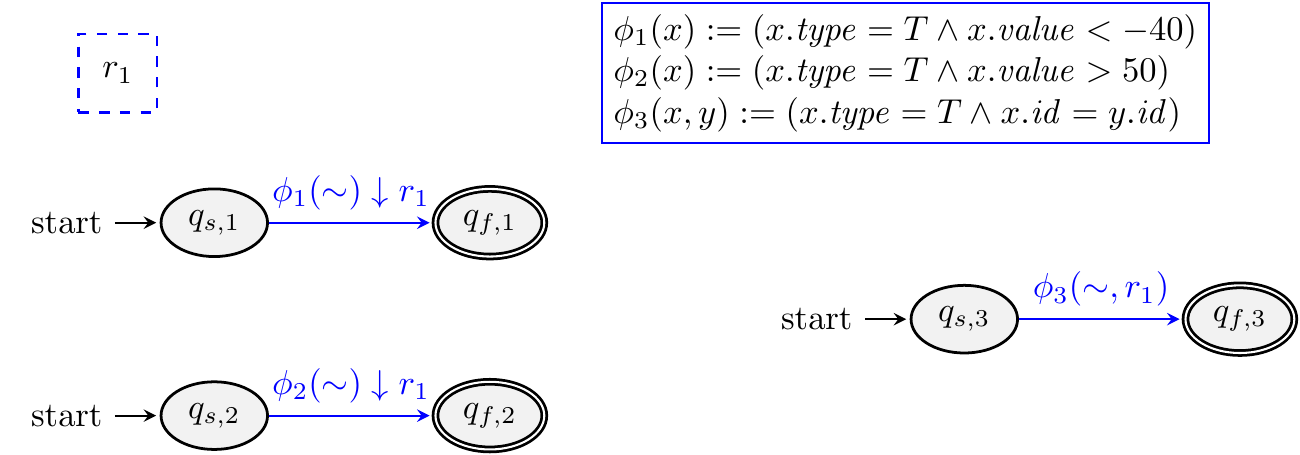}
	\caption{Constructing \sra\ for terminal sub-expressions.}
	\label{fig:srem2sra:example1}
\end{subfigure}\\
\begin{subfigure}[t]{0.76\textwidth}
	\includegraphics[width=0.99\textwidth]{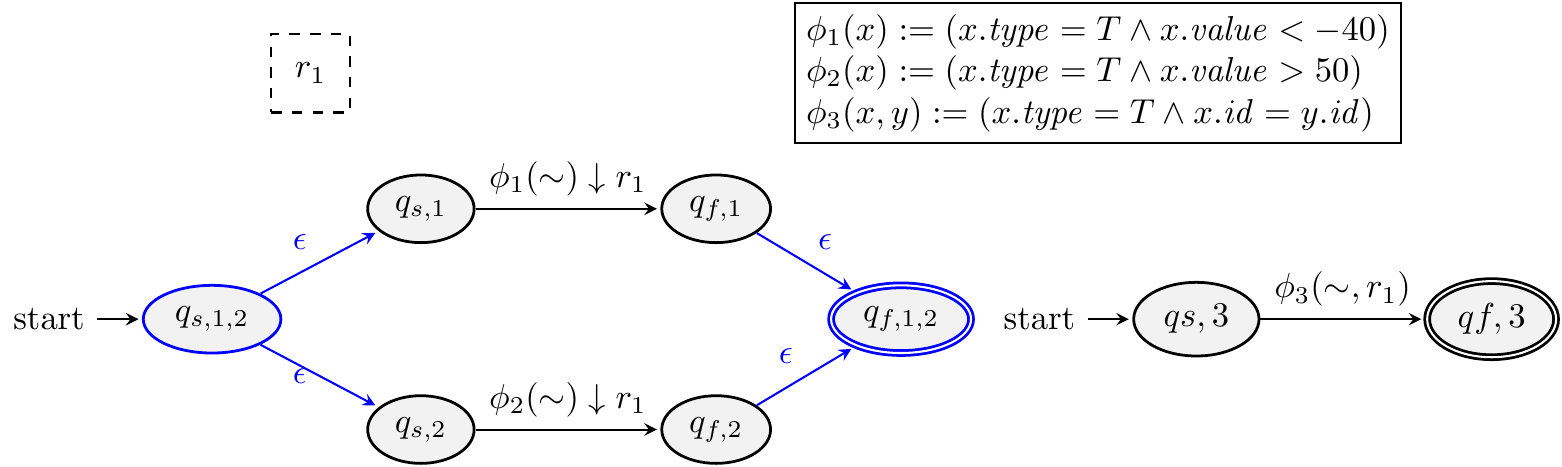}
	\caption{Connecting \sra\ via disjunction.}
	\label{fig:srem2sra:example2}
\end{subfigure}\\
\begin{subfigure}[t]{0.76\textwidth}
	\includegraphics[width=0.99\textwidth]{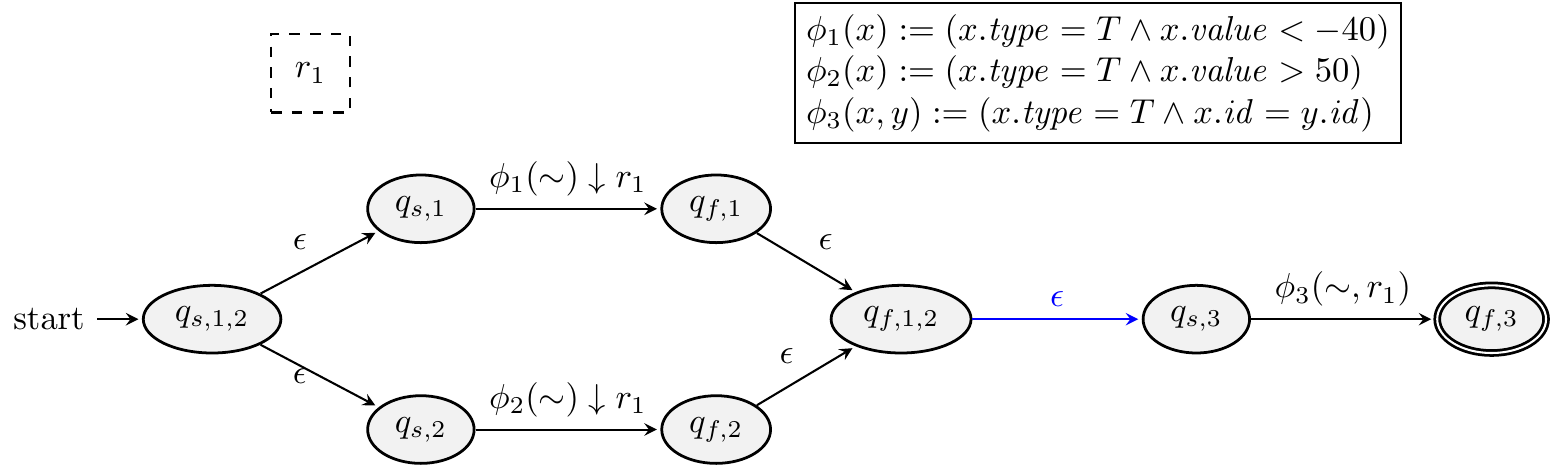}
	\caption{Connecting \sra\ via concatenation.}
	\label{fig:srem2sra:example3}
\end{subfigure}
\caption{Constructing \sra\ from \srem\ \eqref{srem:srem2sra_example}. New elements added at every step are shown in blue.}
\label{fig:srem2sra:example}
\end{figure}

\begin{example}
Here, 
we present an example, 
to give the intuition.
Let 
\begin{equation}
\label{srem:srem2sra_example}
\begin{aligned}
e_{2} := & ((\phi_{1}(\sim) \downarrow r_{1})\ + (\phi_{2}(\sim) \downarrow r_{1})) \cdot \\
		& (\phi_{3}(\sim,r_{1}))
\end{aligned}
\end{equation}
be a \srem,
where 
\begin{equation*}
\label{srem:srem2sra_example_conditions}
\begin{aligned}
\phi_{1}(x) := & (x.\mathit{type}=T \wedge x.\mathit{value} < -40) \\
\phi_{2}(x) :=	& (x.\mathit{type}=T \wedge x.\mathit{value} > 50) \\
\phi_{3}(x,y) :=	& (x.\mathit{type}=T \wedge x.\mathit{id} = y.\mathit{id})
\end{aligned}
\end{equation*}
With this expression,
we want to monitor sensors for possible failures.
We want to detect cases where a sensor records temperatures outside some range of values
(first line of \srem\ \eqref{srem:srem2sra_example})
and continues to transmit measurements (second line),
so that we are alerted to the fact that new measurements might not be trustworthy.
The last condition is a binary formula,
applied to both $\sim$ and $r_{1}$.
Figure \ref{fig:srem2sra:example} shows the process for constructing the \sra\ which is equivalent to \srem\ \eqref{srem:srem2sra_example}.

The algorithm is compositional,
starting from the base cases $e {:=} \phi$ or $e {:=} \phi \downarrow W$.
The three regular expression operators (concatenation, disjunction, Kleene-star) are handled in a manner almost identical as for classical automata.
The subtlety here concerns the handling of registers.
The simplest solution is to gather from the very start all registers mentioned in any sub-expressions of the original \srem\ $e$,
i.e. any registers in the register selection of any transitions and any write registers.
We first create those registers and then start the construction of the sub-automata.
Note that some registers may be mentioned in multiple sub-expressions (e.g., in one that writes to it and then in one that reads its contents).
We only add such registers once.
We treat the registers as a set with no repetitions.

For the example of Figure \ref{fig:srem2sra:example},
only one register is mentioned, 
$r_{1}$.
We start by creating this register.
Then, we move on to the terminal sub-expressions.
There are three basic sub-expressions and three basic automata are constructed:
from $q_{s,1}$ to $q_{f,1}$,
from $q_{s,2}$ to $q_{f,2}$ and
from $q_{s,3}$ to $q_{f,3}$.
See Figure \ref{fig:srem2sra:example1}.
To the first two transitions,
we add the relevant \emph{unary} conditions,
e.g.,
we add $\phi_{1}(x){:=} (x.\mathit{type}{=}T\ {\wedge}\ x.\mathit{value}{<}-40)$ to $q_{s,1} {\rightarrow} q_{f,1}$.
To the third transition,
we add the relevant \emph{binary} condition
$\phi_{3}(x,y) := (x.\mathit{type}=T \wedge x.\mathit{id} = y.\mathit{id})$.
The $+$ operator is handled by joining the \sra\ of the disjuncts through new states and $\epsilon$-transitions.
See Figure \ref{fig:srem2sra:example2}. 
The concatenation operator is handled by connecting the \sra\ of its sub-expressions through an $\epsilon$-transition,
without adding any new states.
See Figure \ref{fig:srem2sra:example3}. 
Iteration,
not applicable in this example,
is handled by joining the final state of the original automaton to its start state through an $\epsilon$-transition. $\diamond$
\end{example}

We can also prove the inverse theorem,
i.e., that every \sra\ can be converted to a \srem.
To do so, however, we will need two lemmas.
The first is the standard lemma about $\epsilon$ elimination,
stating that we can always eliminate all $\epsilon$ transitions from a \sra\ to get an equivalent \sra\ with no $\epsilon$ transitions.
\begin{lemma}
\label{lemma:epsilon}
For every \sra\ $A_{\epsilon}$ with $\epsilon$ transitions there exists an equivalent  \sra\ $A_{\notin}$ without $\epsilon$ transitions, i.e., a \sra\ such that $\mathcal{L}(A_{\epsilon}) = \mathcal{L}(A_{\notin})$.
\end{lemma}
\begin{proof}
See Appendix \ref{sec:proof:epsilon}.
\end{proof}

The next lemma that we will require concerns the ability of \sra\ to write to multiple registers at the same time. 
The write registers of a transition $\delta$ in Definition \ref{definition:sra},
$\delta.W$,
might not be a singleton.
On the other hand,
according to Definition \ref{definition:srem},
each terminal sub-expression in a \srem\ may write to at most one register.
We can prove though that being able to write to multiple registers at the same time does not add any expressive power to \sra.
Every \sra\ which can write to multiple registers can be converted to a \sra\ whose transitions can write to at most one register.
\begin{definition}
A \sra\ $A$ is called a multi-register \sra\ if there exists a transition $\delta \in A.\Delta$ such that $\lvert \delta.W \rvert > 1$, 
i.e., if there exists a transition that can write to multiple registers.
A \sra\ $A$ is called a single-register \sra\ if for all transitions $\delta \in A.\Delta$ it holds that $\lvert \delta.W \rvert \leq 1$, 
i.e., if each transition can write to at most one register.
$\blacktriangleleft$
\end{definition}

\begin{lemma}
\label{lemma:multi2single}
For every multi-register \sra\ $A_{mr}$ there exists an equivalent single-register \sra\ $A_{sr}$, i.e., a single-register \sra\ such that $\mathcal{L}(A_{mr}) = \mathcal{L}(A_{sr})$.
\end{lemma}
\begin{proof}
See Appendix \ref{sec:proof:multi2single}.
\end{proof}

We are now in a position to prove that every \sra\ can be converted to a \srem.
\begin{theorem}
\label{theorem:sra2srem}
For every \sra\ $A$ there exists an equivalent \srem\ $e$, i.e., a \srem\ such that $\mathcal{L}(A) = \mathcal{L}(e)$.
\end{theorem}
\begin{proof}
See Appendix \ref{sec:proof:sra2srem}.
\end{proof}

\subsection{Closure Properties of SREM/SRA}

We now study the closure properties of \sra\ under union, intersection, concatenation, Kleene-star, complement and determinization.
We first provide the definition for deterministic \sra.
Informally, 
a \sra\ is said to be deterministic if, 
at any time, 
with the same input event, 
it can follow no more than one transition.
The formal definition is as follows:
\begin{definition}[Deterministic \sra\ (\dsra)]
A \sra\ $A$ with $k$ registers $\{r_{1}, \cdots, r_{2}\}$ over a $\mathcal{V}$-structure $\mathcal{M}$ is deterministic if, 
for all transitions $q,\phi_{1} \downarrow W_{1} \rightarrow q_{1} \in A.\Delta $ and $q,\phi_{2} \downarrow W_{2} \rightarrow q_{2} \in A.\Delta$,
if $q_{1} \neq q_{2}$ then,
for all $u \in \mathcal{M}.\mathcal{U}$ and $v \in F(r_{1}, \cdots, r_{2})$,
$(u,v) \models \phi_{1}$ and  $(u,v) \models \phi_{2}$ cannot both hold,
i.e.,
\begin{itemize}
	\item Either $(u,v) \models \phi_{1}$ and $(u,v) \nvDash \phi_{2}$
	\item or $(u,v) \nvDash \phi_{1}$ and $(u,v) \models \phi_{2}$
	\item or $(u,v) \nvDash \phi_{1}$ and $(u,v) \nvDash \phi_{2}$.
\end{itemize}
$\blacktriangleleft$
\end{definition}
In other words,
from all the outgoing transitions from a given state $q$ at most one of them can be triggered on any element $u$ and valuation/register contents $v$.
By definition,
for a deterministic \sra,
at most one run may exist for every string/stream.

We now give the definition for closure under union, intersection, concatenation, Kleene-star, complement and determinization:
\begin{definition}[Closure of \sra]
\label{definition:closure}
We say that \sra\ are closed under: 
\begin{itemize}
	\item union if, for every \sra\ $A_{1}$ and $A_{2}$,
	there exists a \sra\ $A$ such that $\mathcal{L(A)} = \mathcal{L}(A_{1}) \cup \mathcal{L}(A_{2})$, i.e., a string $S$ is accepted by $A$ iff it is accepted either by $A_{1}$ or by $A_{2}$.
	\item intersection if, for every \sra\ $A_{1}$ and $A_{2}$,
	there exists a \sra\ $A$ such that $\mathcal{L(A)} = \mathcal{L}(A_{1}) \cap \mathcal{L}(A_{2})$, i.e., a string $S$ is accepted by $A$ iff it is accepted by both $A_{1}$ and $A_{2}$.
	\item concatenation if, for every \sra\ $A_{1}$ and $A_{2}$,
	there exists a \sra\ $A$ such that $\mathcal{L}(A) = \mathcal{L}(A_{1}) \cdot \mathcal{L}(A_{2})$, i.e., $S$ is accepted by $A$ iff it can be broken into two sub-strings $S = S_{1} \cdot S_{2}$ such that $S_{1}$ is accepted by $A_{1}$ and $S_{2}$ by $A_{2}$.
	\item Kleene-star if, for every \sra\ $A$,
	there exists a \sra\ $A_{*}$ such that $\mathcal{L}(A_{*}) = (\mathcal{L}(A))^{*}$,
	where $L^{*} = \bigcup\limits_{i \geq 0}{L^{i}}$, i.e., $S$ is accepted by $A_{*}$ iff it can be broken into $S = S_{1} \cdot S_{2} \cdot \cdots$ such that each $S_{i}$ is accepted by $A$.
	\item complement if, for every \sra\ $A$,
	there exists a \sra\ $A_{c}$ such that for every string $S$ it holds that $S \in \mathcal{L}(A) \Leftrightarrow S \notin \mathcal{L}(A_{c})$.
	\item determinization if, for every \sra\ $A$,
	there exists a \dsra\ $A_{D}$ such that $\mathcal{L}(A) = \mathcal{L}(A_{d})$.
\end{itemize}
$\blacktriangleleft$
\end{definition}

We thus have the following for union, intersection, concatenation and Kleene-star:
\begin{theorem}
\label{theorem:closure}
\sra\ and \srem\ are closed under union, intersection, concatenation and Kleene-star.
\end{theorem}
\begin{proof}
See Appendix \ref{sec:proof:closure} for complete proofs.
For union, concatenation and Kleene-star the proof is essentially the proof for converting \srem\ to \sra\ (and we have already proven that \sra\ and \srem\ are equivalent).
For intersection,
we construct a new \sra\ $A$ with $Q = A_{1}.Q \times A_{2}.Q$.
Then,
for each $q = (q_{1},q_{2}) \in Q$ we add a transition $\delta$ to $q' = (q_{1}',q_{2}') \in Q$ if there exists a transition $\delta_{1}$ from $q_{1}$ to $q_{1}'$ in $A_{1}$ and a transition $\delta_{2}$ from $q_{2}$ to $q_{2}'$ in $A_{2}$.
The write registers of $\delta$ are $W = \delta_{1}.W \cup \delta_{2}.W$,
i.e., we use multi-register \sra.
This new \sra\ can reach a final state only if both $A_{1}$ and $A_{2}$ reach their final states on a given string $S$. 
\end{proof}

On the other hand,
\sra\ are not closed under complement:
\begin{theorem}
\label{theorem:complement}
\sra\ and \srem\ are not closed under complement.
\end{theorem}
\begin{proof}
See Appendix \ref{sec:proof:complement}.
\end{proof}

It is also not always possible to determinize them:
\begin{theorem}
\label{theorem:determinization}
\sra\ are not closed under determinization.
\end{theorem}
\begin{proof}
See Appendix \ref{sec:proof:determinization}.
\end{proof}

\sra\ can thus be constructed from four basic operators (union, intersection, concatenation and Kleene-star) in a compositional manner,
providing substantial flexibility and expressive power for CER applications.
However, 
as is the case for register automata \cite{DBLP:journals/tcs/KaminskiF94}, 
\sra\ are not closed under complement,
something which could pose difficulties for handling \emph{negation},
i.e.,
the ability to state that a sub-pattern should not happen for the whole pattern to be detected.

\sra\ are also not closed under determinization,
a result which might seem discouraging.
In the next section,
we show that there exists a sub-class of \srem\ for which a translation to deterministic \sra\ is indeed possible.
This is achieved if we apply a windowing operator and limit the length of strings accepted by \srem\ and \sra.

\subsection{Windowed SREM/SRA}

We can overcome the negative results about complement and determinization
by using windows in \srem\ and \sra.
In general,
CER systems are not expected to remember every past event of a stream
and produce matches involving events that are very distant.
On the contrary,
it is usually the case that CER patterns include an operator that limits the search space of input events,
through the notion of windowing.
This observation motivates the introduction of windowing in \srem.
\begin{definition}[Windowed \srem]
\label{definition:windowed_srem}
Let $e$ be a \srem\ over a $\mathcal{V}$-structure $\mathcal{M}$ and a set of register variables $R = \{r_{1}, \cdots, r_{k}\}$, 
$S$ a string constructed from elements of the universe of $\mathcal{M}$ 
and $v,v' \in F(r_{1}, \cdots, r_{k})$.
A windowed \srem\ (\wsrem) is an expression of the form $e' := e^{[1..w]}$,
where $w \in \mathbb{N}_{1}$. 
We define the relation $(e',S,v) \vdash v'$ as follows:
$(e,S,v) \vdash v'$ and $\lvert S \rvert \leq w$.
$\blacktriangleleft$
\end{definition}

The windowing operator does not add any expressive power to \srem.
We could use the index of an event in the stream as an event attribute
and then add binary conditions in an expression which ensure that the difference between the index of the last event read and the first is no greater that $w$.
It is more convenient, however, to have an explicit operator for windowing.

We first show how we can construct a so-called ``unrolled \sra''\ from a windowed expression:
\begin{lemma}
\label{lemma:windowed_srem}
For every windowed \srem\ there exists an equivalent unrolled \sra\ without any loops, 
i.e., a \sra\ where each state may be visited at most once.
\end{lemma}
\begin{proof}
\input{algo_unrolling_simplified}
The full proof and the complete construction algorithm are presented in Appendix \ref{sec:proof:windowed_srem}.
\end{proof}
\begin{example}
Here, we provide only the general outline of the algorithm and an example.
Consider, e.g., the following \srem:
\begin{equation}
\label{srem:t_seq_h_filter_eq_id_w}
e_{3} := ( (\top)^{*} \cdot \mathit{TypeIsT}(\sim) \downarrow r_{1}) \cdot (\top)^{*} \cdot (\mathit{TypeIsH}(\sim) \wedge \mathit{EqualId}(\sim,r_{1})))^{[1..w]}
\end{equation}
It can skip any number of events with the first sub-expression $(\top)^{*}$.
Then it expects to find an event with type $T$ and stores it to register $r_{1}$ (sub-expression $\mathit{TypeIsT}(\sim) \downarrow r_{1}$).
The third sub-expression is again $(\top)^{*}$,
meaning that, 
after seeing a $T$ event, 
we are allowed to skip events.
Finally,
with the last sub-expression ($\mathit{TypeIsH}(\sim) \wedge \mathit{EqualId}(\sim,r_{1})$),
if some event after the $T$ event is of type $H$ and they have the same identifier,
then the string is accepted,
provided that its length is also at most $w$. 
Figure \ref{fig:det:example_unrolled} shows the steps taken for constructing the equivalent unrolled \sra\ for this expression.
A simplified version of the unrolling algorithm is shown in Algorithm \ref{algorithm:unrolling_simplified}.

\begin{figure}[!ht]
    \centering
    \begin{subfigure}[b]{0.55\textwidth}
        \includegraphics[width=\textwidth]{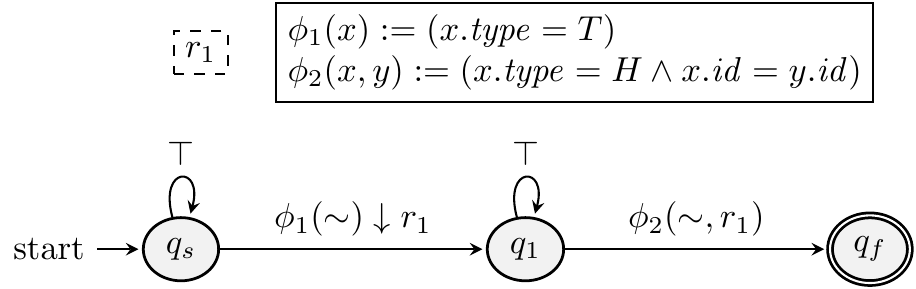}
        \caption{\sra\ for \srem\ \eqref{srem:t_seq_h_filter_eq_id_w} before unrolling.}\label{fig:det:initial}
    \end{subfigure}
    \begin{subfigure}[b]{0.75\textwidth}
        \includegraphics[width=\textwidth]{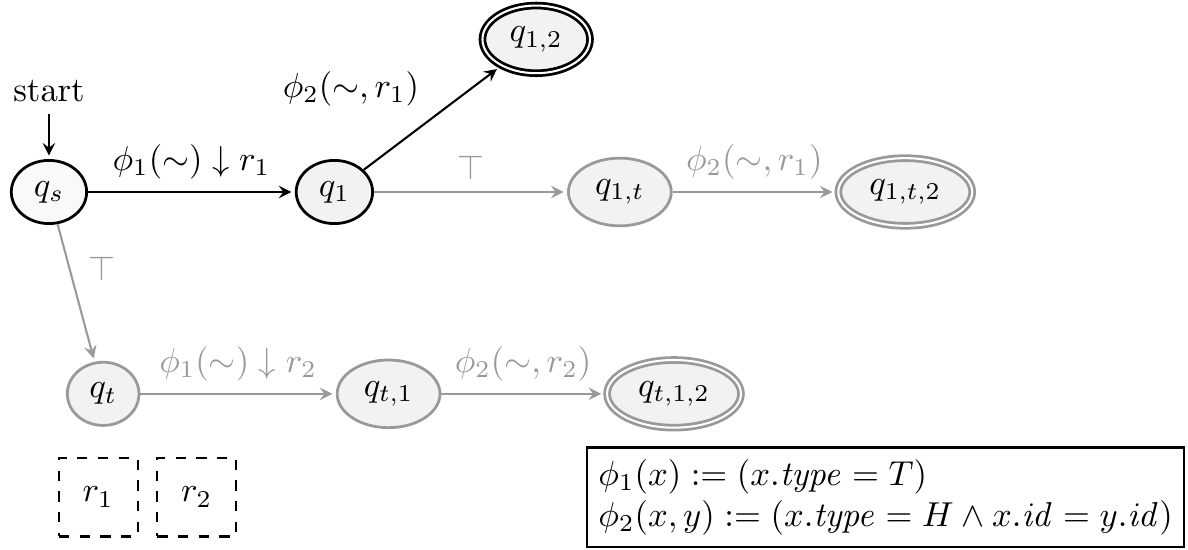}
        \caption{\sra\ for \srem\ \eqref{srem:t_seq_h_filter_eq_id_w} after unrolling cycles, for $w=3$ (whole \sra, black and light gray states) and $w=2$ (top 3 states in black).}\label{fig:det:example_unrolled}
    \end{subfigure}
    ~ 
    \begin{subfigure}[b]{0.75\textwidth}
        \includegraphics[width=\textwidth]{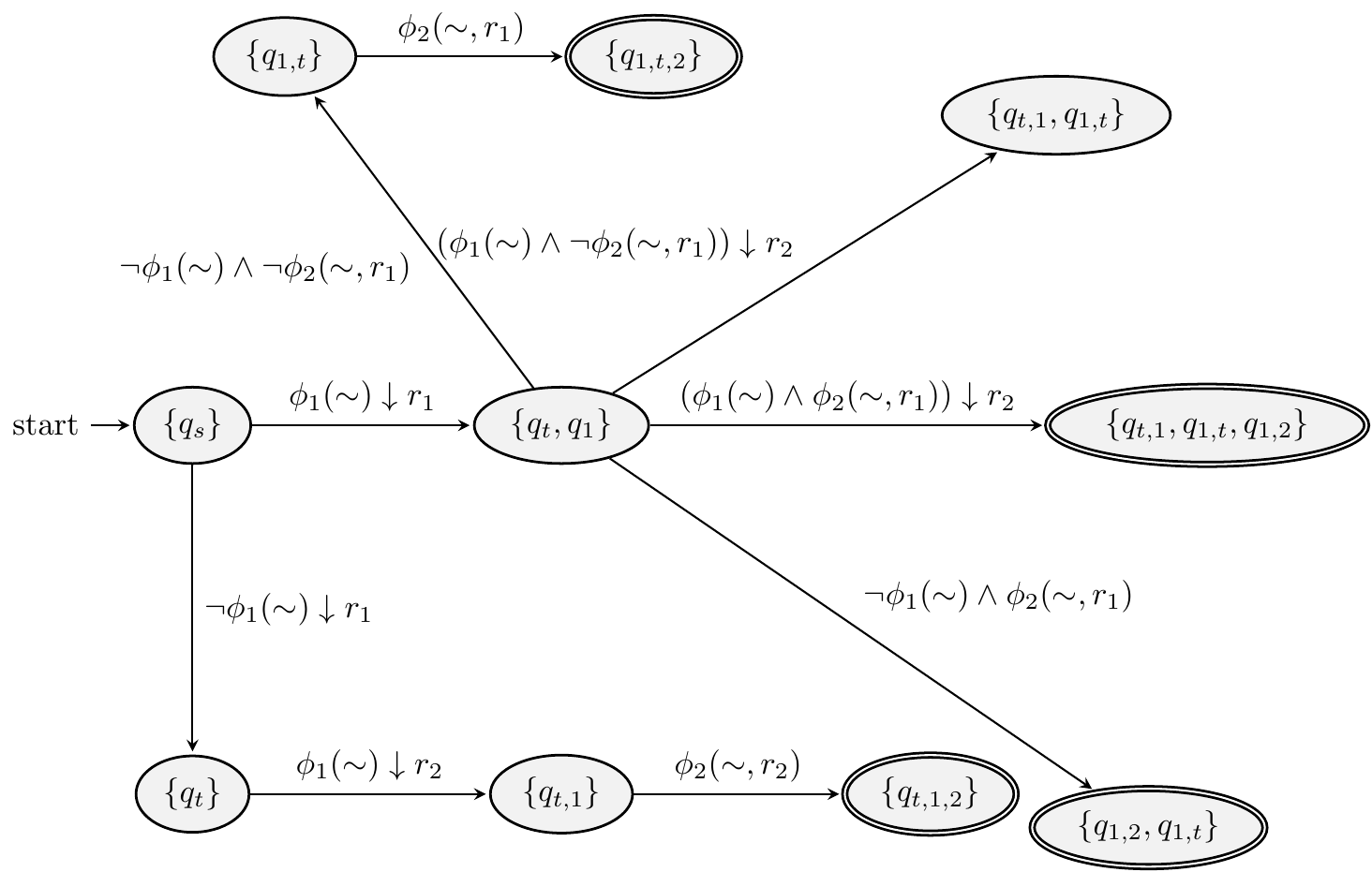}
        \caption{\dsra\ (only part of it), for $w=3$.}\label{fig:det:output_free}
    \end{subfigure}
    \caption{Constructing \dsra\ for \srem\ \eqref{srem:t_seq_h_filter_eq_id_w}.}\label{fig:det:example}
\end{figure}

The construction algorithm first produces a \sra\ as usual,
without taking the window operator into account
(see line \ref{algorithm:unrolling_simplified:line:sra} of Algorithm \ref{algorithm:unrolling_simplified}).
For our example,
the result would be the \sra\ of Figure \ref{fig:det:initial}
(please, note that the automaton of Figure \ref{fig:det:initial} is slightly different than that of Figure \ref{fig:example1} due to the presence of the first $(\top)^{*}$ sub-expression).
Then the algorithm eliminates any $\epsilon$-transitions (line \ref{algorithm:unrolling_simplified:line:eliminate}).
The next step is to use this \sra\ in order to create the equivalent unrolled \sra\ (\usra).
The rationale behind this step is that the window constraint essentially imposes an upper bound on the number of registers that would be required for a deterministic \sra.
For our example,
if $w{=}3$,
then we know that we will need at least one register,
if a $T$ event is immediately followed by an $H$ event.
We will also need at most two registers,
if two consecutive $T$ events appear before an $H$ event.
The function of the \usra\ is to create the number of registers that will be needed,
through traversing the original \sra. 
Algorithm \ref{algorithm:unrolling_simplified} does this by enumerating all the walks of length up to $w$
on the \sra\ graph,
by unrolling any cycles.
Lines \ref{algorithm:unrolling_simplified:line:unroll_start} -- \ref{algorithm:unrolling_simplified:line:unroll_end} of Algorithm \ref{algorithm:unrolling_simplified} show this process in a simplified manner.
The \usra\ for our example is shown in Figure \ref{fig:det:example_unrolled} for $w{=}3$.
The actual algorithm does not perform an exhaustive enumeration, 
but incrementally creates the \usra,
by using the initial \sra\ as a generator of walks.
Every time we expand a walk,
we add a new transition, 
a new state and possibly a new register,
as clones of the original transition, state and register.
In our example,
we start by creating a clone of $q_{s}$ in Figure \ref{fig:det:initial},
also named $q_{s}$ in Figure \ref{fig:det:example_unrolled}.
From the start state of the initial \sra\, 
we have two options.
Either loop in $q_{s}$ through the $\top$ transition
or move to $q_{1}$ through the transition with the $\phi_{1}$ condition.
We can thus expand $q_{s}$ of the \usra\ with two new transitions:
from $q_{s}$ to $q_{t}$ and from $q_{s}$ to $q_{1}$ in Figure \ref{fig:det:example_unrolled}.
We keep expanding the \sra\ this way until we reach final states and without exceeding $w$.
As a result, 
the final \usra\ has the form of a tree, 
whose walks and runs are of length up to $w$. $\diamond$
\end{example}

A \usra\ then allows us to capture windowed expressions.
Note though that the algorithm we presented above,
due to the unrolling operation,
can result in a combinatorial explosion of the number of states of the \dsra,
especially for large values of $w$.
Its purpose here was mainly to establish Lemma \ref{lemma:windowed_srem}.

Having a \usra\ makes it easy to subsequently construct a \dsra:
\begin{theorem}
\label{theorem:wsrem2dsra}
For every windowed \srem\ there exists an equivalent deterministic \sra.
\end{theorem}
\begin{proof}
The proof for determinization is presented in Appendix \ref{sec:proof:wsrem2dsra}.
It is constructive and the determinization algorithm is based on the powerset construction of the states of the non-deterministic \sra. 
It is similar to the algorithm for symbolic automata \cite{DBLP:conf/lpar/VeanesBM10,DBLP:conf/cav/DAntoniV17}.
It does not add or remove any registers.
It initially constructs the powerset of the states of the \usra. 
The members of this powerset will be the states of the \dsra.
It then tries to make each such new state, say $q_{d}$, deterministic,
by creating transitions with mutually exclusive conditions when they have the same output.
The construction of these mutually exclusive conditions is done by gathering the conditions of all the transitions that have as their source a member of $q_{d}$.
Out of these conditions,
the set of \emph{minterms} is created,
i.e.,
the mutually exclusive conjuncts constructed from the initial conditions,
where each conjunct is a condition in its original or its negated form.
A transition is then created for each minterm,
with $q_{d}$ being the source.
Then, only one transition can be triggered,
since these minterms are mutually exclusive.
\end{proof}

\begin{example}
As an example,
Figure \ref{fig:det:output_free} shows the result of converting the \usra\ of Figure \ref{fig:det:example_unrolled} to a \dsra.
We have simplified somewhat the conditions of each transition due to the presence of the $\top$ predicates in some of them.
For example, 
the minterm $\phi_{1} {\wedge} \neg \top$ for the start state is unsatisfiable and can be ignored while $\phi_{1} {\wedge} \top$ may be simplified to $\phi_{1}$.
The figure shows only part of the \dsra\ to avoid clutter.
Note that some of the rightmost states may be further expanded.
For example,
state $\{q_{t,1},q_{1,t}\}$ (top right) can be expanded.
With the minterm $\phi_{2}(\sim,r_{1}) \wedge \neg \phi_{2}(\sim,r_{2})$,
it would go to the final state $\{q_{1,t,2}\}$ 
(not shown in the figure). $\diamond$
\end{example}

Being able to derive a deterministic \sra\ is important for Complex Event Forecasting (CEF),
since,
as we will show,
determinization is an important intermediate step in this task.
A deterministic \sra\ essentially provides us with the ``symbols'' with which we populate a prediction suffix tree,
the structure that captures the statistical properties of an input event stream.

We may now prove,
as a corollary,
that windowed \sra\ are also closed under complement:
\begin{corollary}
\label{corollary:wsra_complement}
Windowed \sra\ are closed under complement.
\end{corollary}
\begin{proof}
See Appendix \ref{sec:proof:wsra_complement}.
\end{proof}

%% file: algo_unrolling_simplified.tex
\begin{algorithm}
\caption{Constructing unrolled \sra\ for windowed \srem\ (simplified).}
\label{algorithm:unrolling_simplified}
\SetAlgoNoLine
\KwIn{Windowed \srem\ $e' := e^{[1..w]}$}
\KwOut{Deterministic \sra\ $A_{e'}$ equivalent to $e'$}
$A_{e,\epsilon} \leftarrow \mathit{ConstructSRA}(e)$\; \label{algorithm:unrolling_simplified:line:sra} 
$A_{e} \leftarrow \mathit{EliminateEpsilon}(A_{e,\epsilon})$\; \label{algorithm:unrolling_simplified:line:eliminate} 
enumerate all walks of $A_{e}$ of length up to $w$; \label{algorithm:unrolling_simplified:line:unroll_start} \tcp{{\footnotesize Now unroll $A_{e}$.}}\
join walks through disjunction\;
\hspace{0.01cm} collapse common prefixes;\label{algorithm:unrolling_simplified:line:unroll_end}
\end{algorithm}

%% file: streaming_sra.tex
\section{Streaming \sra\ for Complex Event Recognition}
\label{sec:streaming_sra}

We have thus far described how \srem\ and \sra\ can be applied to bounded strings that are known in their totality before recognition.
A string is given to a \sra\ and an answer is expected about whether the whole string belongs to the automaton's language or not.
However, 
in CER/F we are required to handle continuously updated streams of events and detect instances of \srem\ satisfaction as soon as they appear in a stream. 
For example, 
the automaton of the classical regular expression $a \cdot b$ would accept only the string $a,b$.
In a streaming setting, 
we would like the automaton to report a match every time this string appears in a stream.
For the stream $a,b,c,a,b,c$, 
we would thus expect two matches to be reported,
one after the second symbol and one after the fifth
(assuming that we are interested only in contiguous matches).

Slight modifications are required so that \srem\ and \sra\ may work in a streaming setting
(the discussion in this section develops along the lines presented in our previous work \cite{DBLP:journals/vldb/AlevizosAP21}, with the difference that here we are concerned with symbolic automata with memory).
First, 
we need to make sure that the automaton can start its recognition after every new element.
If we have a classical regular expression $R$,
we can achieve this by applying on the stream the expression $\Sigma^{*} \cdot R$,
where $\Sigma$ is the automaton's (classical) alphabet.
For example,
if we apply $R := \{a,b,c\}^{*} \cdot (a \cdot b)$ on the stream $a,b,c,a,b,c$,
the corresponding automaton would indeed reach its final state after reading the second and the fifth symbols.
In our case, 
events come in the form of tuples with both numerical and categorical values. 
Using database systems terminology we can speak of tuples from relations of a database schema \cite{DBLP:conf/icdt/GrezRU19}.
These tuples constitute the universe $\mathcal{U}$ of a $\mathcal{V}$-structure $\mathcal{M}$.
A stream $S$ then has the form of an infinite sequence $S=t_{1},t_{2},\cdots$, where $t_{i} \in \mathcal{U}$.
Our goal is to report the indices $i$ at which a complex event is detected.

More precisely,
if $S_{1..k}=\cdots,t_{k-1},t_{k}$ is the prefix of $S$ up to the index $k$,
we say that an instance of a \srem\ $e$ is detected at $k$ iff there exists a suffix $S_{m..k}$ of $S_{1..k}$ such that $S_{m..k} \in \mathcal{L}(e)$.
In order to detect complex events of a \srem\ $e$ on a stream, 
we use a streaming version of \srem\ and \sra.

\begin{definition}[Streaming \srem\ and \sra]
If $e$ is a \srem, 
then $e_{s}= \top^{*} \cdot e$ is called the streaming \srem\ (\ssrem) corresponding to $e$.
A \sra\ $A_{e_{s}}$ constructed from $e_{s}$ is called a streaming \sra\ (\ssra) corresponding to $e$.
$\blacktriangleleft$
\end{definition}

Using $e_{s} = \top^{*} \cdot e$ we can detect complex events of $e$ while reading a stream $S$,
since a stream segment $S_{m..k}$ belongs to the language of $e$ iff the prefix $S_{1..k}$ belongs to the language of $e_{s}$.
The prefix $\top^{*}$ lets us skip any number of events from the stream and start recognition at any index $m, 1 \leq m \leq k$.

\begin{proposition}
\label{proposition:streamingsrem}
If $S=t_{1},t_{2},\cdots$ is a stream of elements from a universe $\mathcal{U}$ of a $\mathcal{V}$-structure $\mathcal{M}$, 
where $t_{i} \in \mathcal{U}$  
and $e$ is a \srem\ over $\mathcal{M}$,
then, 
for every $S_{m..k}$, $S_{m..k} \in \mathcal{L}(e)$ iff $S_{1..k} \in \mathcal{L}(e_{s})$ (and $S_{1..k} \in \mathcal{L}(A_{e_{s}})$).
\end{proposition}
\begin{proof}
See Appendix \ref{sec:proof:streamingsrem}.
\end{proof}

Note that \ssrem\ and \ssra\ are just special cases of \srem\ and \sra\ respectively.
Therefore, 
every result that holds for \srem\ and \sra\ also holds for \ssrem\ and \ssra\ as well.

%% file: complexity.tex
\section{Notes on Complexity}
\label{sec:complexity}

In the theory of formal languages it is customary to present complexity results for various decision problems, most commonly for the problem of non-emptiness (whether an expression or automaton accepts at least one string),
that of membership (deciding whether a given string belongs to the language of an expression/automaton) and that of universality (deciding whether a given expression/automaton accepts every possible string).
We briefly discuss here these problems for the case of \srem\ and \sra.

The complexity of these problems for \srem\ and \sra\ depends heavily on the nature of the conditions used as terminal expressions in \srem\ and as transition guards in \sra,
e.g., the $\mathit{EqualId}(\sim,r_{1})$ condition in \srem\ \eqref{srem:t_seq_h_filter_eq_id_w}.
This, in turn, depends on the complexity of deciding whether a given element from the universe $\mathcal{U}$ of a $\mathcal{V}$-structure $\mathcal{M}$ belongs to a relation $R$ from $\mathcal{M}$.
Since we have not imposed until now any restrictions on such relations,
the complexity of the aforementioned decision problems can be ``arbitrarily'' high and thus we cannot provide specific bounds.
If, for example, the problem of evaluating a relation $R$ is NP-complete and this relation is used in a \srem/\sra\ condition,
this then implies that the problem of membership immediately becomes at least NP-complete.
In fact, 
if the problem of deciding whether an element from $\mathcal{U}$ belongs to a relation $R$ is undecidable,
then the membership problem becomes also undecidable.

We can, however, provide some rough bounds by looking at the complexity of these problems for the case of register automata
(see \cite{DBLP:journals/jcss/LibkinTV15}).
Register automata are a special case of \sra,
where the only allowed relations are the binary relations of equality and inequality.
We assume that these relations may be evaluated in constant time.
For the problem of universality,
we know that it is undecidable for register automata.
We can thus infer that it remains so for \sra\ as well.
On the other hand,
the problem of non-emptiness is decidable but PSPACE-complete.
The same problem for \sra\ is thus PSPACE-complete.
Finally,
the problem of membership is NP-complete.
Therefore,
it is also at least NP-complete for \sra.
Note that membership is the most important problem for the purposes of CER/F,
since in CER/F we continuously try to check whether a string (a suffix of the input stream) belongs to the language of a pattern's automaton.
In general,
if we assume that the problem of membership in all relations $R$ is decidable in constant time,
then the complexity of the decision problems for \sra\ coincides with that for register automata.

If we focus our attention even further on windowed \sra,
as is the case in CER/F,
then we can estimate more precisely the complexity of processing a single event from a stream.
This is the most important operation for CER/F.
A windowed \sra\ can first be determinized (offline) to obtain a \dsra.
Assume that the resulting \dsra\ $A$ has $k$ registers and $c$ conditions/minterms.
We also assume that evaluating a condition requires constant time and that accessing a register also takes constant time.
In the worst case,
after a new element/event arrives,
we need to evaluate all of the conditions/minterms on the $c$ outgoing transitions of the current state to determine which one of them is triggered.
We may also need to access all of the $k$ registers in order to evaluate the conditions.
Therefore, 
the complexity of updating the state of the \dsra\ $A$ is $O(c+k)$
(assuming that each register is accessed only once and its contents are provided to every condition which references that register).

%% file: markov.tex
\section{Complex Event Forecasting with Markov Models}
\label{sec:markov}

We now show how we can use the framework of \srem\ and \sra\ to perform Complex Event Forecasting (CEF).
The main idea behind our forecasting method is the following:
Given a pattern $e$ in the form of a \srem, 
we first construct an automaton. 
In order to perform event forecasting, 
we translate the \sra\ to an equivalent deterministic \sra.
This \dsra\ can then be used to learn a probabilistic model, 
typically a Markov model, 
that encodes dependencies among the events in an input stream.
Note that deterministic \sra\ are important because they allow us,
as we will show,
to produce a stream of ``symbols'' from the initial stream of events.
By using deterministic \sra,
we can map each input event to a single symbol and then use this derived stream of symbols to learn a Markov model. 
With non-deterministic \sra\,
each element/event from the string/stream may trigger multiple transitions and thus such a mapping is not possible.
The probabilistic model is learned from a portion of the input stream which acts as a training dataset and it is then used to derive forecasts about the expected occurrence of the complex event encoded by the automaton.
After learning a model,
we need to estimate the so-called \emph{waiting-time distributions} for each state of our automaton. 
These distributions let us know the probability of reaching a final state from any other automaton state in $k$ events from now.
These distributions are then used to estimate forecasts,
which generally have the form of an interval within which a complex event has a high probability of occurring.

We discern three cases and present them in order of increasing complexity:
\begin{itemize}
	\item We have only unary conditions applied to the last event and an arbitrary (finite or infinite) universe. 
	In this case, 
	we do not need registers.
	\item We have $n$-ary conditions (with $n \geq 1$) and a finite universe. 
	In this case, 
	registers are helpful, 
	but may not be necessary.  
	If we have a register automaton $A$ and a finite universe $\mathcal{U}$,
	we can always create an automaton $A_{U}$ with states $A.Q \times \mathcal{U}$ and appropriate transitions so that $A_{U}$ is equivalent to $A$ but has no registers. 
	Its states can implicitly remember past elements.
	\item The most complex case is when we have $n$-ary conditions and an infinite universe, as is typically assumed in CER/F.
	Registers are necessary in this case. 
\end{itemize}

\subsection{\srem\ with unary conditions}
As a first step,
we assume that the given \srem\ contains only unary conditions.
The universe in this case may be either finite or infinite.
We have already presented how this case can be handled in our previous work \cite{DBLP:journals/vldb/AlevizosAP21}.
We will present here only a high-level overview of our method and then discuss how it can be adjusted in order to accommodate the other two cases.

Before discussing how a \dsra\ can be described by a Markov model,
we first discuss a useful result,
which bears on the importance of being able to use deterministic automata.
It can be shown that a \dsra\ always has an equivalent deterministic classical automaton,
through a simple isomorphic mapping,
retaining the exact same structure for the automaton and simply changing the conditions on the transitions with symbols
\cite{DBLP:books/daglib/0023547}.
This result is important for two reasons:
a) it allows us to use methods developed for classical automata without having to always prove that they are indeed applicable to symbolic automata as well, and
b) it will help us in simplifying our notation, 
since we can use the standard notation of symbols instead of predicates. 
This result implies that,
for every run 
$\varrho = [1,q_{1},v_{1}] \overset{\delta_{1}}{\rightarrow} [2,q_{2},v_{2}] \overset{\delta_{2}}{\rightarrow} \cdots \overset{\delta_{k}}{\rightarrow} [k+1,q_{k+1},v_{k+1}]$ 
followed by a \dsra\ $A_{s}$ by consuming a symbolic string (or stream of events) $S$,
the run that the equivalent classical automaton $A_{c}$ follows by consuming the induced string $S'$ is also
$\varrho' = [1,q_{1},v_{1}] \overset{\delta_{1}}{\rightarrow} [2,q_{2},v_{2}] \overset{\delta_{2}}{\rightarrow} \cdots \overset{\delta_{k}}{\rightarrow} [k+1,q_{k+1},v_{k+1}]$,
i.e., $A_{c}$ follows the same copied/renamed states and the same copied/relabeled transitions.
We can then use symbols and strings (lowercase letters to denote symbols), 
as in classical theories of automata,
bearing in mind that,
in our case,
each symbol always corresponds to a condition. 
Details may be found in \cite{DBLP:journals/vldb/AlevizosAP21}.

This equivalent deterministic classical automaton can be used to convert a string/stream of elements/events to a string/stream of symbols.
Since each element may trigger only a single transition,
the initial string of elements may be mapped to a string of symbols.
Each transition $\delta_{i}$ from the initial run $\varrho$ corresponds to a transition $\delta_{i}$ from run $\varrho'$.
Since each transition from $\varrho'$ corresponds to a single symbol,
we can map the whole stream of input events to a single string of symbols. 
This means that we can use techniques developed in the context of deterministic classical automata and apply them to our case.
One such class of techniques concerns the question of how we can build a probabilistic model that captures the statistical properties of the streams to be processed by an automaton. 
Such a model would allow us to make inferences about the automaton's expected behavior as it reads event streams.

We have proposed the use of a variable-order Markov model (VMM) \cite{DBLP:journals/jair/BegleiterEY04,DBLP:journals/ml/RonST96,DBLP:conf/nips/RonST93,DBLP:journals/tcom/ClearyW84,DBLP:journals/tit/WillemsST95}.
Compared to fixed-order Markov models,
VMMs allow us to increase their order $m$ (how many events they can remember) to higher values and thus capture longer-term dependencies,
which can lead to a better accuracy.

The idea behind VMMs is the following:
let $\Sigma$ denote an alphabet, 
$\sigma \in \Sigma$ a symbol from that alphabet and $s \in \Sigma^{m}$ a string of length $m$ of symbols from that alphabet.
The aim is to derive a predictor $\hat{P}$ from the training data such that the average log-loss on a test sequence $S_{1..k}$ is minimized.
The loss is
given by 
$l(\hat{P},S_{1..k}) = - \frac{1}{T} \sum_{i=1}^{k} log \hat{P}(t_{i} \mid t_{1} \cdots t_{i-1})$.
Minimizing the log-loss is equivalent to maximizing the likelihood $\hat{P}(S_{1..k})=\prod_{i=1}^{k}\hat{P}(t_{i} \mid t_{1} \dots t_{i-1})$.
The average log-loss may also be viewed as a measure of the average compression rate achieved on the test sequence \cite{DBLP:journals/jair/BegleiterEY04}.
The mean (or expected) log-loss ($-\boldsymbol{E}_{P}\{log \hat{P}(S_{1..k}) \}$) is minimized if the derived predictor $\hat{P}$ is indeed the actual distribution $P$ of the source emitting sequences.

For fixed-order Markov models, 
the predictor $\hat{P}$ is derived through the estimation of conditional distributions $\hat{P}(\sigma \mid s)$,
with $m$ constant and equal to the assumed order of the Markov model
($\sigma$ is a single symbol and $s$ a string of length $m$). 
VMMs, 
on the other hand, 
relax the assumption of $m$ being fixed.
The length of the ``context'' $s$ may vary,
up to a \emph{maximum} order $m$, 
according to the statistics of the training dataset.  
By looking deeper into the past only when it is statistically meaningful,
VMMs can capture both short- and long-term dependencies.

We use Prediction Suffix Trees (\pst),
as described in \cite{DBLP:journals/ml/RonST96,DBLP:conf/nips/RonST93},
as our VMM of choice.
Assuming that we have derived an initial predictor $\hat{P}$
(by scanning the training dataset and estimating various empirical conditional probabilities),
the learning algorithm in \cite{DBLP:journals/ml/RonST96} starts with a tree having only a single node,
corresponding to the empty string $\epsilon$.
Then, 
it decides whether to add a new context/node $s$ by checking whether it is ``meaningful enough'' to expand to $s$. 
This is achieved by checking whether there exists a significant difference between the conditional probability of a symbol $\sigma$ given $s$ and the same probability given the shorter context $\mathit{suffix}(s)$
($\mathit{suffix}(s)$ is the longest suffix of $s$ different than $s$).
A detailed description of how we use \pst\ to perform forecasting may be found in \cite{DBLP:journals/vldb/AlevizosAP21}.

Our goal is to use the \pst\ in order to to calculate the so-called waiting-time distribution for every state $q$ of the automaton $A$.
The waiting-time distribution is the distribution of the index $n$, 
given by the waiting-time variable
$W_{q}=inf\{n: Y_{0},Y_{1},...,Y_{n}\}$,
where $Y_{0} = q$,  
$Y_{i} \in A.Q \backslash A.Q_{f}$ for $i \neq n$ and $Y_{n} \in A.Q_{f}$.
Thus, waiting-time distributions give us the probability to reach a final state from a given state $q$ in $n$ transitions from now.

\begin{example}
We provide here the intuition through an example.
Figure \ref{fig:dfaab} shows an example of a deterministic automaton.
Note that we use symbols on the transitions,
which, as explained,
essentially correspond to conditions.
Figure \ref{fig:pstab} shows a \pst\ which could be constructed from the automaton of Figure \ref{fig:dfaab} and a given training dataset,
with $m=3$. 
This \pst\ is read as follows.
Consider its left-most node, $aa,(0.75,0.25)$.
This means that the probability of encountering an $a$ symbol,
given that the last two symbols are $aa$,
is $0.75$.
The probability of seeing $b$,
on the other hand,
is $0.25$.
The order of this node is 2.
It has not been further expanded to yet another deeper level,
because it was estimated that such an expansion would be statistically insignificant.
For example,
the probability $P(a \mid baa)$ might still be very close to $0.75$ (e.g., $0.747$).
If the same is true for $P(a \mid aaa)$,
then this means that the probability of seeing $a$ is not significantly affected by expanding to contexts of length 3.
If a similar statistical insignificance can be established for the probability of $b$,
then it does not make sense to expand the node,
since its children would not provide us with more information.

\begin{figure}
\centering
\begin{subfigure}[t]{0.4\textwidth}
	\includegraphics[width=0.99\textwidth]{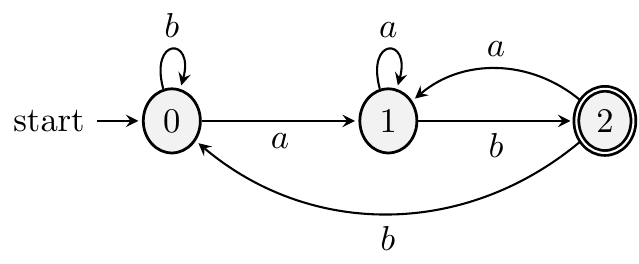}
	\caption{Example automaton $A$.}
	\label{fig:dfaab}
\end{subfigure}
\begin{subfigure}[t]{0.55\textwidth}
	\includegraphics[width=0.99\textwidth]{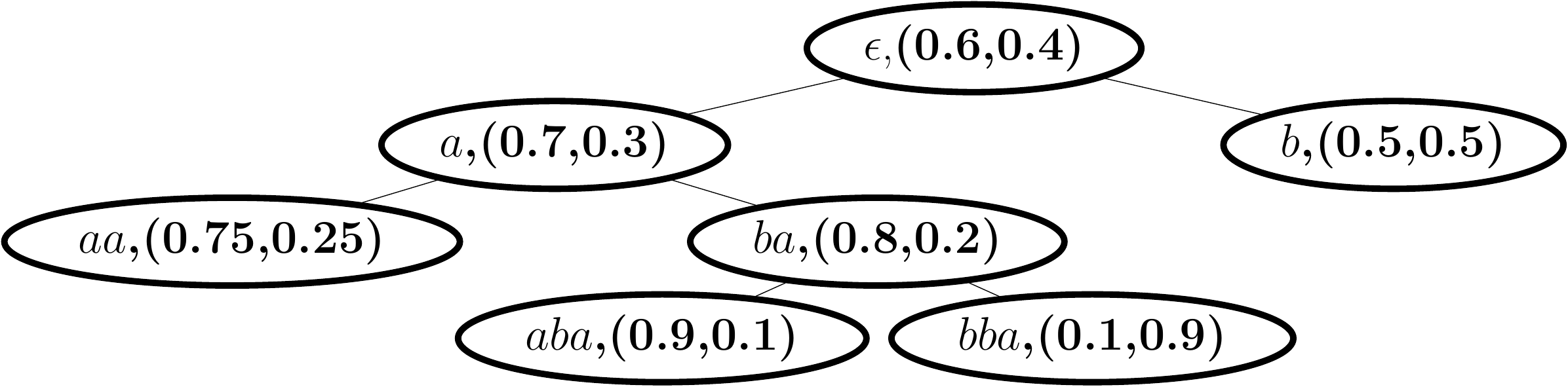}
	\caption{Example \pst\ $T$ for the automaton $A$.}
	\label{fig:pstab}
\end{subfigure}
\begin{subfigure}[t]{0.5\textwidth}
	\includegraphics[width=0.99\textwidth]{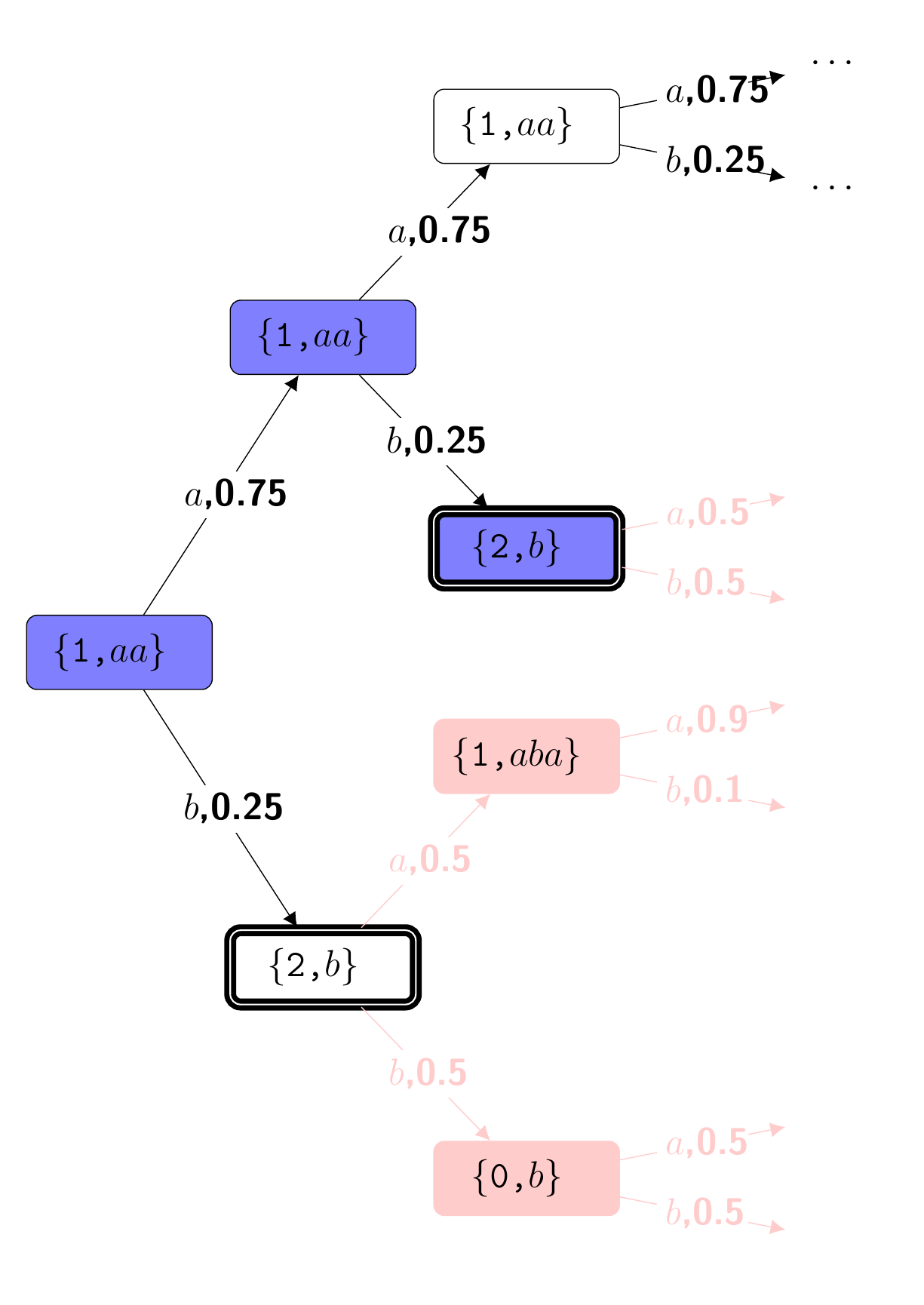}
	\caption{Future paths followed by automaton $A$ and \pst\ $T$ starting from state $1$ of $A$ and node $aa$ of $T$. Purple nodes correspond to the only path of length $k=2$ that leads to a final state. Pink nodes are pruned. Nodes with double borders correspond to final states of $A$.}
	\label{fig:future}
\end{subfigure}
\caption{Example of estimating waiting-time distributions.}
\label{fig:nomc}
\end{figure}

Figure \ref{fig:future} illustrates how we can estimate the probability for any future sequence of states of the \dsra\ $A$ of Figure \ref{fig:dfaab}, 
using the distributions of the \pst\ $T$ of Figure \ref{fig:pstab}
and thus how we can calculate the waiting-time distributions.
We first assume that,
as the system processes events from the input stream,
besides feeding them to $A$,
it also stores them in a buffer that holds the $m$ most recent events,
where $m$ is equal to the maximum order of the \pst\ $T$.
After updating the buffer with a new event,
the system traverses $T$ according to the contents of the buffer and arrives at a leaf $l$ of $T$.
Let us now assume that,
after consuming the last event, 
$A$ is in state $1$ in Figure \ref{fig:dfaab} and $T$ has reached its left-most node, $aa,(0.75,0.25)$ in Figure \ref{fig:pstab}.
This is shown as the left-most node also in Figure \ref{fig:future}.
Each node in this figure has two elements:
the first one is the state of $A$ and the second the node of $T$,
starting with $\{1,aa\}$ as our current ``configuration''.
Each node has two outgoing edges, one for $a$ and one for $b$,
indicating what might happen next and with what probability.
For example,
from the left-most node of Figure \ref{fig:future},
we know that, 
according to $T$, 
we might see $a$ with probability $0.75$ and $b$ with probability $0.25$.
If we do encounter $b$, 
then $A$ will move to state 2 and $T$ will reach leaf $b,(0.5,0.5)$.
This is shown in Figure \ref{fig:future} as the white node $\{2,b\}$.
This node has a double border to indicate that $A$ has reached a final state.

In a similar manner,
we can keep expanding this tree into the future
and use it to estimate the waiting-time distribution for its node $\{1,aa\}$,
i.e., the distribution for state $1$ of the automaton of Figure \ref{fig:dfaab} when we know that the last two read symbols are $aa$.
In order to estimate the probability of reaching a final state for the first time in $k$ transitions,
we first find all the paths of length $k$ which start from the original node 
and end in a final state without including another final state.
In our example of Figure \ref{fig:future},
if $k=1$,
then the path from $\{1,aa\}$ to $\{2,b\}$ is such a path and its probability is $0.25$.
Thus, $P(W_{\{1,aa\}}=1)=0.25$.
For $k=2$,
the path with the purple nodes leads to a final state after 2 transitions.
Its probability is $0.75*0.25=0.1875$,
i.e., the product of the probabilities on the path edges.
Thus, $P(W_{\{1,aa\}}=2)=0.1875$.
If there were more such alternative paths,
we would have to add their probabilities. $\diamond$
\end{example}

We can use the waiting-time distributions to produce various kinds of forecasts.
In the simplest case,
we can perform regression forecasting where we select the future point with the highest probability and return this point as a forecast.
Alternatively, 
we may also perform classification forecasting,
if our goal is to determine how likely it is that a CE will occur within the next $w$ input events.
We can sum the probabilities of the first $w$ points of a distribution
and if this sum exceeds a given threshold
we emit a ``positive'' forecast,
meaning that a CE is indeed expected to occur;
otherwise a ``negative'' forecast is emitted,
meaning that no CE is expected.

\subsection{\srem\ with n-ary conditions on a finite universe}

Thus far,
we have described how we can perform CEF when we only have unary conditions and a finite or infinite universe.
In this case,
we create the deterministic automaton and use it 
to generate a stream of symbols with which we can learn a \pst.
Note that,
when we only have unary conditions and thus no need for registers,
\sra\ are in essence equivalent to symbolic automata.
Symbolic automata are determinizable and closed under complement,
without requiring a window.
Thus,
the method described above for forecasting applies to every \sra\ with unary conditions,
regardless of whether a windowing operator is present.

The next case is when we have a finite universe and $n$-ary conditions,
where $n \geq 1$.
We can follow the same process as described above,
with one important difference.
Since the universe $\mathcal{U}$ is finite,
we can directly map each element of $\mathcal{U}$ to a symbol. 
Therefore,
the \pst\ $T$ can be constructed directly from the elements of $\mathcal{U}$.
In practice, however, 
if the cardinality of $\mathcal{U}$ is high and we have a windowed \srem,
it might be preferable to use the conditions of the \dsra\ $A$,
if their number is significantly lower.
A \pst\ with too many symbols can quickly become hard to manage as we increase its order $m$ and it is thus advisable to avoid increasing recklessly the size of its alphabet.
\begin{example}
For example,
assume that the values for humidity and temperature take only discrete values (low, high) and that we only have two sensors.
Then 
\begin{equation*}
\begin{aligned}
\mathcal{U} = 
\ & \ \{(T,1,\mathit{low}),  (T,1,\mathit{high}), (T,2,\mathit{low}), (T,2,\mathit{high}), \\
\ & \ (H,1,\mathit{low}),  (H,1,\mathit{high}), (H,2,\mathit{low}), (H,2,\mathit{high})\}
\end{aligned}
\end{equation*}
We can then map $(T,1,\mathit{low})$ to $a$, $(T,1,\mathit{high})$ to $b$, etc.
However, if we have an automaton that only checks whether a measurement comes from the same sensor as a measurement stored in a register,
like $\mathit{EqualId}(\sim,r_{1})$,
then we do not need 8 symbols.
We can use only $a$ and $b$,
with $a$ corresponding to $\mathit{EqualId}(\sim,r_{1})$ and $b$ to $\neg \mathit{EqualId}(\sim,r_{1})$.
The automaton is able to convert a stream/string $S$ constructed from $\mathcal{U}$ to a stream/string of $a$ and $b$ symbols,
which can then be used to construct a \pst.  $\diamond$
\end{example}

\subsection{\srem\ with n-ary conditions on a infinite universe}

Finally, 
the most general case is when we have an infinite universe and $n$-ary conditions. 
In this case,
the applicability of our method is necessarily restricted to windowed \srem\ and \sra. 
We can construct a deterministic \sra\ (which, by definition, has only a single run) from the windowed \srem\ and use this \dsra\ to generate a (single) sequence of symbols from a training dataset.
This sequence can then be used to learn a \pst.
Then, 
the \dsra\ and the \pst\ can be combined,
as already described above,
to estimate the waiting-time distributions and the forecasts.
\begin{example}
As an example,
Figure \ref{fig:isomorphic_det} shows the deterministic classical automaton that can be constructed for \srem\ \eqref{srem:t_seq_h_filter_eq_id} and the \dsra\ of Figure \ref{fig:det:output_free}
(note that Figure \ref{fig:det:output_free} and thus Figure \ref{fig:isomorphic_det} do not depict automata in their totality,
but only part of them).
Assume that we use the stream of Table \ref{table:example_stream} as a training dataset.
We feed it to the automaton of Figure \ref{fig:isomorphic_det}.
Upon reading the first input event,
$(T,1,22)$,
transition $a$ is triggered.
Since the automaton is deterministic,
this is the only triggered transition.
Thus, $(T,1,22)$ is mapped to $a$ and the automaton moves to state $\{q_{t},q_{1}\}$.
With the second event, 
$(T,1,24)$,
transition $f$ is triggered.
Thus, $(T,1,24)$ is mapped to $f$.
We repeat this process until the whole stream of events has been mapped to a stream of symbols,
$S = a, f, \cdots$.
$S$ may now be used to learn a \pst\ of a given maximum order.
This \pst,
along with the automaton of Figure \ref{fig:isomorphic_det},
can be used to estimate the waiting-time distributions,
as in the example of Figure \ref{fig:future}. $\diamond$

\begin{figure}[t]
\begin{centering}
\includegraphics[width=0.65\linewidth]{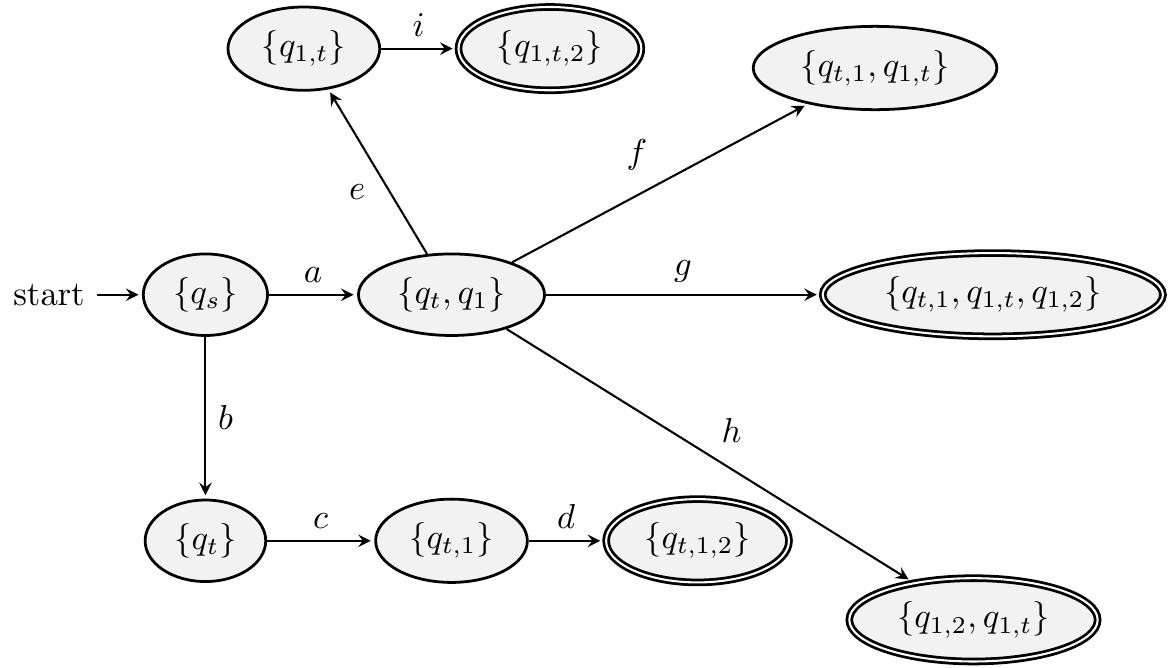}
\caption{Classical deterministic automaton corresponding to \srem\ \eqref{srem:t_seq_h_filter_eq_id}.}
\label{fig:isomorphic_det}
\end{centering}
\end{figure}

\end{example}


%

%% file: outro.tex
\section{Summary \& Future Work}
\label{sec:outro}

We presented an automaton model, \sra, that can act as a computational model for patterns with $n$-ary conditions ($n \geq 1$),
which are quintessential for practical CER/F applications. 
\sra\ thus extend the expressive power of symbolic automata.
They also extend the expressive power of register automata,
through the use of conditions that are more complex than (in)equality predicates.
\sra\ have nice compositional properties, 
without imposing severe restrictions on the use of operators.
Most of the standard operators in CER,
such as concatenation/sequence, union/disjunction, intersection/conjunction and Kleene-star/iteration,
may be used freely.
This is not the case though for complement/negation.
We showed that complement may be used and determinization is also possible,
if a window operator is used,
a very common feature in CER.
We briefly discussed the complexity of the problems of non-emptiness, membership and universality.
Although the problem of membership in general is at least NP-complete,
in cases where we can use windowed, deterministic \sra,
the cost of updating the state of such an automaton after reading a single element is linear in the number of registers and conditions.
We then described how prediction suffix trees may be used to provide a probabilistic description for the behavior of \sra.
Prediction suffix trees can look deep into the past and make accurate inferences about the future behavior of \sra,
thus allowing us to forecast when a complex event is expected to occur.

As a next step, 
we intend to implement the proposed framework for CER/F, 
extending our open-source engine,
Wayeb\footnote{Wayeb source code:\url{https://github.com/ElAlev/Wayeb}},
which currently supports only unary conditions.
We also intend to investigate in the future the possibility of providing more precise complexity results for \srem\ and \sra,
both from the point of view of formal languages and from the point of view of CER/F,
where some extra constraints may exist.
For example,
besides updating the state of an automaton after reading a new element,
we may also need to take into account the time required to report the input events contributing to the detection of a complex event
(for more details about this kind of complexity,
please consult \cite{DBLP:conf/icdt/GrezRU19,DBLP:conf/icdt/GrezR20},
where a model for evaluating efficiency in CER is presented,
along with complexity results for symbolic transducers without memory).
In the future,
we intend to investigate the complexity of decision problems for \sra\ from this point of view.


%% file: appendix.tex
\section{Appendix}
\label{section:appendix}

\subsection{Proof of Theorem \ref{theorem:srem2sra}}
\label{sec:proof:srem2sra}
\input{proofs_srem2sra}

\subsection{Proof of Lemma \ref{lemma:epsilon}}
\label{sec:proof:epsilon}
\input{proofs_epsilon}

\subsection{Proof of Lemma \ref{lemma:multi2single}}
\label{sec:proof:multi2single}
\input{proofs_multi2single}

\subsection{Proof of Theorem \ref{theorem:sra2srem}}
\label{sec:proof:sra2srem}
\input{proofs_sra2srem}

\subsection{Proof of Theorem \ref{theorem:closure}}
\label{sec:proof:closure}
\input{proofs_closure}

\subsection{Proof of Theorem \ref{theorem:complement}}
\label{sec:proof:complement}
\input{proofs_complement}

\subsection{Proof of Theorem \ref{theorem:determinization}}
\label{sec:proof:determinization}
\input{proofs_determinization}

\subsection{Structural Properties of \sra}
\label{sec:structural_properties}
\input{proofs_sra_properties}

\subsection{Proof of Lemma \ref{lemma:windowed_srem}}
\label{sec:proof:windowed_srem}
\input{proofs_windowed_srem}

\subsection{Proof of Theorem \ref{theorem:wsrem2dsra}}
\label{sec:proof:wsrem2dsra}
\input{proofs_wsrem2dsra}

\subsection{Proof of Corollary \ref{corollary:wsra_complement}}
\label{sec:proof:wsra_complement}
\input{proofs_wsra_complement}

\subsection{Proof of Proposition \ref{proposition:streamingsrem}}
\label{sec:proof:streamingsrem}
\input{proofs_streaming_sra}

%% file: proofs_srem2sra.tex
\ifdefined\FI
\begin{theorem}
For every \srem\ $e$ there exists an equivalent \sra\ $A$, i.e., a \sra\ such that $\mathcal{L}(e) = \mathcal{L}(A)$.
\end{theorem}
\else
\begin{theorem*}
For every \srem\ $e$ there exists an equivalent \sra\ $A$, i.e., a \sra\ such that $\mathcal{L}(e) = \mathcal{L}(A)$.
\end{theorem*}
\fi

\begin{proof}

For a \srem\ $e$ and valuations $v$, $v'$,
let $\mathcal{L}(e,v,v')$ denote all strings $S$ such that $(e,S,v) \vdash v'$.
Similarly, for a \sra\ $A$,
let $\mathcal{L}(A,v,v')$ denote all the strings $S=t_{1}, \cdots, t_{n}$ such that there exists an accepting run $[1,q_{1},v_{1}] \overset{\delta_{1}}{\rightarrow} [2,q_{2},v_{2}] \overset{\delta_{2}}{\rightarrow} \cdots \overset{\delta_{n}}{\rightarrow} [n,q_{n+1},v_{n+1}]$,
where $v_{1} = v$ and $v_{n+1} = v'$. 
For every possible \srem\ $e$,
we will construct a corresponding \sra\ $A$ and then prove either that $\mathcal{L}(e) = \mathcal{L}(A)$ or that $\mathcal{L}(e,v,v') = \mathcal{L}(A,v,v')$.
The latter implies that $\mathcal{L}(e,\sharp,v'') = \mathcal{L}(A,\sharp,v'')$ for some valuation $v''$ or equivalently $\mathcal{L}(e) = \mathcal{L}(A)$,
which is our goal.
The proof is inductive.
We prove directly the base cases for the simple expressions $e := \emptyset$,
$e := \epsilon$, $e := \phi = R(x_{1}, \cdots, x_{n})$ and $e := \phi = R(x_{1}, \cdots, x_{n}) \downarrow w$.
For the complex expression $e := e_{1} \cdot e_{2}$, $e := e_{1} + e_{2}$ and $e' = e^{*}$,
we use as an inductive hypothesis that our target result hods for the sub-expressions and then prove that it also holds for the top expression.
For example, 
for $e := e_{1} \cdot e_{2}$,
we assume that $\mathcal{L}(e_{1},v,v'') = \mathcal{L}(A_{1},v,v'')$ and that $\mathcal{L}(e_{2},v'',v') = \mathcal{L}(A_{2},v'',v')$.

We must be careful, however, with the valuations.
If, for example, $v$ applies to the \sra\ $A$,
does it also apply to the sub-automaton $A_{1}$,
if $A$ and $A_{1}$ have different registers?
We can avoid this problem and make all valuations compatible 
(i.e., having the same domain as functions)
by fixing the registers for all expressions and sub-expressions.
We can estimate the registers that we need for a top expression $e$ by scanning its conditions and write operations.
Let $\mathit{reg}(e)$ be a function applied to a \srem\ $e$.
We define it as follows:
\begin{equation}
\mathit{reg}(e) =  
  \begin{cases}
    \emptyset & \quad \text{if } e = \emptyset   \\
    \emptyset & \quad \text{if } e = \epsilon \\
    \{ x_{1} \} \cup \cdots \cup \{ x_{n} \} \cup \{ w \} & \text{if } e = R(x_{1},\cdots, x_{n}) \downarrow w \\
    \mathit{reg}(e_{1}) \cup \mathit{reg}(e_{2}) & \text{if } e = e_{1} \cdot e_{2} \\
    \mathit{reg}(e_{1}) \cup \mathit{reg}(e_{2}) & \text{if } e = e_{1} + e_{2} \\
    \mathit{reg}(e_{1}) & \text{if } e = (e_{1})^{*} 
  \end{cases}
\end{equation}
For our proofs that follow,
we first apply this function to the top expression $e$ to obtain $R_{top} = \mathit{reg}(e)$ and we use $R_{top}$ as the set of registers for all automata and sub-automata.
All valuations can thus be compared without any difficulties,
since they will have the same domain $R_{top}$.

\begin{figure}
\centering
\begin{subfigure}[t]{0.45\textwidth}
	\includegraphics[width=0.95\textwidth]{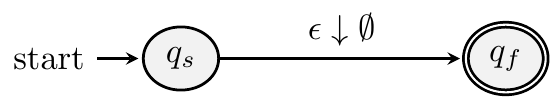}
	\caption{Base case of a single $\epsilon$ condition, $e := \epsilon$.}
	\label{fig:srem2sra:epsilon}
\end{subfigure}
\begin{subfigure}[t]{0.45\textwidth}
	\includegraphics[width=0.95\textwidth]{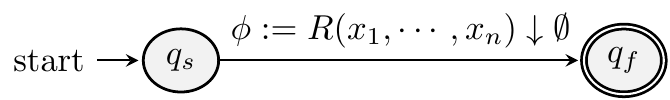}
	\caption{Base case of a single condition, $e := \phi = R(x_{1}, \cdots, x_{n})$.}
	\label{fig:srem2sra:phi}
\end{subfigure}\\
\begin{subfigure}[t]{0.55\textwidth}
	\includegraphics[width=0.95\textwidth]{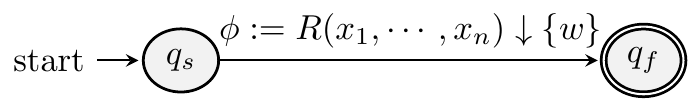}
	\caption{Base case of a single condition with a write register, $e := \phi \downarrow W = R(x_{1}, \cdots, x_{n}) \downarrow \{ w \}$.}
	\label{fig:srem2sra:phiW}
\end{subfigure}\\
\begin{subfigure}[t]{0.65\textwidth}
	\includegraphics[width=0.95\textwidth]{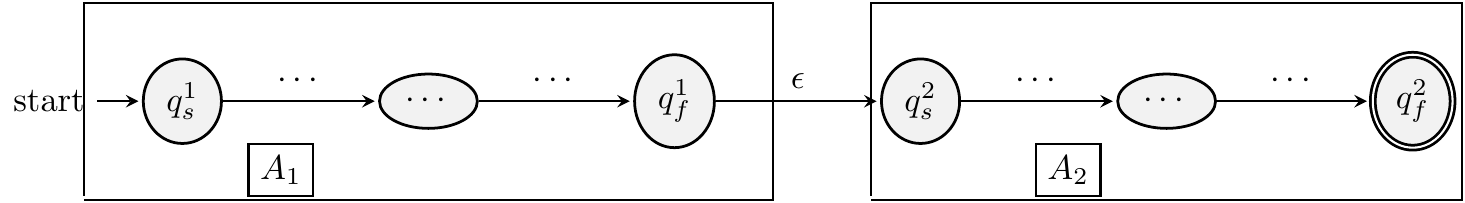}
	\caption{Concatenation. $e = e_{1} \cdot e_{2}$.}
	\label{fig:srem2sra:seq}
\end{subfigure}\\
\begin{subfigure}[t]{0.65\textwidth}
	\includegraphics[width=0.95\textwidth]{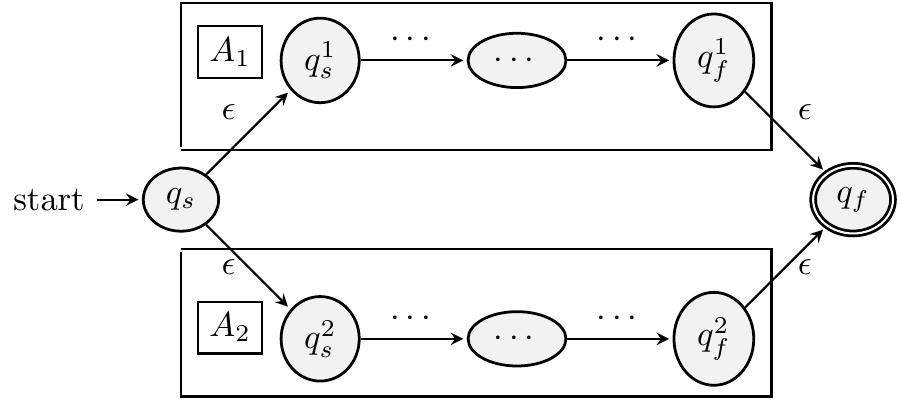}
	\caption{OR. $e = e_{1} + e_{2}$.}
	\label{fig:srem2sra:or}
\end{subfigure}\\
\begin{subfigure}[t]{0.65\textwidth}
	\includegraphics[width=0.95\textwidth]{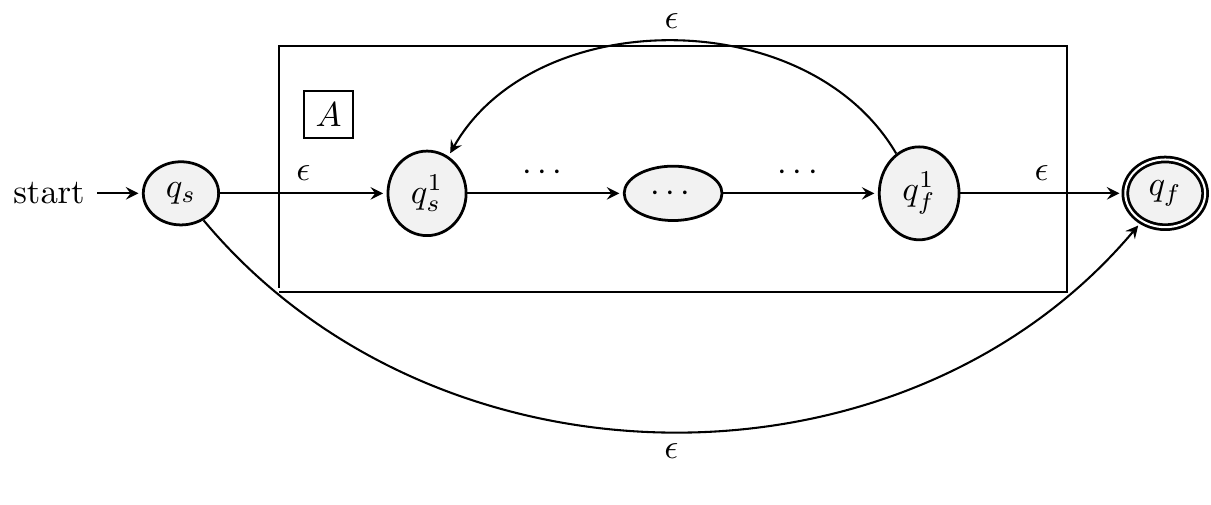}
	\caption{Iteration. $e^{'} = e^{*}$.}
	\label{fig:srem2sra:iter}
\end{subfigure}
\caption{The cases for constructing a \sra\ from a \srem.}
\label{fig:srem2sra}
\end{figure}

\textbf{Assume $e := \emptyset$.}
In this case we know that $\mathcal{L}(e,v,v')=\emptyset$ for any valuations $v$ and $v'$.
Thus $\mathcal{L}(e)=\emptyset$.  
We can then construct a \sra\ $A = (Q,q_{s},Q_{f},R,\Delta)$
where $Q = \{q_{s}\}$, $Q_{f}=\emptyset$, $R=R_{top}$ and $\Delta=\emptyset$.
It is obvious that $A$ does not accept any strings.
Thus $\mathcal{L}(A)=\emptyset$.

\textbf{Assume $e := \epsilon$.}
We know that $\mathcal{L}(e)=\{ \epsilon \}$.
We can then construct a \sra\ $A = (Q,q_{s},Q_{f},R,\Delta)$
where $Q = \{q_{s}, q_{f}\}$, $Q_{f}=\{q_{f}\}$, $R=R_{top}$, $\Delta=\{ \delta \}$ and $\delta = q_{s},\epsilon \downarrow \emptyset \rightarrow q_{f}$.
See Figure \ref{fig:srem2sra:epsilon}.
It is obvious that $A$ accepts only the empty string since there is only one path that leads to the final state and this path goes through an $\epsilon$ transition.
Thus $\mathcal{L}(A)=\{ \epsilon \}$.

\textbf{Assume $e := \phi = R(x_{1}, \cdots, x_{n})$, where $\phi$ is a condition and all $x_{i}$ belong to a set of register variables $\{r_{1},\cdots,r_{k}\}$.}
We construct the following \sra\ $A=(Q,q_{s},Q_{f},R,\Delta)$,
where $Q = \{q_{s}, q_{f}\}$, $Q_{f}=\{q_{f}\}$, $R=R_{top}$, $\Delta=\{ \delta \}$ and $\delta = q_{s},\phi \downarrow \emptyset \rightarrow q_{f}$.
See Figure \ref{fig:srem2sra:phi}.

We first prove $S \in \mathcal{L}(e,v,v') \Rightarrow S \in \mathcal{L}(A,v,v')$ for a string $S$.
It is obvious that $S$ must be composed of a single element, i.e., $S=t_{1}$.
Since $S=t_{1}$ is accepted by $e$ starting from the valuation $v$,
this means that $(\phi,S,v) \vdash v'$, with $v'=v$,
according to the second case of Definition \ref{definition:srem_semantics}.
Thus $(t_{1},v) \models \phi$.
This then implies that the second case in the definition of a successor configuration 
(see Definition \ref{definition:configuration}) holds.
As a result,
$A$, upon reading $S$, moves to its final state $q_{f}$ and accepts $S$.
This move does not change the valuation, thus $v'=v$.
We have thus proven that $S \in \mathcal{L}(A,v,v')$.

The inverse direction,
$S \in \mathcal{L}(A,v,v') \Rightarrow S \in \mathcal{L}(e,v,v')$,
can be proven in a similar manner.

\textbf{Assume $e := \phi = R(x_{1}, \cdots, x_{n}) \downarrow w$, where $\phi$ is a condition, all $x_{i}$ belong to a set of register variables $\{r_{1},\cdots,r_{k}\}$ and $w$ a write register (not necessarily one of $r_{i}$).}
We construct the following \sra\ $A=(Q,q_{s},Q_{f},R,\Delta)$,
where $Q = \{q_{s}, q_{f}\}$, $Q_{f}=\{q_{f}\}$, $R=R_{top}$, $\Delta=\{ \delta \}$ and $\delta = q_{s},\phi \downarrow \{ w \} \rightarrow q_{f}$.
See Figure \ref{fig:srem2sra:phiW}.

The proof is essentially the same as that for the previous case.
The only difference is that we need to use the third case from the definition of successor configurations (Definition \ref{definition:configuration}).
This means that $v' = v[w \leftarrow t_{1}]$.
If $w \in R$, 
then $t_{1}$ is stored in $w$ and $v'(w) = t_{1}$.
Otherwise,
$v'$ remains the same as $v$.

\textbf{Assume $e := e_{1} \cdot e_{2}$, where $e_{1}$ and $e_{2}$ are \srem.}
We first construct $A_{1}$ and $A_{2}$, 
the \sra\ for $e_{1}$ and $e_{2}$ respectively.
We construct the following \sra\ $A=(Q,q_{s},Q_{f},R,\Delta)$,
where $Q = A_{1}.Q \cup A_{2}.Q$, $q_{s}=A_{1}.q_{s}$, $Q_{f}=\{A_{2}.q_{f}\}$, $R=R_{top}$, $\Delta= A_{1}.\Delta \cup A_{2}.\Delta \cup \{ \delta \}$ and $\delta = A_{1}.q_{f},\epsilon \rightarrow A_{2}.q_{s}$.
See Figure \ref{fig:srem2sra:seq}.
We thus simply connect $A_{1}$ and $A_{2}$ with an $\epsilon$ transition.
Notice that $A_{1}.R$ and $A_{2}.R$ may overlap.
Their union retains only one copy of each register,
if a register appears in both of them.

We first prove $S \in \mathcal{L}(e,v,v') \Rightarrow S \in \mathcal{L}(A,v,v')$ for a string $S$.
Since $S \in \mathcal{L}(e,v,v')$,
$S$ can be broken into two sub-strings $S_{1}$ and $S_{2}$ such that
$S= S_{1} \cdot S_{2}$, $(e_{1},S_{1},v) \vdash v''$ and $(e_{2},S_{2},v'') \vdash v'$.
This is equivalent to $S_{1} \in \mathcal{L}(e_{1},v,v'')$ and $S_{2} \in \mathcal{L}(e_{2},v'',v')$.
From the induction hypothesis 
(i.e., that what we want to prove holds for the sub-expressions $e_{1}$, $e_{2}$ and their automata $A_{1}$, $A_{2}$)
it follows that $S_{1} \in \mathcal{L}(A_{1},v,v'')$ and $S_{2} \in \mathcal{L}(A_{2},v'',v')$.
Notice that if $A_{1}$ and $A_{2}$ have different sets of registers,
we can always expand $A_{1}.R$ and $A_{2}.R$ to their union,
without affecting in any way the behavior of the automata.
Now, let $l_{1} = \lvert S_{1} \rvert$ and $l_{2} = \lvert S_{2} \rvert$.
From $S_{1} \in \mathcal{L}(A_{1},v,v'')$ it follows that there exists an accepting run $\varrho_{1}$ of $A_{1}$ over $S_{1}$ such that
$\varrho_{1}=[1,A_{1}.q_{s},v] \rightarrow \cdots \rightarrow [l_{1}+1,A_{1}.q_{f},v'']$.
Similarly,
from $S_{2} \in \mathcal{L}(A_{2},v'',v')$ it follows that there exists an accepting run $\varrho_{2}$ of $A_{2}$ over $S_{2}$ such that
$\varrho_{2}=[1,A_{2}.q_{s},v''] \rightarrow \cdots \rightarrow [l_{2}+1,A_{2}.q_{f},v']$.
Let's construct a run by connecting $\varrho_{1}$ and $\varrho_{2}$ with an $\epsilon$ transition:
$\varrho = [1,A_{1}.q_{s},v] \rightarrow \cdots \rightarrow [l_{1}+1,A_{1}.q_{f},v''] \overset{A_{1}.q_{f},\epsilon \rightarrow A_{2}.q_{s}}{\rightarrow} [l_{1}+2,A_{2}.q_{s},v''] \rightarrow \cdots \rightarrow [l_{1}+l_{2}+1,A_{2}.q_{f},v']$.
We can see that this is indeed an accepting run of $A$.
Thus $S \in \mathcal{L}(A,v,v')$.

The inverse direction,
$S \in \mathcal{L}(A,v,v') \Rightarrow S \in \mathcal{L}(e,v,v')$,
can be proven in a similar manner.
Since $S \in \mathcal{L}(A,v,v')$,
there exists an accepting run $\varrho$ of $A$ over $S$.
By the construction of $A$, however,
this run must be in the form $\varrho = \varrho_{1} \overset{\epsilon}{\rightarrow} \varrho_{2}$ with $\varrho_{1}$ being an accepting run of $A_{1}$ over a string $S_{1}$ and $\varrho_{2}$ an accepting run of $A_{2}$ over $S_{2}$, 
where $S = S_{1} \cdot S_{2}$.
We then use the induction hypothesis to prove that
$S_{1} \in \mathcal{L}(e_{1},v,v'')$ and $S_{2} \in \mathcal{L}(e_{2},v'',v')$
and finally that $S \in \mathcal{L}(e,v,v')$.

\textbf{Assume $e := e_{1} + e_{2}$, where $e_{1}$ and $e_{2}$ are \srem.}
We first construct $A_{1}$ and $A_{2}$, 
the \sra\ for $e_{1}$ and $e_{2}$ respectively.
We construct the following \sra\ $A=(Q,q_{s},Q_{f},R,\Delta)$,
where $Q = A_{1}.Q \cup A_{2}.Q \cup \{ q_{s}, q_{f} \}$, $Q_{f}=\{ q_{f} \}$, $R=R_{top}$, $\Delta= A_{1}.\Delta \cup A_{2}.\Delta \cup \{ \delta_{s,1}, \delta_{s,2}, \delta_{1,f}, \delta_{2,f} \}$ and $\delta_{s,1} = q_{s},\epsilon \rightarrow A_{1}.q_{s}$, $\delta_{s,2} = q_{s},\epsilon \rightarrow A_{2}.q_{s}$, $\delta_{1,f} = A_{1}.q_{f},\epsilon \rightarrow q_{f}$, $\delta_{2,f} = A_{2}.q_{f},\epsilon \rightarrow q_{f}$.
See Figure \ref{fig:srem2sra:or}.
We thus create a new state,
$q_{s}$,
acting as the start state and connect it through $\epsilon$ transitions to the start states of $A_{1}$ and $A_{2}$.
We also create a new final state and connect to it the final states of $A_{1}$ and $A_{2}$.
Again, $A_{1}.R$ and $A_{2}.R$ may overlap.
Their union retains only one copy of each register,
if a register appears in both of them.

It is easy to prove that $S \in \mathcal{L}(e,v,v') \Rightarrow S \in \mathcal{L}(A,v,v')$ for a string $S$.
If $(e_{1},S,v) \vdash v'$,
this implies that $e_{1}$ is accepted by $A_{1}$.
It is thus also accepted by $A$.
Similarly if $(e_{2},S,v) \vdash v'$ for $A_{2}$.
The inverse direction has a similar proof.

\textbf{Assume $e' := e^{*}$, where $e$ is a \srem.}
We construct a new \sra\ $A'$ as shown in Figure \ref{fig:srem2sra:iter}.
We first construct the \sra\ for $e$, $A$.
We create a new final and a new start state.
We connect the new start state to the old start and to the new final.
We connect the old final to the new final and the old start.
$R$ is again $R_{top}$.

We first prove that $S \in \mathcal{L}(e,v,v') \Rightarrow S \in \mathcal{L}(A,v,v')$ for a string $S$.
Since $S \in \mathcal{L}(e,v,v')$,
$S = S_{1} \cdot S'$ such that $(e,S_{1},v) \vdash v''$ and $(e^{*},S',v'') \vdash v'$.
Equivalently,
this implies that
$(e,S_{1},v) \vdash v_{1}$ and
$(e,S_{2},v_{1}) \vdash v_{2}$ and
$(e,S_{3},v_{2}) \vdash v_{3}$ etc
until $(e,S_{n},v_{n-1}) \vdash v_{n}$,
where $v_{n} = v'$.
We can then construct the run $\varrho = \varrho_{1} \overset{\epsilon}{\rightarrow} \varrho_{2} \overset{\epsilon}{\rightarrow} \cdots \overset{\epsilon}{\rightarrow} \varrho_{n}$.
It is easy to see that $\varrho$ is an accepting run of $A'$.
Similarly for the inverse direction.
\end{proof}

%% file: proofs_epsilon.tex
\ifdefined\FI
\begin{lemma}
For every \sra\ $A_{\epsilon} $with $\epsilon$ transitions there exists an equivalent  \sra\ $A_{\notin}$ without $\epsilon$ transitions, i.e., a \sra\ such that $\mathcal{L}(A_{\epsilon}) = \mathcal{L}(A_{\notin})$.
\end{lemma}
\else
\begin{lemma*}
For every \sra\ $A_{\epsilon} $with $\epsilon$ transitions there exists an equivalent  \sra\ $A_{\notin}$ without $\epsilon$ transitions, i.e., a \sra\ such that $\mathcal{L}(A_{\epsilon}) = \mathcal{L}(A_{\notin})$.
\end{lemma*}
\fi

\begin{proof}
\input{algo_epsilon}
We first give the algorithm.
See Algorithm \ref{algorithm:epsilon}.
Note that in this algorithm,
the function $\mathit{Enclose}$ is the usual function for $\epsilon$-enclosure in standard automata theory and we will not repeat it here (see \cite{DBLP:books/daglib/0016921}).
Suffice it to say that,
when applied to a state $q$ (or set of states $\{q_{i}\}$),
it returns all the states we can reach from $q$ (or all $q_{i}$)
by following only $\epsilon$-transitions.
It is also worth noting that the algorithm does not create the power-set of states and then connects them through transitions.
It creates those subsets it needs by ``forward-looking'' for what is necessary,
but it is equivalent to the power-set construction algorithm.
We will prove that $S \in \mathcal{L}(A_{\epsilon}) \Leftrightarrow S \in \mathcal{L}(A_{\notin})$ for a string $S$.

We first prove the direction $S \in \mathcal{L}(A_{\epsilon}) \Rightarrow S \in \mathcal{L}(A_{\notin})$.
The other direction can be proven similarly.
Let $\varrho_{\epsilon}$ denote an accepting run of $A_{\epsilon}$ over $S$,
where $k = \lvert S \rvert$ is the length of $S$.
\begin{equation}
\label{run:epsilon}
\begin{aligned}
\varrho_{\epsilon} = & [1,q_{\epsilon,1}=q_{\epsilon,s},v_{\epsilon,1}=\sharp] \overset{\epsilon}{\rightarrow} [\cdots] \overset{\epsilon}{\rightarrow} \cdots & \text{sub-run 1}  \\
 & \overset{\delta_{\epsilon,1}}{\rightarrow} [2,q_{\epsilon,2},v_{\epsilon,2}] \overset{\epsilon}{\rightarrow} [\cdots] \overset{\epsilon}{\rightarrow} \cdots  & \text{sub-run 2}  \\
 & \cdots & \\
 & \overset{\delta_{\epsilon,i-1}}{\rightarrow} [i,q_{\epsilon,i},v_{\epsilon,i}] \overset{\epsilon}{\rightarrow} [\cdots] \overset{\epsilon}{\rightarrow} [i',q_{\epsilon,i'},v_{\epsilon,i'}]  & \text{sub-run i}  \\
 & \overset{\delta_{\epsilon,i}}{\rightarrow} [i+1,q_{\epsilon,i+1},v_{\epsilon,i+1}] \overset{\epsilon}{\rightarrow} [\cdots] \overset{\epsilon}{\rightarrow} \cdots  & \text{sub-run i+1}  \\
 & \cdots & \\
 & \overset{\delta_{\epsilon,k}}{\rightarrow} [k+1,q_{\epsilon,k+1} \in Q_{\epsilon,f},v_{\epsilon,k+1}]  & \text{sub-run k+1}  
\end{aligned}
\end{equation}
Let $\varrho_{\notin}$ denote a run of $A_{\notin}$ over $S$.
\begin{equation}
\label{run:noepsilon}
\begin{aligned}
\varrho_{\notin} = & [1,q_{\notin,1}=q_{\notin,s},v_{\notin,1}=\sharp] & \text{sub-run 1} \\
 & \overset{\delta_{\notin,1}}{\rightarrow} [2,q_{\notin,2},v_{\notin,2}] & \text{sub-run 2} \\
 & \cdots & \\
 & \overset{\delta_{\notin,i-1}}{\rightarrow} [i,q_{\notin,i},v_{\notin,i}] & \text{sub-run i}\\
 & \overset{\delta_{\notin,i}}{\rightarrow} [i+1,q_{\notin,i+1},v_{\notin,i+1}] & \text{sub-run i+1}\\
 & \cdots & \\
 & \overset{\delta_{\notin,k}}{\rightarrow} [k+1,q_{\notin,k+1},v_{\notin,k+1}] & \text{sub-run k+1}
\end{aligned}
\end{equation}
$\varrho^{\notin}$ necessarily follows $k$ transitions,
since it does not have any $\epsilon$-transitions.
On the other hand, 
$\varrho^{\epsilon}$ may follow more than $k$ transitions ($j \geq k$),
because several $\epsilon$ transitions may intervene between ``actual'', non-$\epsilon$ transitions, 
as shown in Run \ref{run:epsilon}.
The number of non-$\epsilon$ transitions is still $k$.
$\varrho_{\epsilon}$ is thus necessarily composed of $k$ ``sub-runs'',
where the first configuration of each sub-run is reached via a non-$\epsilon$ transition, followed by a sequence of 0 or more $\epsilon$ transitions.
Each line in Run \ref{run:epsilon} is such a sub-run.
We can also split $\varrho_{\notin}$ in sub-runs,
but in this case each such sub-run will be simply composed of a single configuration.
See Run \ref{run:noepsilon}.

We will prove the following.
For each sub-run $i$ of $\varrho_{\epsilon}$,
it holds that:
\begin{enumerate}
	\item $q_{\notin,i} \in q_{\epsilon,i}$. In fact, $q_{\notin,i} = \mathit{Enclosure}(q_{\epsilon,i})$.
	\item $v_{\epsilon,i} = v_{\notin,i}$, i.e., $A_{\epsilon}$ and $A_{\notin}$ have the same register contents at each $i$.
\end{enumerate} 
We can prove this inductively.
We assume that the above claims hold for $i$ and then we can show that they must necessarily hold for $i+1$.
Since they are obviously true for $i=1$,
they are then true for $i=k+1$ as well.
Thus, $q_{\notin,k+1} \in Q_{\notin,f}$ and $\varrho_{\notin}$ is an accepting run as well.

First, notice that in each sub-run $i$ of $\varrho_{\epsilon}$,
$v_{\epsilon,i}$ remains the same,
since $\epsilon$ transitions never modify the contents of the registers.
Thus, in $\varrho_{\epsilon}$, $v_{\epsilon,i'}=v_{\epsilon,i}$.
It is also obviously true that $i'=i$,
since $\epsilon$ transitions do not read elements from $S$ and thus the automaton's head does not move.
The only thing that could possibly change is $q_{\epsilon,i'}$,
so that, in general, $q_{\epsilon,i'} \neq q_{\epsilon,i}$.
Therefore,
in Run \ref{run:epsilon},
we move from sub-run $i$ to sub-run $i+1$ by jumping from $q_{\epsilon,i'}$ to $q_{\epsilon,i+1}$.
This implies that $\delta_{\phi,i}$,
connecting $q_{\epsilon,i'}$ to $q_{\epsilon,i+1}$,
is triggered when the contents of the register are those of $v_{\epsilon,i'}=v_{\epsilon,i}$.

Now, $q_{\epsilon,i'}$ belongs to the enclosure of $q_{\epsilon,i}$.
Otherwise, it would be impossible to reach it from $q_{\epsilon,i}$ by following only $\epsilon$ transitions.  
From the induction hypothesis we know that $q_{\notin,i}$ must be the enclosure of $q_{\epsilon,i}$.
From the construction algorithm for $A_{\notin}$ (Algorithm \ref{algorithm:epsilon}) 
we also know that the transition $\delta_{\epsilon,i}$ also exists in $A_{\notin}$,
with $q_{\notin,i}$ its source.
$\delta_{\notin,i}$ is has the same condition and references the same registers as $\delta_{\epsilon,i}$.
Since $\delta_{\epsilon,i}$ is triggered with $v_{\epsilon,i'}$,
$\delta_{\notin,i}$ must also be triggered because $v_{\notin} = v_{\epsilon,i}$ (by the induction hypothesis) and thus $v_{\notin} = v_{\epsilon,i'}$.
From the construction algorithm,
we can see that $q_{\notin,i+1}$ will be the enclosure of $q_{\epsilon,i+1}$ 
The state $q_{e}$ in Algorithm \ref{algorithm:epsilon} is $q_{\epsilon,i'}$ in Run \ref{run:epsilon},
while state $p_{\epsilon}$ in the algorithm is state $q_{\epsilon,i+1}$ in $\varrho_{\epsilon}$.
$p_{\notin}$ is the thus the enclosure of $q_{\epsilon,i+1}$.
Thus there exists $q_{\notin,i+1}$ which we can reach from $q_{\notin,i}$ and which is the enclosure of $q_{\epsilon,i}$.
The second part of the induction hypothesis is obviously true for $i+1$,
i.e., $v_{\epsilon,i+1} = v_{\notin,i+1}$,
since exactly the same registers are modified in exactly the same way by $\delta_{\epsilon,i}$ and $\delta_{\notin,i}$.
 
Therefore, $q_{\epsilon,k+1} \in q_{\notin,k+1}$ which implies that $q_{\notin,k+1} \in Q_{\notin,f}$ and thus $\varrho_{\notin}$ is an accepting run of $A_{\notin}$ over $S$.
\end{proof}

%% file: algo_epsilon.tex
\begin{algorithm}
\KwIn{\sra\ $A_{\epsilon}$, possibly with $\epsilon$ transitions}
\KwOut{\sra\ $A_{\notin}$ without $\epsilon$-transitions}
$q_{\notin,s} \leftarrow Enclose(A_{\epsilon}.q_{s})$; $Q_{\notin} \leftarrow \{ q_{\notin,s} \}$; $\Delta_{\notin} \leftarrow \emptyset$\;
\eIf{$\exists q \in q_{\notin,s}: q \in A_{\epsilon}.Q_{f}$}{
	$Q_{\notin,f} \leftarrow \{ q_{\notin,s} \}$\;
}
{
	$Q_{\notin,f} \leftarrow \emptyset$\;
}
$\mathit{frontier} \leftarrow \{q_{\notin,s}\}$\;
\ForEach{$q_{\notin} \in \mathit{frontier}$}{
	\ForEach{$q_{\epsilon} \in q_{\notin}$}{
		\ForEach{$\delta_{\epsilon} \in A_{\epsilon}.\Delta: \delta_{\epsilon}.\mathit{source}=q_{\epsilon} \wedge \delta_{\epsilon} \neq \epsilon$}{
			$p_{\epsilon} \leftarrow \delta_{\epsilon}.target$\;
			$p_{\notin} \leftarrow Enlose(p_{\epsilon})$\; \label{line:epsilon:enclose_target}
			$Q_{\notin} \leftarrow Q_{\notin} \cup \{ p_{\notin} \}$\;
			\If{$\exists q \in p_{\notin}: q \in A_{\epsilon}.Q_{f}$}{
				$Q_{\notin,f} \leftarrow Q_{\notin,f} \cup \{ p_{\notin} \}$\;
			}
			$\delta_{\notin} \leftarrow \mathit{CreateNewTransition}(q_{\notin},\delta_{\epsilon}.\phi \downarrow \delta_{\epsilon}.W \rightarrow p_{\notin})$\;
			$\Delta_{\notin} \leftarrow \Delta_{\notin} \cup \{\delta_{\notin}\}$\;
			$\mathit{frontier} \leftarrow \mathit{frontier} \cup \{p_{\notin}\}$\;
		}
	}
	$\mathit{frontier} \leftarrow \mathit{frontier} \setminus \{q_{\notin}\}$\;
}
$A_{\notin} \leftarrow (Q_{\notin}, q_{\notin,s}, Q_{\notin,f}, A_{\epsilon}.R,\Delta_{\notin})$\;
$\mathtt{return}\ A_{\notin}$;
\caption{Eliminating $\epsilon$-transitions ($\mathit{EliminateEpsilon}$).}
\label{algorithm:epsilon}
\end{algorithm}

%% file: proofs_multi2single.tex
\ifdefined\FI
\begin{lemma}
For every multi-register \sra\ $A_{mr}$ there exists an equivalent single-register \sra\ $A_{sr}$, i.e., a single-register \sra\ such that $\mathcal{L}(A_{mr}) = \mathcal{L}(A_{sr})$.
\end{lemma}
\else
\begin{lemma*}
For every multi-register \sra\ $A_{mr}$ there exists an equivalent single-register \sra\ $A_{sr}$, i.e., a single-register \sra\ such that $\mathcal{L}(A_{mr}) = \mathcal{L}(A_{sr})$.
\end{lemma*}
\fi

\begin{proof}
The proof is constructive.
We construct a new single-register \sra\ and show that it has the same language as the multi-register \sra.
The new \sra, $A_{sr}$, has the same number of registers as $A_{mr}$.
The main difference is that $A_{sr}$ has more states than $A_{mr}$. 
See \cite{DBLP:journals/tcs/KaminskiF94} for a similar proof about register automata.

Let $A_{mr} = (Q_{mr},q_{mr,s},Q_{mr,f},R_{mr},\Delta_{mr})$ and $A_{sr} = (Q_{sr},q_{sr,s},Q_{sr,f},R_{sr},\Delta_{sr})$ denote the multi- and single-register \sra\ respectively.
Let $w= \lvert R_{mr} \rvert = \lvert R_{sr} \rvert$ denote the number of registers (the same for $A_{mr}$ and $A_{sr}$).
Let $p = (p_{1}, \cdots, p_{w}) \in (2^{R_{mr}})^{w}: \bigcup_{k=1}^{w}=R_{mr} \text{ and } p_{i} \cap p_{j} = \emptyset \text{ for } i \neq j$.
In other words,
each $p_{i}$ is a subset of $R_{mr}$ and the union of all $p_{i}$ gives us $R_{mr}$.
Therefore, $p$ denotes a partition of $R_{mr}$.
For example,
if $R_{mr} = \{ r_{1}, r_{2}, r_{3}, r_{4} \}$,
a possible partition would be $p = (\{ r_{1}, r_{3}\}, \{r_{2}\}, \{r_{4}\}, \emptyset)$.
Let $P_{R_{mr}} = \{p \mid p \text{ is a partition of } R_{mr}\}$ denote the set of all possible partitions of $R_{mr}$.

The general idea is that we want to establish a correspondence between the registers of $A_{mr}$ and those of $A_{sr}$.
If all the registers in $R_{mr}$ have different contents,
then each one of them may correspond to a unique register in $R_{sr}$.
However, since a transition in $A_{mr}$ may write to multiple registers,
at some point in a run of $A_{mr}$,
some of its registers will have the same contents.
For example,
if $R_{mr} = \{ r_{1}, r_{2}, r_{3}, r_{4} \}$, 
a transition may write to $r_{1}$ and $r_{3}$ at the same time.
In this case then,
the registers of $R_{mr}$ may be partitioned as follows,
according to which of them have the same contents: 
$p = (\{ r_{1}, r_{3}\}, \{r_{2}\}, \{r_{4}\}, \emptyset)$.
Now, we could map each register of $R_{rs}$ to one of the $p_{i}$ in $p$.
Repeated values in $R_{mr}$ would then exist as single values in $R_{rs}$.
The next issue would then be how we could actually track in a run of $A_{mr}$ the registers that have the same value(s).
We could actually achieve this by combining the states of $A_{mr}$ with every possible partition.

For $A_{sr}$ we would then have:
\begin{itemize}
	\item $Q_{sr} = Q_{mr} \times P_{R_{mr}}$, where $\times$ indicates the Cartesian product.
	\item $q_{sr,s} = (q_{mr,s},p_{s})$. where $p_{s} = (R_{mr},\overbrace{\emptyset,\cdots,\emptyset}^{w-1})$.
	\item $Q_{sr,f} = Q_{mr,f} \times P_{R_{mr}}$.
	\item $R_{mr} = \{r_{mr,1}, \cdots, r_{mr,w}\}$.
	\item The set of transitions $\Delta_{sr}$ is defined as follows:
	\begin{itemize}
		\item For every $\delta_{mr} \in \Delta_{mr}: \delta_{mr} = q_{mr},\epsilon \rightarrow q_{mr}'$, 
		i.e., for every $\epsilon$ transition of $A_{mr}$,
		we add the transitions $\delta_{sr} = (q_{mr},p),\epsilon \rightarrow (q_{mr}',p)$,
		one for each $p \in P_{R_{mr}}$.
		\item For every $\delta_{mr} \in \Delta_{mr}: \delta_{mr} = q_{mr},\phi_{mr} \rightarrow q_{mr}'$, 
		i.e., for every non-$\epsilon$ transition of $A_{mr}$ that does not write to any registers,
		we do the following.
		Let $R_{\phi}$ denote the set of registers referenced/accessed by $\phi_{mr}$.
		For every $p \in P_{R_{mr}}$ we add a transition $\delta_{sr} = (q_{mr},p),\phi_{sr} \rightarrow (q_{mr}',p)$.
		$\phi_{sr}$ is the same as $\phi_{mr}$ with the following difference.
		Each $r_{mr,i} \in R_{\phi}$ 
		(i.e., each register referenced by $\phi_{mr}$) 
		is replaced in the condition by $r_{sr,j} \in R_{sr}$,
		where $j$ is such that $r_{mr,i} \in p_{j}$. 
		For example, 
		if $p = (p_{1},p_{2},p_{3},p_{4}) = (\{ r_{mr,1}, r_{mr,3}\}, \{r_{mr,2}\}, \{r_{mr,4}\}, \emptyset)$ and $\phi_{mr} = \phi(r_{mr,4})$, then, $\phi_{sr} = \phi(r_{sr,3})$, since $r_{mr,4} \in p_{3}$.
		Notice that only one such $j$ exists because $p$ is a partition and, by definition, the different sets of a partition do not have common elements.
		\item For every $\delta_{mr} \in \Delta_{mr}: \delta_{mr} = q_{mr},\phi_{mr} \downarrow (\cdots, r_{mr,i}, \cdots) \rightarrow q_{mr}'$
		i.e., for every non-$\epsilon$ transition of $A_{mr}$ that also writes to registers,
		we do the following.
		For every $p \in P_{R_{mr}}$ we add a transition $\delta_{sr} = (q_{mr},p),\phi_{sr} \downarrow r_{sr,k} \rightarrow (q_{mr}',p')$. 
		$\phi_{sr}$ is defined as in the previous case.
		Now, let $R_{w} = (\cdots, r_{mr,i}, \cdots)$ denote all the write registers.
		$k$ in $r_{rs,k}$ is defined as the minimal integer such that $p_{k} \subseteq R_{w}$.
		Additionally, $p'$ is defined as follows. 
		$p_{k}' = p_{k} \cup R_{w}$ and $p_{k'}' = p_{k'} \setminus R_{w}$ for $k' \neq k$.
		For example,
		if $p = (p_{1},p_{2},p_{3},p_{4}) = (\{ r_{mr,1}, r_{mr,3}\}, \{r_{mr,2}\}, \{r_{mr,4}\}, \emptyset)$ and $R_{w} = \{r_{mr,1}, r_{mr,2}, r_{mr,3} \}$ then
		$k = 1$, since $p_{1} = \{ r_{mr,1}, r_{mr,3}\} \subset R_{w}$.
		We also have that $p'= (p_{1}',p_{2}',p_{3}',p_{4}') = (\{ r_{mr,1}, r_{mr,2}, r_{mr,3}\}, \emptyset, \{r_{mr,4}\}, \emptyset)$.
	\end{itemize}
\end{itemize}

We want to show that $\mathcal{L}(A_{mr}) = \mathcal{L}(A_{sr})$.
First, assume that $S \in \mathcal{L}(A_{mr})$ for a string $S$.
We will show that $S \in \mathcal{L}(A_{sr})$.
Let $\varrho_{mr}$ be an accepting run of $A_{mr}$ over $S$:
\begin{equation*}
\begin{aligned}
\varrho_{mr} = & [1,q_{mr,1}=q_{mr,s},v_{mr,1}=\sharp] \overset{\delta_{mr,1}}{\rightarrow} \\
& \cdots \\
 & \overset{\delta_{mr,i-1}}{\rightarrow} [i,q_{mr,i},v_{mr,i}] \overset{\delta_{mr,i}}{\rightarrow} \\
 & \cdots \\
 & \overset{\delta_{mr,l}}{\rightarrow} [l+1,q_{mr,l+1} \in Q_{mr,f},v_{mr,l+1}]
\end{aligned}
\end{equation*}
Let $\varrho_{sr}$ be a run of $A_{sr}$ over $S$:
\begin{equation*}
\begin{aligned}
\varrho_{sr} = & [1,q_{sr,1}=q_{sr,s},v_{sr,1}=\sharp] \overset{\delta_{sr,1}}{\rightarrow} \\
& \cdots \\
 & \overset{\delta_{sr,i-1}}{\rightarrow} [i,q_{sr,i},v_{sr,i}] \overset{\delta_{sr,i}}{\rightarrow} \\
 & \cdots \\
 & \overset{\delta_{sr,l}}{\rightarrow} [l+1,q_{sr,l+1},v_{sr,l+1}]
\end{aligned}
\end{equation*}
We need to show that $q_{sr,l+1} \in Q_{sr,f}$.
We can prove this inductively.
As our induction hypothesis,
we assume that $q_{sr,i} = (q,p)$,
where 
\begin{enumerate}
	\item $q = q_{mr,i}$ and
	\item $p= (p_{i,1},\cdots,p_{i,w})$ such that
for all $1 \leq k \leq w$ (i.e., all registers of $A_{sr}$)
$v_{sr,i}(r_{sr,k}) = v_{mr,i}(r_{mr,k'})$ for all $r_{mr,k'} \in p_{i,k}$.
\end{enumerate}
In other words,
at the $i^{th}$ configuration in $\varrho_{sr}$,
we have reached a state $q_{sr,i}$.
The first element of this state must be $q_{mr,i}$,
i.e., the state at the $i_{th}$ configuration of $\varrho_{mr}$.
The second element must be a partition (of the registers of $A_{mr}$).
The contents of the first register of $A_{sr}$, $r_{sr,1}$,
must be equal to the contents of the registers of the first set in the partition,
the contents of the second register $r_{st,2}$ equal to those of the second set, etc.
For example,
if $p = (p_{i,1},p_{i,2},p_{i,3},p_{i,4}) = (\{ r_{mr,1}, r_{mr,3}\}, \{r_{mr,2}\}, \{r_{mr,4}\}, \emptyset)$,
then the contents of $r_{sr,1}$ must be the same as those of $r_{mr,1}$ and $r_{mr,3}$.
Similarly, the contents of $r_{sr,2}$ and $r_{mr,2}$ must be the same.
Similarly for $r_{sr,3}$ and $r_{mr,4}$.
We must then show that,
if the induction hypothesis holds for the $i^{th}$ configuration,
it also holds for the $(i+1)^{th}$ one.

First, assume that $\delta_{mr,i} = q_{mr,i},\epsilon \rightarrow q_{mr,i+1}$.
Then, from the construction of $A_{sr}$,
we know that there exists $\delta_{sr} \in \Delta_{sr}$ such that $\delta_{sr} = (q_{mr,i},p),\epsilon \rightarrow (q_{mr,i+1},p)$.
The first part of the induction hypothesis then still holds because we can move to a state whose first element is $q_{mr,i+1}$.
The second part of the hypothesis also holds because $p$ remains the same and we already know that this part holds for $p$ from the hypothesis itself.

Now, assume that  $\delta_{mr,i} = q_{mr,i},\phi_{mr} \rightarrow q_{mr,i+1}$,
with $\phi_{mr} \neq \epsilon$.
Then, from the construction of $A_{sr}$,
we know that there exists $\delta_{sr} \in \Delta_{sr}$ such that $\delta_{sr} = (q_{mr,i},p),\phi_{sr} \rightarrow (q_{mr,i+1},p)$ and the condition $\phi_{sr}$ is triggered.
We can prove the latter claim about $\phi_{sr}$ by noticing that $\phi_{mr}$ is triggered.
Now, from the construction,
$\phi_{sr}$ is the same as $\phi_{mr}$ with its arguments/registers appropriately replaced,
as described above.
Without loss of generality,
assume that $\phi_{mr}$ references all of its registers
(if this is not the case,
we can always construct an equivalent condition that references all registers, 
but does not actually access any of the redundant ones).
We can write it as follows:
\begin{equation*}
\phi_{mr} = \phi(r_{mr,1}, \cdots, r_{mr,w})
\end{equation*}
$\phi_{sr}$ can be written as follows:
\begin{equation*}
\phi_{sr} = \phi(r_{sr,i_{1}}, \cdots, r_{mr,i_{w}})
\end{equation*}
where $i_{1}$ is such that $r_{mr,1} \in p_{i_{1}}$,
i.e., $i_{1}$ is the partition set from $p$ where $r_{mr,1}$ belongs.
Similarly for $i_{2}$, etc.
For example,
if $R_{mr} = \{r_{mr,1},r_{mr,2},r_{mr,3},r_{mr,4}\}$ and $p = (p_{1},p_{2},p_{3},p_{4}) = (\{ r_{mr,1}, r_{mr,3}\}, \{r_{mr,2}\}, \{r_{mr,4}\}, \emptyset)$ 
then
\begin{equation*}
\phi_{mr} = \phi(r_{mr,1}, r_{mr,2}, r_{mr,3}, r_{mr,4})
\end{equation*}
and 
\begin{equation*}
\phi_{sr} = \phi(r_{sr,1}, r_{sr,2}, r_{sr,1}, r_{sr,3})
\end{equation*}
From the induction hypothesis,
we know, however,
that
$v_{mr,i}(r_{mr,j}) = v_{sr,i}(r_{sr,i_{j}})$.
Therefore,
$\phi_{mr}$ and $\phi_{sr}$ essentially have the same arguments.
Since $\phi_{sr}$ is also triggered,
the first part of the induction hypothesis holds for $(i+1)$ as well.
The second part also holds since $p$ again remains the same.

Finally, assume that $\delta_{mr,i} = q_{mr,i},\phi_{mr} \downarrow R_{w} \rightarrow q_{mr,i+1}$,
with $\phi_{mr} \neq \epsilon$ and $R_{w} \neq \emptyset$.
Then, from the construction of $A_{sr}$,
we know that there exists $\delta_{sr} \in \Delta_{sr}$ such that $\delta_{sr} = (q_{mr,i},p),\phi_{sr} \downarrow r_{rs,k} \rightarrow (q_{mr,i+1},p')$ and the condition $\phi_{sr}$ is triggered.
The proof for $\phi_{sr}$ being triggered is the same as in the previous case.
We additionally need to prove the second part of the induction hypothesis,
since $p$ now becomes $p'$.
In other words,
we need to prove for $p'$ that the contents of a register $j$ of $A_{sr}$ are the same as those of the registers of $A_{mr}$ contained in the $j^{th}$ set in the partition $p'$. 
Indeed, this is the case.
From the construction,
we know that the partition set $p_{k}$
(reminder: $r_{sr,k}$ is the write register in $\delta_{sr}$)
becomes $p_{k}' = p_{k} \cup R_{w}$.
Since $p_{k} \subset R_{w}$,
this means that $p_{k}' = R_{w}$.
Thus the hypothesis still holds for $p_{k}'$,
since $r_{sr,k}$ and all registers in $R_{w}$ will have the same value.
The hypothesis also holds for $k' \neq k$.
Since $p_{k'}' = p_{k'} \setminus R_{w}$ and the hypothesis holds for $p$,
this means that $p_{k'}$ had registers with the same contents before the writing to $R_{w}$. 
If we remove from $p_{k'}$ the changed registers to obtain $p_{k'}'$,
then $p_{k'}'$ will still have registers with the same contents after the writing to $R_{w}$.
Additionally, since $r_{sr,k'}$ was not changed, 
it will still have the same contents as the registers in $p_{k'}'$.

We know that the hypothesis holds for $i=1$ in the runs $\varrho_{mr}$ and $\varrho_{sr}$,
because $q_{sr,s} = (q_{mr,s},p_{s})$,
with $p_{s} = (R_{mr}, \emptyset, \cdots, \emptyset)$.
The first part of the hypothesis obviously holds.
The second part also holds since all registers are empty in both runs at the beginning.
Therefore, the hypothesis holds for $i=2$, $i=3$, etc.
$\varrho_{sr}$ is thus an accepting run of $A_{sr}$ over $S$.

We have proven that $S \in \mathcal{L}(A_{mr}) \Rightarrow S \in \mathcal{L}(A_{sr})$.
The inverse direction,
i.e., $S \in \mathcal{L}(A_{sr}) \Rightarrow S \in \mathcal{L}(A_{mr})$,
can be proven similarly.
We first assume that there is an accepting run $\varrho_{sr}$ of $A_{sr}$ over $S$ and then show, in a similar manner, that there exists an accepting runf= $\varrho_{mr}$ of $A_{mr}$ over $S$. 
\end{proof}

%% file: proofs_sra2srem.tex
\ifdefined\FI
\begin{theorem}
For every \sra\ $A$ there exists an equivalent \srem\ $e$, i.e., a \srem\ such that $\mathcal{L}(A) = \mathcal{L}(e)$.
\end{theorem}
\else
\begin{theorem*}
For every \sra\ $A$ there exists an equivalent \srem\ $e$, i.e., a \srem\ such that $\mathcal{L}(A) = \mathcal{L}(e)$.
\end{theorem*}
\fi
\begin{proof}

The proof develops along lines similar to the corresponding proof for classical and register automata \cite{DBLP:books/daglib/0086373,DBLP:journals/jcss/LibkinTV15}.
It uses a generalized version of \sra, 
denoted by \gsra.
These are \sra\ whose transitions are not equipped with a single condition,
but with a whole \srem.
For example,
the \gsra\ 
$A=(Q,q_{s},Q_{f},R,\Delta)$,
where $Q = \{q_{s}, q_{f}\}$, $Q_{f}=\{q_{f}\}$, $R=\{r_{1}\}$, $\Delta=\{ \delta \}$ and $\delta = q_{s},(\phi_{1} \downarrow \{ r_{1} \}) \cdot (\phi_{2}(r_{1}))  \rightarrow q_{f}$.
The single transition $\delta$ can read two characters at the same time,
apply $\phi_{1}$ on the first one,
store this character in $r_{1}$ and then apply $\phi_{2}$ on the second character and the contents of $r_{1}$.
We assume that all $\epsilon$ transitions have been eliminated and that the \sra\ is single-register 
(and if not, 
it has been converted to one,
as shown in Appendix \ref{sec:proof:multi2single}).
We also demand that 
a) \gsra\ have a single start state with no incoming transitions and with outgoing transitions to every other state, 
b) they have a single final state with no outgoing transitions and with incoming transitions from every other state,
c) there is an arrow connecting any two other states.
We say that a \gsra\ $A_{g}$ accepts a string $S$ if $S = S_{1} \cdot S_{2} \cdot \cdots \cdot S_{k}$ and there exists a run $\varrho = [1,q_{1},v_{1}] \rightarrow \cdots \rightarrow [i,q_{i},v_{i}] \rightarrow \cdots \rightarrow [l+1,q_{k+1},v_{k+1}]$,
where the state of the first configuration is the start state of $A_{g}$ ($q_{1}=A_{g}.q_{s}$),
the state of the last configuration is its final state ($q_{k+1}=A_{g}.q_{f}$) and for each $i$ there exists a transition $\delta$ of $A_{g}$ such that  $\delta = q_{i},e_{i} \rightarrow q_{i+1}$ and $S_{i} \in \mathcal{L}(e_{i},v_{i},v_{i+1})$.

We first convert the initial \sra\ $A$ to a \gsra\ $A_{g}$ as follows.
We add a new start state and connect it to the old start state with an $\epsilon$ transition.
We also connect with such a transition the old final state to a new final state.
It there are multiple transitions between any two states,
we combine them into a single transition whose condition will be the union of the conditions of the previous transitions 
(we can do this since we are allowed to have \srem\ on the transitions now).
Finally, if there exist states that are not connected to each other,
we connect them with $\emptyset$ transitions
(making sure that we do not add incoming transitions to the start state or outgoing transitions to the final state).
This procedure will produce a \gsra\ $A_{g}$ which will be equivalent in terms of its language to the original \sra\ $A$.

The basic idea is to start removing states from $A_{g}$,
one at a time,
without affecting the language it accepts.
This procedure is repeated until we are left with 2 states.
At this point,
the \gsra\ will have one start and one final state,
connected with a single transition.
The \srem\ on this transition is finally returned as the \srem\ corresponding to the initial \sra\ $A$.
The critical step in this process is of course the one where a state is removed.
We must ensure that any repairs we make to the remaining transitions do not affect the automaton's language.
We first select a state to remove, $q_{rip}$.
This can be any state, 
except for the start or the final states.
We then check all pairs of states $q_{i}$ and $q_{j}$.
We need to make sure that the new automaton will be able to move from $q_{i}$ to $q_{j}$ with exactly the same strings as when $q_{rip}$ was present.
We thus have to modify the \srem\ on the transition from $q_{i}$ to $q_{j}$.
The modification is the following.
Assume, that, in the old automaton, before the removal,
the following hold:
\begin{itemize}
	\item $q_{i},e_{1} \rightarrow q_{rip}$.
	\item $q_{rip},e_{2} \rightarrow q_{rip}$.
	\item $q_{i},e_{3} \rightarrow q_{j}$.
	\item $q_{i},e_{4} \rightarrow q_{j}$.
\end{itemize}
Notice that such transitions exist for every pair $q_{i}$ and $q_{j}$,
since, by definition, transitions exist between all pairs of states.
Then, after removing $q_{rip}$,
the \srem\ on the transition from $q_{i}$ to $q_{j}$ becomes $(e_{1} \cdot (e_{2})^{*} \cdot e_{3}) + e_{4}$.
See Algorithm \ref{algorithm:convert}.
\input{algo_gsra_convert}

We now need to prove that $A_{g}$ with $n$ states ($A_{g,n}$)
and $A_{g}$ with $n-1$ states ($A_{g,n-1}$),
as constructed from $A_{g,n}$ via Algorithm \ref{algorithm:convert}, 
are equivalent, i.e.,
that $S \in \mathcal{L}(A_{g,n},v,v')  \Leftrightarrow S \in \mathcal{L}(A_{g,n-1},v,v')$.
If we can prove this,
then the last step of the recursion of Algorithm \ref{algorithm:convert} (Line \ref{line:convert:base}) will give us an automaton that is equivalent to our initial \sra\ $A$.
Moreover, the \srem\ on the single transition of this \gsra\ is obviously the desired \srem.

First, assume that $S \in \mathcal{L}(A_{g,n},v,v')$.
Then, there exists an accepting run of $A_{g,n}$
\begin{equation*}
\varrho = [1,A_{g,n}.q_{s},v_{1}] \rightarrow \cdots \rightarrow [i,q_{i},v_{i}] \rightarrow \cdots \rightarrow [l+1,A_{g,n}.q_{f},v_{k+1}]
\end{equation*}

Now, assume that $q_{i} \neq q_{rip}$ for all $q_{i}$ of $\varrho$,
including, of course, the start and final states.
We claim that this run would also be an accepting run for $A_{g,n-1}$.
To prove this,
notice that between any two successive configurations $[i,q_{i},v_{i}] \rightarrow [j,q_{j},v_{j}]$ appearing in $\varrho$,
there exists a transition that is triggered between $q_{i}$ and $q_{j}$. 
Let $e_{4}$ denote the \srem\ on this transition.
In $A_{g,n-1}$ the \srem\ on this transition would become $(e_{1} \cdot (e_{2})^{*} \cdot e_{3}) + e_{4}$.
Thus, it would also be triggered, 
due to the presence of the term $e_{4}$,
assuming the same valuation $v_{i}$.
We can inductively prove,
based on the length of the run,
that this holds for any two successive configurations.
It obviously holds for the first and second configurations,
since we start with empty registers and thus the same valuation.
As a result, $v_{2}$ would also be the same.
Inductively, we can prove the same for $v_{3}$, etc.

The other case is when $q_{i} = q_{rip}$ for some $i$ (or multiple $i$s) in $\varrho$.
We would then have in $\varrho$ a sequence of successive configurations like the following:
\begin{equation*}
[i,q_{i},v_{i}] \rightarrow [j,q_{rip},v_{j}] \rightarrow \cdots \rightarrow [\cdots,q_{rip},\cdots] \rightarrow \cdots \rightarrow [k,q_{rip},v_{k}] \rightarrow [l,q_{l},v_{l}]
\end{equation*}
We claim that,
if we remove from $\varrho$ all configurations with $q_{rip}$,
then the remaining configurations would form an accepting run of $A_{g,n-1}$.
To prove this,
notice that the transition from $q_{i}$ to $q_{rip}$ in $\varrho$ would happen through the $e_{1}$ \srem.
Then, every loop (if any) from $q_{rip}$ to itself would occur because of the $e_{2}$ \srem.
Finally, the jump from $q_{rip}$ to $q_{l}$ would happen through $e_{3}$.
Thus, the move from $q_{i}$ to $q_{l}$ would happen via $e_{1} \cdot (e_{2})^{*} \cdot e_{3}$.
But this is exactly one of the disjuncts on the transition from $q_{i}$ to $q_{l}$ in $A_{g,n-1}$.
We can therefore remove all consecutive configurations with $q_{rip}$ from $\varrho$.
If there are multiple such sequences in $\varrho$,
we can repeat the same process as many times as necessary.

We have thus proven that  $S \in \mathcal{L}(A_{g,n},v,v')  \Rightarrow S \in \mathcal{L}(A_{g,n-1},v,v')$

Conversely, assume that $S \in \mathcal{L}(A_{g,n-1},v,v')$.
There is then an accepting run $\varrho$ of $A_{g,n-1}$.
The move between each sequence of successive configurations, 
from state $q_{i}$ to $q_{j}$, 
happens either due to $e_{4}$ or due to $e_{1} \cdot (e_{2})^{*} \cdot e_{3}$.
In the former case,
we can retain this move as is in a new run for $A_{g,n}$.
In the latter case,
we can insert between the configurations of $q_{i}$ and $q_{j}$ a sequence of configurations with $q_{rip}$,
as already described previously.
Thus, $S \in \mathcal{L}(A_{g,n-1},v,v') \Rightarrow S \in \mathcal{L}(A_{g,n},v,v')$.
This completes our proof.

\end{proof}

%% file: algo_gsra_convert.tex
\begin{algorithm}[!h]
\KwIn{A \gsra\ $A$ with $n$ states.}
\KwOut{A \gsra\ $A_{g}$ with 2 states, equivalent to $A$.}
\eIf{$\lvert A.Q \rvert = 2$}{
	$\mathtt{return}\ A$\;
	\label{line:convert:base}
}
{
	Pick an element $q_{rip}$ from $A.Q$ other than $A.q_{s}$ or $A.q_{f}$\;
	$Q' \leftarrow Q - \{ q_{rip} \}$\;
	$\Delta' \leftarrow \emptyset$\;
	\tcc{{\footnotesize Assume $\delta(q_{i},q_{j})$ returns the \srem\ on the transition from $q_{i}$ to $q_{j}$.}}\
	\ForEach{$q_{i} \in Q' - \{ A.q_{f} \}$ and $q_{j} \in Q' - \{ q_{f} \}$}{
		$\delta' \leftarrow q_{i},((e_{1} \cdot (e_{2})^{*} \cdot e_{3}) + e_{4}) \rightarrow q_{j}$ for $e_{1} = \delta(q_{i},q_{rip})$, $e_{2} = \delta(q_{rip},q_{rip})$, $e_{3} = \delta(q_{rip},q_{j})$, $e_{4} = \delta(q_{i},q_{j})$\;
		$\Delta' \leftarrow \Delta' \cup \{ \delta' \}$\;
	}
	$A' \leftarrow (Q',A.q_{s},A.q_{f},A.R,\Delta')$\;
	$\mathtt{return}\ \mathtt{CONVERT}(A')$\;
}
\caption{Converting a \gsra\ with $n$ states to a \gsra\ with 2 states.}
\label{algorithm:convert}
\end{algorithm}

%% file: proofs_closure.tex
\ifdefined\FI
\begin{theorem}
\sra\ and \srem\ are closed under union, intersection, concatenation and Kleene-star.
\end{theorem}
\else
\begin{theorem*}
\sra\ and \srem\ are closed under union, intersection, concatenation and Kleene-star.
\end{theorem*}
\fi

\begin{proof}
For union, concatenation and Kleene-star the proof is essentially the proof for converting \srem\ to \sra.
For concatenation,
if we have \sra\ $A_{1}$ and $A_{2}$ we construct $A$ as in Figure \ref{fig:srem2sra:seq}.
For union,
we construct the \sra\ as in Figure \ref{fig:srem2sra:or}.
For Kleene-star, 
we construct the \sra\ as in Figure \ref{fig:srem2sra:iter}.
The only difference in these constructions is that we now assume,
without loss of generality,
that the $A_{1}.R \cap A_{2}.R = \emptyset$,
i.e., that $A_{1}$ and $A_{2}$ have different sets of registers and that the automaton $A$ constructed from $A_{1}$ and $A_{2}$ retains all registers of both $A_{1}$ and $A_{2}$.
For example, 
if we have two \sra\ $A_{1}$ and $A_{2}$ and we want to construct a \sra\ $A$ such that $\mathcal{L}(A) = \mathcal{L}(A_{1}) \cdot \mathcal{L}(A_{2})$ then we connect $A_{1}$'s final state to $A_{2}$'s start state via an $\epsilon$ transition.
It is easy to see that if $S_{1} \in \mathcal{L}(A_{1})$ and $S_{2} \in \mathcal{L}(A_{2})$ then $S = S_{1} \cdot S_{2} \in \mathcal{L}(A)$.
$S_{1}$ will force $A$ to move to $A_{1}'s$ final state 
(both $A$ and $A_{1}$ start with empty registers).
Subsequently, $A$ will jump to $A_{2}$'s start state and then $S_{2}$ will force $A$ to go to $A_{2}$'s final state which is $A$'s final state,
since $A_{2}$'s registers in $A$ are empty when $A_{2}$ starts reading $S_{2}$. 

We will now prove closure under intersection.
Let $A_{1} = (Q_{1},q_{1,s},Q_{1,f},R_{1},\Delta_{1})$ and $A_{2} = (Q_{2},q_{2,s},Q_{2,f},R_{2},\Delta_{2})$ be two \sra.
We wan to construct a \sra\ $A = (Q,q_{s},Q_{f},R,\Delta)$ such that 
$\mathcal{L}(A) = \mathcal{L}(A_{1}) \cap \mathcal{L}(A_{2})$.
We construct $A$ as follows:
\begin{itemize}
	\item $Q = Q_{1} \times Q_{2}$.
	\item $q_{s} = (q_{1,s},q_{2,s})$.
	\item $Q_{f} = (q_{1},q_{2})$, where $q_{1} \in Q_{1,f}$ and $q_{2} \in Q_{2,f}$,
	i.e., $Q_{f} = Q_{1,f} \times Q_{2,f}$.
	\item $R = R_{1} \cup R_{2}$, assuming, without loss of generality, that $R_{1} \cup R_{2} = \emptyset$.
	\item For each $q = (q_{1},q_{2}) \in Q$ we add a transition $\delta$ to $q' = (q_{1}',q_{2}') \in Q$ if there exists a transition $\delta_{1}$ from $q_{1}$ to $q_{1}'$ in $A_{1}$ and a transition $\delta_{2}$ from $q_{2}$ to $q_{2}'$ in $A_{2}$.
	The condition of $\delta$ is $\phi = \delta_{1}.\phi \wedge \delta_{2}.\phi$.
	The write registers of $\delta$ are $W = \delta_{1}.W \cup \delta_{2}.W$ 
	(notice that, 
	if $\delta_{1}.W \neq \emptyset$ and $\delta_{2}.W \neq \emptyset$, 
	this creates a multi-register \sra,
	even if $A_{1}$ and $A_{2}$ are single-register).
	Thus, $\delta = (q_{1},q_{2}),(\delta_{1}.\phi \wedge \delta_{2}.\phi) \downarrow (\delta_{1}.W \cup \delta_{2}.W) \rightarrow (q_{1}',q_{2}')$.
\end{itemize}
It is evident that, 
if a string $S$ is accepted by both $A_{1}$ and $A_{2}$,
it is also accepted by $A$.
If $A$ is not accepted either by $A_{1}$ or $A_{2}$,
then it is not accepted by $A$.
Therefore,
$\mathcal{L}(A) = \mathcal{L}(A_{1}) \cap \mathcal{L}(A_{2})$.

Since \sra\ and \srem\ are equivalent,
\srem\ are also closed under union, intersection, concatenation and Kleene-star.
\end{proof}

%% file: proofs_complement.tex
\ifdefined\FI
\begin{theorem}
\sra\ and \srem\ are not closed under complement.
\end{theorem}
\else
\begin{theorem*}
\sra\ and \srem\ are not closed under complement.
\end{theorem*}
\fi

\begin{proof}
The proof is by a counter example.
\begin{figure}[t]
\centering
\includegraphics[width=0.75\textwidth]{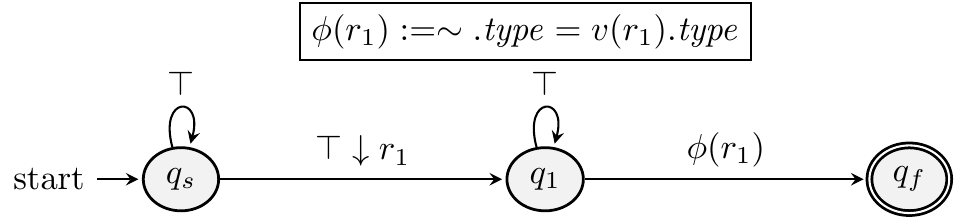}
\caption{\sra\ accepting strings which have the same type in two elements. Notice that $\sim$ denotes the current event (last event read from the string).}
\label{fig:complement}
\end{figure}
Let $A$ denote the \sra\ of Figure \ref{fig:complement}.
This \sra\ reads strings composed of tuples.
Each tuple contains an attribute called $\mathit{type}$, 
taking values from a finite or infinite alphabet.
The symbol $\sim$ simply denotes the current element of the string,
i.e., the last element read from it.
Therefore, $A$ accepts strings in which there are two elements with the same type,
regardless of the length of $S$.
Assume that there exists a \sra\ $A_{c}$ which accepts only when $A$ does not accept.
In other words, 
$A_{c}$ accepts all strings $S$ whose elements all have a different type.
Let $k = \lvert A_{c}.R \rvert$ be the number of registers of $A_{c}$.
Let $\lvert S \rvert =  k+m$,
where $m > 1$,
be the length of a string $S$ whose elements all have different types.
However, $A_{c}$ cannot possibly exist.
At the end of $S$,
as $A_{c}$ is ready to read the last element of $S$,
it must have stored all of the previous $k+m-1$ elements of $S$.
But $A$ has only $k$ registers,
whereas $k+m-1 > k$,
since $m>1$.
Thus, $A_{c}$ cannot exist.
\end{proof}

%% file: proofs_determinization.tex
\ifdefined\FI
\begin{theorem}
\sra\ are not closed under determinization.
\end{theorem}
\else
\begin{theorem*}
\sra\ are not closed under determinization.
\end{theorem*}
\fi

\begin{proof}
The proof is again by a counter example.
\begin{figure}[t]
\centering
\includegraphics[width=0.75\textwidth]{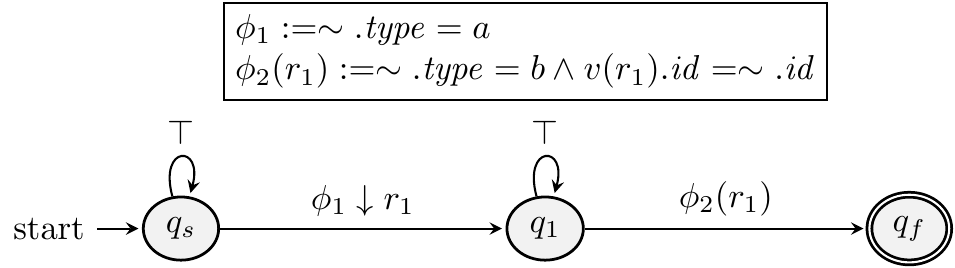}
\caption{\sra\ accepting all strings containing an $a$ element followed by a $b$ element, whose identifiers are the same.}
\label{fig:determinization}
\end{figure}
Let $A$ denote the \sra\ of Figure \ref{fig:determinization}.
This \sra\ reads strings composed of tuples.
Each tuple contains an attribute, 
called $\mathit{type}$, 
taking values from a finite or infinite alphabet.
It also contains another tuple, 
called $\mathit{id}$,
taking integer values.
$A$ thus accepts strings $S$ that contain an $a$ followed by a $b$,
whose ids are equal,
regardless of the length of $S$. 

Assume there exist a \dsra\ $A_{d}$ with $k$ registers which is equivalent to $A$.
Let 
\begin{equation*}
S = (a, 1) (b, 2) 
\end{equation*}
be a string given to $A_{d}$.
After reading $S_{1}=(a,1)$,
$A_{d}$ must store it in a register $r_{1}$ in order to be able to compare it when $(b,2)$ arrives.
Let 
\begin{equation*}
S' = (a, 1) (a, 3) (b, 2) 
\end{equation*}
After reading $S_{1}'=(a,1)$,
$A_{d}$ must store it in the register $r_{1}$,
since $A_{d}$ is deterministic and follows a single run.
Thus, it must have the exact same behavior after reading $s_{1}$ and $S_{1}'$.
But we must also store $S_{2}'=(a,3)$ after reading it.
Additionally,
$S_{2}'$ must be stored in a different register $r_{2}$.
We cannot overwrite $r_{1}$.
If we did this and $S_{1}'$ were $(a,2)$,
then we would not be able to match $(a,2)$ to $S_{3}'=(b,2)$ and $S'=(a,2)(a,3)(b,2)$ would not be accepted.
Now, let
\begin{equation*}
S'' = \underbrace{(a, \cdots) (a, \cdots) \cdots (a, \cdots)}_{k+1 \text{ elements}}  (b, 2) 
\end{equation*}
With a similar reasoning,
all of the first $k+1$ elements of $S''$ must be stored after reading them.
But this is a contradiction,
as $A_{d}$ can store at most $k$ different elements.
Therefore, there does not exist a \dsra\ which is equivalent to $A$.
\end{proof}

%% file: proofs_sra_properties.tex
In the proofs that follow,
we will need to refer to some structural properties of \sra.
We present them here.
Without loss of generality,
we assume that each state of a \sra\ is accessible from its start state.
Inaccessible states can always be removed without affecting the behavior of an automaton.
Therefore, a \sra,
in terms of its structure,
can be viewed as a weakly connected directed graph.
The usual notions of walks and trails from graph theory also apply for a \sra.
However,
since we are interested in walks and trails from its start state,
and in order to avoid introducing new notation and terminology,
in what follows,
we will stick to the already introduced terms.
We will talk about states, instead of nodes/vertices,
and about transitions, 
instead of edges.

\begin{definition}[Walk over \sra]
A walk $w$ over a \sra\ $A$ is a sequence of transitions \linebreak $w=<\delta_{1},\cdots,\delta_{k}>$,
such that:
\begin{itemize}
	\item $\forall \delta_{i}\ \delta_{i} \in A.\Delta$
	\item $\delta_{1}.\mathit{source} = A.q_{s}$
	\item $\forall \delta_{i},\delta_{i+1}\ \delta_{i}.\mathit{target}=\delta_{i+1}.\mathit{source}$
\end{itemize}
We say that such a walk is of length $k$.
By $W_{A}$ we denote the set of all walks over $A$ and by $W_{A,q}$, we denote the set of walks over $A$ that end in state $q$, i.e.,
$W_{A,q}=\{w: w \in W_{A} \wedge \delta_{k}.\mathit{target}=q \}$.
\end{definition}

\begin{definition}[Trail over \sra]
A trail $t$ over a \sra\ $A$ is a walk $w=<\delta_{1},\cdots,\delta_{k}>$ over $A$, such that:
\begin{itemize}
	\item $\forall \delta_{i},\delta_{j}\ \delta_{i}.\mathit{source} \neq \delta_{j}.\mathit{source}$
	\item $\forall \delta_{i}\ \delta_{i}.\mathit{source} \neq \delta_{i}.\mathit{target}$
\end{itemize}
If $T_{A}$ is the set of all trails over $A$,
$T_{A,q}$ is the set of all trails ending in state $q$, i.e.,
$T_{A,q}=\{t: t \in T_{A} \wedge \delta_{k}.\mathit{target}=q \}$.
\end{definition}
In other words, a trail is a walk without state revisits (and, as a consequence, without transition revisits).

\begin{proposition}[Every walk contains a trail]
\label{proposition:trails_in_walks}
For every walk to a non-start state $q$ ($w=<\delta_{1},\cdots,\delta_{k}> \in W_{A,q}$, $q \in A.Q \setminus \{A.q_{s}\}$),
there exists a trail to $q$ ($t=<\delta_{1}^{'},\cdots,\delta_{l}^{'}> \in T_{A,q}$)
such that all transitions of the trail ($\delta_{i}^{'}$) appear in the walk $w$ in the same order as in the trail, i.e.,
$w$ can be written as $w=<\cdots,\delta_{1}^{'},\cdots,\delta_{l}^{'},\cdots>$.
We say that $t$ is contained in $w$.
\end{proposition}
\input{proofs_proposition_trails_in_walks}

\begin{definition}[Register appearance in a trail]
We say that a register $r$ appears in a trail if there exists at least one transition $\delta$ in the trail such that $r \in \delta.W$.
\end{definition}
In other words,
a trail must write to $r$ to say that it appears in it.

Notice that a run of $A$ over a stream $S$ induces a walk over $A$.
\begin{definition}[Walk induced by a run]
If 
\begin{equation*}
\varrho=[1,q_{s},v_{1}] \overset{\delta_{1}}{\rightarrow} [2,q_{2},v_{2}] \overset{\delta_{2}}{\rightarrow} \cdots \overset{\delta_{n-1}}{\rightarrow} [n,q_{n},v_{n}]
\end{equation*}
is a run of a \sra\ $A$,
then
\begin{equation*}
w_{\varrho} = <\delta_{1},\cdots,\delta_{n-1}>
\end{equation*}
is the walk induced by $\varrho$.
\end{definition}

%% file: proofs_proposition_trails_in_walks.tex
\begin{proof}
The proof is by induction on the length of the walk.
The proposition trivially holds for walks of length $k=1$.
For walks of length $k+1$,
if $w$ is already a trail, then the proposition holds for $k+1$.
If $w$ is not a trail, then a state is visited at least twice.
Removing all transitions between these two visits results in a walk
for which, by the induction hypothesis,
the proposition already holds.
Therefore, it holds for the complete walk too.

We give a detailed proof by induction on the length of the walk.
\begin{itemize}
	\item Base case for $k=1$. 
	If $w=<\delta_{1}>$ is a walk to $q$, then, $t=<\delta_{1}>$ is also a trail to $q$,
	since $q \neq q^{s}$.
	\item Assume 
	\begin{equation*}
	w = q_{1} \overset{\delta_{1}}{\rightarrow} q_{2} \cdots q_{i} \overset{\delta_{i}}{\rightarrow} q_{i+1} \cdots
	q_{j} \overset{\delta_{j}}{\rightarrow} q_{j+1} \cdots q_{k} \overset{\delta_{k}}{\rightarrow} q_{k+1} \overset{\delta_{k+1}}{\rightarrow} q_{k+2}
	\end{equation*}		
	is a walk of length $k+1$ (where we use a slightly different notation to explicitly show the visited states).
	\begin{itemize}
		\item If $w$ is already a trail, then the proposition holds for $k+1$.
		\item If $w$ is not a trail, then a state is visited at least twice.
		Assume that it is $q_{i}$, visited again as $q_{j}$, i.e., $q_{i}=q_{j}$.
		Remove from $w$ all transitions $\delta_{l}$, $i \leq l < j$.
		Then we get 
		\begin{equation*}
		w' = q_{1} \overset{\delta_{1}}{\rightarrow} q_{2} \cdots q_{i} \overset{\delta_{j}}{\rightarrow} q_{j+1} \cdots q_{k} \overset{\delta_{k}}{\rightarrow} q_{k+1} \overset{\delta_{k+1}}{\rightarrow} q_{k+2}
		\end{equation*}
		or, equivalently, since $q_{i}=q_{j}$
		\begin{equation*}
		w' = q_{1} \overset{\delta_{1}}{\rightarrow} q_{2} \cdots q_{j} \overset{\delta_{j}}{\rightarrow} q_{j+1} \cdots q_{k} \overset{\delta_{k}}{\rightarrow} q_{k+1} \overset{\delta_{k+1}}{\rightarrow} q_{k+2}
		\end{equation*}
		Notice that $w'$ is indeed a walk, since all its transitions are valid,
		including the one that stitches together the two sub-walks ($q_{j} \overset{\delta_{j}}{\rightarrow} q_{j+1}$).
		Moreover, its length is at most $k$, since we removed at least one transition.
		Therefore, by the induction hypothesis, there exists a trail $t'$ to $q_{k+2}$,
		contained in $w'$. 
		But $t'$ is also contained in $w$. 
		Therefore, $t'$ is a trail to $q_{k+2}$ contained in $w$ and the proposition holds for walks of length $k+1$ as well.
	\end{itemize}
\end{itemize}
\end{proof}

%% file: proofs_windowed_srem.tex
\ifdefined\FI
\begin{theorem}
For every windowed \srem\ there exists an equivalent unrolled \sra\ without any loops, 
i.e., a \sra\ where each state may be visited at most once.
\end{theorem}
\else
\begin{theorem*}
For every windowed \srem\ there exists an equivalent unrolled \sra\ without any loops, 
i.e., a \sra\ where each state may be visited at most once.
\end{theorem*}
\fi

\begin{proof}

Let $e_{w} := e^{[1..w]}$.
Algorithm \ref{algorithm:wsrem2sra} shows how we can construct $A_{e_{w}}$.
The basic idea is that we first construct as usual the \sra\ $A_{e}$ for the sub-expression $e$
(and eliminate $\epsilon$-transitions).
We can then use $A_{e}$ to enumerate all the possible walks of $A_{e}$ of length up to $w$ and then join them in a single \sra\ through disjunction.
Essentially,
we need to remove cycles from every walk of $A_{e}$ by ``unrolling'' them as many times as necessary,
without the length of the walk exceeding $w$.
This ``unrolling'' operation is performed by the (recursive) Algorithm \ref{algorithm:unroll_cycles}.
Because of this ``unrolling'',
a state of $A_{e}$ may appear multiple times as a state in $A_{e_{w}}$.
We keep track of which states of $A_{e_{w}}$ correspond to states of $A_{e}$ through the function $\mathit{CopyOfQ}$ in the algorithm.
For example, if $q_{e}$ is a state of $A_{e}$, $q_{e_{w}}$ a state of $A_{e_{w}}$ and
$\mathit{CopyOfQ}(q_{e_{w}}) = q_{e}$,
this means that $q_{e_{w}}$ was created as a copy of $q_{e}$
(and multiple states of $A_{e_{w}}$ may be copies of the same state of $A_{e}$,
i.e., $\mathit{CopyOfQ}$ is a surjective but not an injective function).
We do the same for the registers as well, 
through the function $\mathit{CopyOfR}$.
The algorithm avoids an explicit enumeration,
by gradually building the automaton as needed,
through an incremental expansion.
Of course, walks that do not end in a final state may be removed,
either after the construction or online,
whenever a non-final state cannot be expanded.

\input{algo_unroll_cycles}

The lemma is a direct consequence of the construction algorithm.
First, note that,
by the construction algorithm,
there is a one-to-one mapping (bijective function) between the walks/runs of $A_{e_{w}}$
and the walks/runs of $A_{e}$ of length up to $w$.
We can show that if $\varrho_{e}$ is a run of $A_{e}$ of length up to $w$ over a string $S$
($\varrho_{e}$ has at most $w$ transitions),
then the corresponding run $\varrho_{e_{w}}$ of $A_{e_{w}}$ is indeed a run and if $\varrho_{e}$ is accepting so is $\varrho_{e_{w}}$.
By definition, 
since the runs have no $\epsilon$-transitions and are at most of length $w$,
$\lvert S \rvert \leq w$.

We first prove the following proposition:

\ifdefined\FI
\begin{proposition}
There exists a run of $A_{e}$ over a string $S$ of length up to $w$
\begin{equation*}
\varrho_{e}=[1,q_{e,1}=A_{e}.q_{s},v_{e,1}] \overset{\delta_{e,1}}{\rightarrow}  \cdots \overset{\delta_{e,i-1}}{\rightarrow} [n,q_{e,i},v_{e,i}] \overset{\delta_{e,i}}{\rightarrow} \cdots \overset{\delta_{e,n-1}}{\rightarrow} [n,q_{e,n},v_{e,n}]
\end{equation*}
iff
there exists a run $\varrho_{e_{w}}$ of $A_{e_{w}}$ 
\begin{equation*}
\varrho_{e_{w}}=[1,q_{e_{w},1}=A_{e_{w}}.q_{s},v_{e_{w},1}] \overset{\delta_{e_{w},1}}{\rightarrow}  \cdots \overset{\delta_{e_{w},i-1}}{\rightarrow} [n,q_{e_{w},i},v_{e_{w},i}] \overset{\delta_{e_{w},i}}{\rightarrow} \cdots \overset{\delta_{e_{w},n-1}}{\rightarrow} [n,q_{e_{w},n},v_{e_{w},n}]
\end{equation*}
such that:
\begin{itemize}
	\item $\mathit{CopyOfQ}(q_{e_{w},i}) = q_{e,i}$
	\item $v_{e,i}(r_{e})=v_{e_{w},i}(r_{e_{w}})$, 
	if
	$\mathit{CopyOfR}(r_{e_{w}})=r_{e}$
	and
	$r_{e_{w}}$ appears last among the registers that are copies of $r_{e}$ in $\varrho_{e_{w}}$.
\end{itemize}
\end{proposition}
\else
\begin{proposition*}
There exists a run of $A_{e}$ over a string $S$ of length up to $w$
\begin{equation*}
\varrho_{e}=[1,q_{e,1}=A_{e}.q_{s},v_{e,1}] \overset{\delta_{e,1}}{\rightarrow}  \cdots \overset{\delta_{e,i-1}}{\rightarrow} [n,q_{e,i},v_{e,i}] \overset{\delta_{e,i}}{\rightarrow} \cdots \overset{\delta_{e,n-1}}{\rightarrow} [n,q_{e,n},v_{e,n}]
\end{equation*}
iff
there exists a run $\varrho_{e_{w}}$ of $A_{e_{w}}$ 
\begin{equation*}
\varrho_{e_{w}}=[1,q_{e_{w},1}=A_{e_{w}}.q_{s},v_{e_{w},1}] \overset{\delta_{e_{w},1}}{\rightarrow}  \cdots \overset{\delta_{e_{w},i-1}}{\rightarrow} [n,q_{e_{w},i},v_{e_{w},i}] \overset{\delta_{e_{w},i}}{\rightarrow} \cdots \overset{\delta_{e_{w},n-1}}{\rightarrow} [n,q_{e_{w},n},v_{e_{w},n}]
\end{equation*}
such that:
\begin{itemize}
	\item $\mathit{CopyOfQ}(q_{e_{w},i}) = q_{e,i}$
	\item $v_{e,i}(r_{e})=v_{e_{w},i}(r_{e_{w}})$, 
	if
	$\mathit{CopyOfR}(r_{e_{w}})=r_{e}$
	and
	$r_{e_{w}}$ appears last among the registers that are copies of $r_{e}$ in $\varrho_{e_{w}}$.
\end{itemize}
\end{proposition*}
\fi

We say that a register $r$ appears in a run at position $i$ if $r \in \delta_{i}.W$,
i.e.,
if the $i^{th}$ transition writes to $r$.
We say that a register $r_{e_{w}}$,
where $\mathit{CopyOfR}(r_{e_{w}})=r_{e}$, 
appears last if no other copies of $r_{e}$ appear after $r_{e_{w}}$ in a run.
The notion of a register's (last) appearance also applies for walks of $A_{e_{w}}$,
since $A_{e_{w}}$ is a directed acyclic graph,
as can be seen by Algorithms \ref{algorithm:unroll_cycles0} and \ref{algorithm:unroll_cyclesk}
(they always expand ``forward'' the \sra, 
without creating any cycles and without converging any paths).

\begin{proof}
The proof is by induction on the length of the runs $k$, with $k \leq w$.
We prove only one direction (assume a run $\varrho_{e}$ exists).
The other is similar.

\textbf{Base case: $k=0$.}
For both \sra, 
only the start state and the initial configuration with all registers empty is possible.
Thus, $v_{e,i}=v_{e_{w},i}=\sharp$ for all registers. 
By Algorithm \ref{algorithm:unroll_cycles0} (line \ref{line:unroll_cycles:qs_descendent}),
we know that $\mathit{CopyOf}(q_{e_{w},s}) = q_{e,s}$.

\textbf{Case for $0 < k+1 \leq w$, assuming the proposition holds for $k$.}
Let
\begin{equation*}
\varrho_{e,k+1} = \cdots [k,q_{e,k},v_{e,k}] \overset{\delta_{e,k}}{\rightarrow} [k+1,q_{e,k+1},v_{e,k+1}]
\end{equation*}
and
\begin{equation*}
\varrho_{e_{w},k+1} = \cdots [k,q_{e_{w},k},v_{e_{w},k}] \overset{\delta_{e_{w},k}}{\rightarrow} [k+1,q_{e_{w},k+1},v_{e_{w},k+1}]
\end{equation*}
be the runs of $A_{e}$ and $A_{e_{w}}$ respectively of length $k+1$ over the same $k+1$ elements of a string $S$.
We know that $\varrho_{e,k+1}$ is an actual run and we need to construct $\varrho_{e_{w},k+1}$, 
knowing, by the induction hypothesis,
that there is an actual run up to $q_{e_{w},i+k}$.
Now, by the construction algorithm,
we can see that if $\delta_{e,k}$ is a transition of $A_{e}$ from $q_{e,k}$ to $q_{e,k+1}$,
there exists a transition $\delta_{e_{w},k}$ with the same condition
from $q_{e_{w},k}$ to a $q_{e_{w},k+1}$ such that $\mathit{CopyOfQ}(q_{e_{w},k+1})=q_{e,k+1}$.
Moreover, if $\delta_{e,k}$ is triggered,
so does $\delta_{e_{w},k}$,
because the registers in the register selection of $\delta_{e_{w},k}$
are copies of the corresponding registers in $\delta_{e,k}.\phi.rs$.
By the induction hypothesis,
we know that the contents of the registers in $\delta_{e,k}.\phi.rs$ will be equal to the contents of their corresponding registers in $\varrho_{e_{w}}$ that appear last.
But these are exactly the registers in $\delta_{e_{w},k}.\phi.rs$
(see line \ref{line:unroll_cycles:latest_appearance} in Algorithm \ref{algorithm:unroll_cyclesk}).
We can also see that the part of the proposition concerning the valuations $v$ also holds.
If $\delta_{e,k}.W = \{ r_{e} \}$ and $\delta_{e_{w},k}.W = \{ r_{e_{w}} \}$,
then we know,
by the construction algorithm
(line \ref{line:unroll_cycles:rn_descendent}),
that $\mathit{CopyOfR}(r_{e_{w}}) = r_{e}$ and $r_{e_{w}}$ will be the last appearance of a copy of $r_{e}$ in $\varrho_{e_{w},k+1}$.
Thus the proposition holds for $0 < k+1 \leq w$ as well.
\end{proof}

The above proposition must necessarily hold for accepting runs as well.
Therefore,
$A_{e}$ accepts the same language as $A_{e_{w}}$.
\end{proof}

We also note that $w$ must be a number greater than (or equal to) 
the minimum length of the walks induced by the accepting runs of $A_{e}$
(which is something that can be computed by the structure of the expression).
Although this is not a formal requirement,
if it is not satisfied,
then $A_{e_{w}}$ won't detect any matches.

%% file: algo_unroll_cycles.tex
\begin{algorithm}
\SetAlgoNoLine
\KwIn{Windowed \srem\ $e' := e^{[1..w]}$}
\KwOut{\sra\ $A_{e'}$ equivalent to $e'$}
$A_{e,\epsilon} \leftarrow ConstructSRA(e)$; \tcp{{\footnotesize As described in Appendix \ref{sec:proof:srem2sra}.}}\
$A_{e,ms} \leftarrow \mathit{EliminateEpsilon}(A_{e,\epsilon})$; \tcp{{\footnotesize See Algorithm \ref{algorithm:epsilon}. $A_{e,ms}$ might be multi-register.}}\
$A_{e} \leftarrow \mathit{ConvertToSingleRegister}(A_{e,ms})$; \tcp{{\footnotesize As described in Appendix \ref{sec:proof:multi2single}.}}\
$A_{e'} \leftarrow Unroll(A_{e},w)$; \tcp{{\footnotesize See Algorithm \ref{algorithm:unroll_cycles}.}}\ 
$\mathtt{return}\ A_{e'}$\;
\caption{Constructing \sra\ for a windowed \srem\ ($\mathit{ConstructWSRA}$).}
\label{algorithm:wsrem2sra}
\end{algorithm}

\begin{algorithm}
\SetAlgoNoLine
\KwIn{\sra\ $A$ and integer $k \geq 0$}
\KwOut{\sra\ $A_{k}$ with runs of length up to $k$}
\eIf{$k=0$}{
	$(A_{k},\mathit{Frontier},\mathit{CopyOfQ},\mathit{CopyOfR}) \leftarrow \mathit{Unroll0}(A)$;
	\tcp{{\footnotesize Algorithm \ref{algorithm:unroll_cycles0}}}
}
{
	$(A_{k},\mathit{Frontier},\mathit{CopyOfQ},\mathit{CopyOfR}) \leftarrow \mathit{UnrollK}(A,k)$;
	\tcp{{\footnotesize Algorithm \ref{algorithm:unroll_cyclesk}}}
}
$\mathtt{return}\ (A_{k},\mathit{Frontier},\mathit{CopyOfQ},\mathit{CopyOfR})$\;
\caption{Unrolling cycles for windowed \srem\ ($\mathit{Unroll}$).}
\label{algorithm:unroll_cycles}
\end{algorithm}

\begin{algorithm}
\SetAlgoNoLine
\KwIn{\sra\ $A$}
\KwOut{\sra\ $A_{0}$ with runs of length 0}
$q \leftarrow \mathit{CreateNewState}()$\;
$\mathit{CopyOfQ} \leftarrow \{q \rightarrow A.q_{s}\}$\; \label{line:unroll_cycles:qs_descendent}
$\mathit{CopyOfR} \leftarrow \emptyset$\;
$\mathit{Frontier} \leftarrow \{q\}$\;
$Q_{f} \leftarrow \emptyset$\;
\If{$A.q_{s} \in A.Q_{f}$}{
	$Q_{f} \leftarrow Q_{f} \cup \{q\}$\;
}
$A_{0} \leftarrow (\{q\},q,Q_{f},\emptyset,\emptyset)$\;
$\mathtt{return}\ (A_{0},\mathit{Frontier},\mathit{CopyOfQ},\mathit{CopyOfR})$\;
\caption{Unrolling cycles for windowed \srem, base case ($\mathit{Unroll0}$).}
\label{algorithm:unroll_cycles0}
\end{algorithm}

\begin{algorithm}
\KwIn{\sra\ $A$ and integer $k > 0$}
\KwOut{\sra\ $A_{k}$ with runs of length up to $k$}
$(A_{k-1},\mathit{Frontier},\mathit{CopyOfQ},\mathit{CopyOfR}) \leftarrow \mathit{Unroll}(A,k-1)$\;
$\mathit{NextFrontier} \leftarrow \emptyset$\;
$Q_{k} \leftarrow A_{k-1}.Q$; $Q_{k,f} \leftarrow A_{k-1}.Q_{f}$;
$R_{k} \leftarrow A_{k-1}.R$; $\Delta_{k} \leftarrow A_{k-1}.\Delta$\;
\ForEach{$q \in \mathit{Frontier}$}{
	$q_{c} \leftarrow \mathit{CopyOfQ}(q)$\;
	\ForEach{$\delta \in A.\Delta: \delta.\mathit{source} = q_{c}$}{
		$q_{new} \leftarrow \mathit{CreateNewState}()$\;
		$Q_{k} \leftarrow Q_{k} \cup \{q_{new}\}$\;
		$\mathit{CopyOfQ} \leftarrow \mathit{CopyOfQ} \cup \{ q_{new} \rightarrow \delta.\mathit{target} \}$\;
		\If{$\delta.\mathit{target} \in A.Q_{f}$}{
			$Q_{k,f} \leftarrow Q_{k,f} \cup \{q_{new}\}$\;
		}	
		\eIf{$\delta.W = \emptyset$}{
			$R_{new} \leftarrow \emptyset$\;	
		}
		{
			$r_{new} \leftarrow \mathit{CreateNewRegister}()$\;
			$R_{k} \leftarrow R_{k} \cup \{ r_{new} \}$\;
			$R_{new} \leftarrow \{ r_{new} \}$\;
			$\mathit{CopyOfR} \leftarrow \mathit{CopyOfR} \cup \{r_{new} \rightarrow \delta.r \}$; \label{line:unroll_cycles:rn_descendent}
			\tcp{{\footnotesize $\delta.r$ single element of $\delta.W$}}
		}		
		$\phi_{new} \leftarrow \delta.\phi$\;
		$rs_{new} \leftarrow ()$\;
		\tcc{{\footnotesize By $\delta.\phi.rs$ we denote the register selection of $\delta.\phi$, i.e., all the registers referenced by $\delta.\phi$ in its arguments. $rs$ is represented as a list.}}\
		\ForEach{$r \in \delta.\phi.rs$}{
			\tcc{{\footnotesize $\mathit{FindLastAppearance}$ returns a register that is a copy of $r$ and appears last in the trail of $A_{k-1}$ to $q$ (no other copies of $r$ appear after $r_{latest}$). Due to the construction, only a single walk/trail to $q$ exists.}}\
			$r_{latest} \leftarrow \mathit{FindLastAppearance}(r,q,A_{k-1})$\; \label{line:unroll_cycles:latest_appearance}
			\tcc{{\footnotesize $::$ denotes the operation of appending an element at the end of a list. $r_{latest}$ is appended at the end of $rs_{new}$.}}\
			$rs_{new} \leftarrow rs_{new} :: r_{latest}$\;
		}
		$\delta_{new} \leftarrow \mathit{CreateNewTransition}(q,\phi_{new}(rs_{new}) \downarrow R_{new} \rightarrow q_{new})$\;
		$\Delta_{k} \leftarrow \Delta_{k} \cup \{ \delta_{new} \}$\;
		$\mathit{NextFrontier} \leftarrow \mathit{NextFrontier} \cup \{q_{new}\}$\;
	}
}
$A_{k} \leftarrow (Q_{k}, A_{k-1}.q_{s}, Q_{k,f}, R_{k}, \Delta_{k})$\;
$\mathtt{return}\ (A_{k},\mathit{NextFrontier},\mathit{CopyOfQ},\mathit{CopyOfR})$\;
\caption{Unrolling cycles for windowed expressions, $k > 0$ ($\mathit{UnrollK}$).}
\label{algorithm:unroll_cyclesk}
\end{algorithm}

%% file: proofs_wsrem2dsra.tex
\ifdefined\FI
\begin{theorem}
For every windowed \srem\ there exists an equivalent deterministic \sra.
\end{theorem}
\else
\begin{theorem*}
For every windowed \srem\ there exists an equivalent deterministic \sra.
\end{theorem*}
\fi

\begin{proof}
The process for constructing a deterministic \sra\ (\dsra) from a windowed \srem\ is shown in Algorithm \ref{algorithm:determinization}.
It first constructs a non-deterministic \sra\ (\nsra) and then uses the power set of this \nsra's states to construct the \dsra.
For each state $q_{d}$ of the \dsra,
it gathers all the conditions from the outgoing transitions of the states of the \nsra\ $q_{n}$ ($q_{n} \in q_{d}$),
it creates the (mutually exclusive) \emph{minterms} of these conditions, 
i.e., the set of maximal satisfiable Boolean combinations of the conditions.
It then creates transitions, 
based on these minterms.
Please, note that we use the ability of a transition to write to more than one registers.
So, from now on,
$\delta.W$ will be a set that is not necessarily a singleton.
This allows us to retain the same set of registers, i.e.,
the set of registers $R$ will be the same for the \nsra\ and the \dsra.
A new transition created for the \dsra\ may write to multiple registers,
if it ``encodes'' multiple transitions of the \nsra,
which may write to different registers.
It is also obvious that the resulting \sra\ is deterministic,
since the various minterms out of every state are mutually exclusive,
i.e., at most one may be triggered.
Intuitively,
having a windowed \sra\ allows us to construct a deterministic \sra\ with as many registers as necessary.
Therefore,
it is always possible to have available all past $w$ elements.
This is not possible in the counter-example of Section \ref{sec:proof:determinization},
where we showed that \sra\ are not in general determinizable.

First, we will prove the following proposition:
\begin{proposition}
There exists a run $\varrho_{n}$ over a string $S$ which $A_{n}$ can follow by reading the first $k$ tuples of $S$,
iff there exists a run $\varrho_{d}$ 
that $A_{d}$ can follow by reading the same first $k$ tuples,
such that, if
\begin{equation*}
\varrho_{n} = [1,q_{n,1},v_{n,1}] \overset{\delta_{n,1}}{\rightarrow} \cdots \overset{\delta_{n,i-1}}{\rightarrow} [i,q_{n,k},v_{n,i}] \overset{\delta_{n,i}}{\rightarrow} \cdots \overset{\delta_{n,k-1}}{\rightarrow} [k,q_{n,k},v_{n,k}]
\end{equation*}
and
\begin{equation*}
\varrho_{d} = [1,q_{d,1},v_{d,1}] \overset{\delta_{d,1}}{\rightarrow}  \cdots \overset{\delta_{d,i-1}}{\rightarrow} [i,q_{d,i},v_{d,i}] \overset{\delta_{d,i}}{\rightarrow} \cdots \overset{\delta_{d,k-1}}{\rightarrow} [k,q_{d,k},v_{d,k}]
\end{equation*}
are the runs of $A_{n}$ and $A_{d}$ respectively, then,
\begin{itemize}
	\item $q_{n,i} \in q_{d,i}\ \forall i: 1 \leq i \leq k$
	\item if $r \in A_{d}.R$ appears in $\varrho_{n}$, then it appears in $\varrho_{d}$
	\item $v_{n,i}(r) = v_{d,i}(r)$  for every $r$ that appears in $\varrho_{n}$ (and $\varrho_{d})$
\end{itemize}
\end{proposition}

We say that a register $r$ appears in a run at position $i$ if $r \in \delta_{i}.W$.

\input{algo_determinization}

\begin{proof}
We will prove only direction (the other is similar).
Assume there exists a run $\varrho_{n}$.
We will prove that there exists a run $\varrho_{d}$ by induction on the length $k$ of the run.

\textbf{Base case: $k=0$.}
Then $\varrho_{n}=[1,q_{n,1},\sharp]=[1,q_{n,s},\sharp]$.
The run $\varrho_{d}=[1,q_{d,s},\sharp]$ is indeed a run of the \dsra\
that satisfies the proposition,
since $q_{n,s} \in q_{d,s} = \{q_{n,s}\}$ 
(by the construction algorithm, line \ref{line:determinization:start_state}),
all registers are empty
and no registers appear in the runs. 

\textbf{Case $k>0$.}
Assume the proposition holds for $k$.
We will prove it holds for $k+1$ as well.
Let 
\begin{equation}
\label{run:n}
\varrho_{n,k+1} = \cdots [k,q_{n,k},v_{n,k}] \begin{cases}
\overset{\delta_{n,k}^{1}}{\rightarrow} [k+1,q_{n,k+1}^{1},v_{n,k+1}^{1}] \\
\overset{\delta_{n,k}^{2}}{\rightarrow} [k+1,q_{n,k+1}^{2},v_{n,k+1}^{2}] \\
\cdots \\
\overset{\delta_{n,k}^{m}}{\rightarrow} [k+1,q_{n,k+1}^{m},v_{n,k+1}^{m}]
\end{cases}
\end{equation}
be the possible runs that can follow a run $\varrho_{n,k}$ after the \nsra\ reads the $(k+1)^{th}$ tuple.
Notice that,
typically,
since $A_{n}$ is non-deterministic,
there might be multiple runs $\varrho_{n,k}$ and each such run can spawn its own multiple runs $\varrho_{n,k+1}$.
The same reasoning that we present below applies to all these $\varrho_{n,k}$.

We need to find a run of the \dsra\ like:
\begin{equation*}
\varrho_{d,k+1} = \cdots [k,q_{d,k},v_{d,k}] \overset{\delta_{d,k}}{\rightarrow} [k+1,q_{d,k+1},v_{d,k+1}]
\end{equation*}

\begin{figure}[t]
\centering
\begin{subfigure}[t]{0.25\textwidth}
	\includegraphics[width=0.99\textwidth]{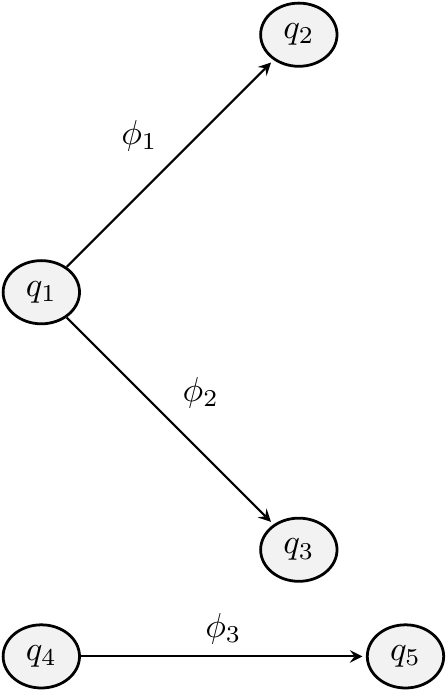}
	\caption{\nsra.}
	\label{fig:determinization_example:nsra}
\end{subfigure}
\begin{subfigure}[t]{0.73\textwidth}
	\includegraphics[width=0.99\textwidth]{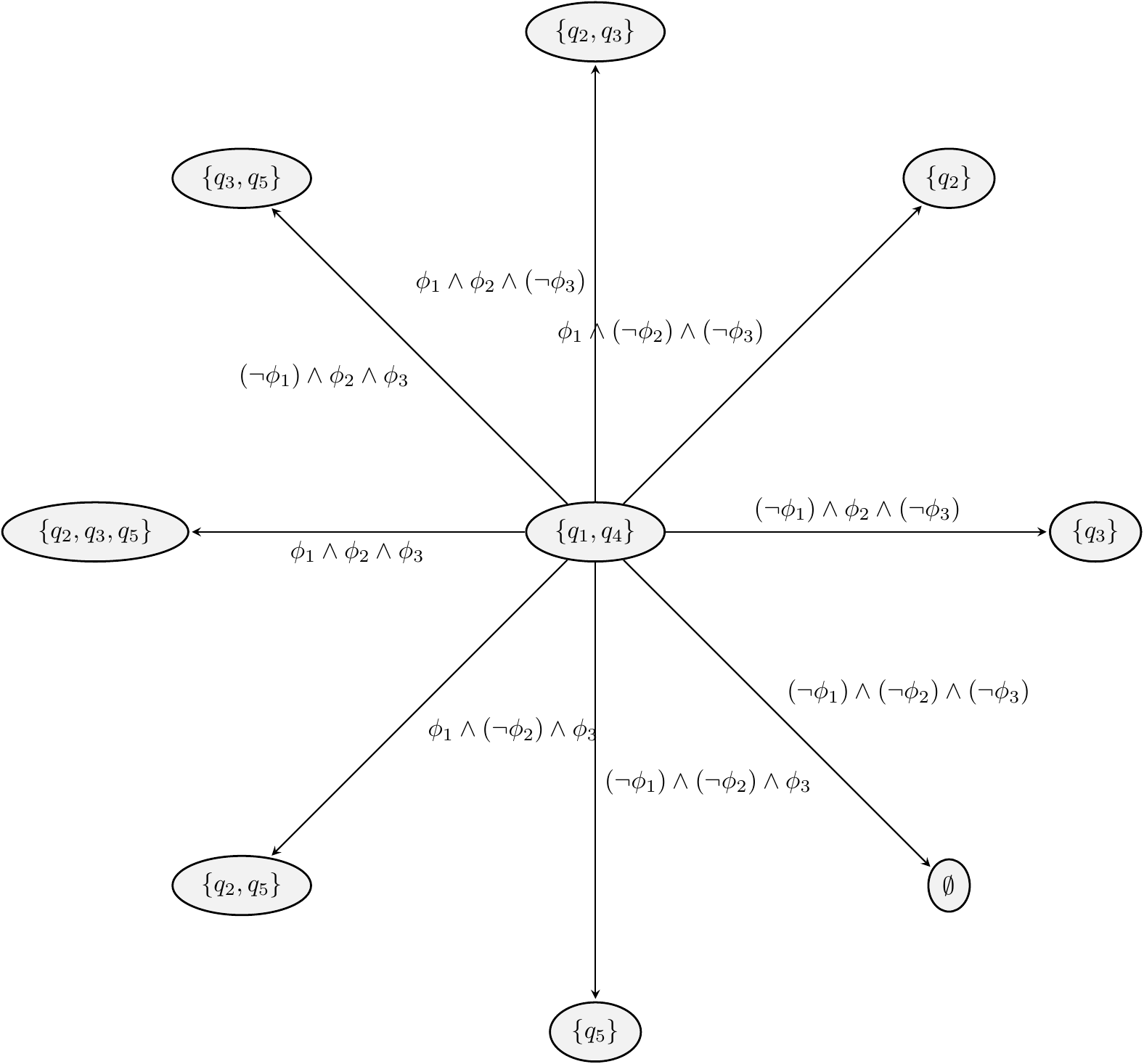}
	\caption{\dsra.}
	\label{fig:determinization_example:dsra}
\end{subfigure}
\caption{Example of converting a \nsra\ to a \dsra.}
\label{fig:determinization_example}
\end{figure}

By the induction hypothesis,
we know that $q_{n,k} \in q_{d,k}$.
By the construction Algorithm \ref{algorithm:determinization},
we then know that,
if $\phi_{n,k}^{j}=\delta_{n,k}^{j}.\phi$ is the condition of a transition that takes the non-deterministic run to $q_{n,k+1}^{j}$,
then there exists a transition $\delta_{d,k}$ in the \dsra\ from $q_{d,k}$ whose condition will be a minterm,
containing all the $\phi_{n,k}$ in their positive form
and all other possible conditions in their negated form.
Moreover, 
the target of that transition, 
$q_{d,k+1}$, 
contains all $q_{n,k+1}^{j}$.
More formally,  $q_{d,k+1} = \bigcup\limits_{j=1}^{m}{q_{n,k+1}^{j}}$.

As an example,
see Figure \ref{fig:determinization_example}.
Figure \ref{fig:determinization_example:nsra} depicts part of a \nsra.
Figure \ref{fig:determinization_example:dsra} depicts part of the \dsra\ that woyld be constructed from that of Figure \ref{fig:determinization_example:nsra}.
The construction algorithm would create the state $\{q_{2},q_{4}\}$,
the minterms from the conditions of all the outgoing transitions of $q_{2}$ and $q_{4}$ and then attempt to determine which minterm would move the \dsra\ to which subset of $\{q_{2},q_{3},q_{5}\}$.
The results is shown in Figure \ref{fig:determinization_example:dsra}.
Now, assume that a run of the \nsra\ has reached $q_{1}$ via one run and $q_{4}$ via another run,
i.e. $q_{n,k} = q_{1}$ in Eq. \eqref{run:n} for the first of these runs and $q_{n,k} = q_{4}$ for the second.
Assume also that both $\phi_{1}$ and $\phi_{2}$ are triggered after reading the $(k+1)^{th}$ element,
but not $\phi_{3}$.
This means that the \nsra\ would move to $q_{2}$ and $q_{3}$.
In Eq. \eqref{run:n},
this would mean that $m=2$ and that $\delta_{n,k}^{1}.\phi = \phi_{1}$ and $\delta_{n,k}^{2}.\phi = \phi_{2}$.
But in the \dsra\ there is a transition that simulates this move of the \nsra.
The minterm $\phi_{1} \wedge \phi_{2} \wedge (\neg \phi_{3})$ moves the \dsra\ to $\{q_{2},q_{3}\}$.
It contains $\delta_{n,k}^{1}.\phi$ and $\delta_{n,k}^{2}.\phi$ in their positive form and all other conditions (here only $\phi_{3}$) in their negated form.
With a similar reasoning,
we see that the \dsra\ can simulate the \nsra\ for every other possible combination of $\{\phi_{1},\phi_{2},\phi_{3}\}$.

What we have proven thus far is a structural similarity between \nsra\ and \dsra.
We also need to prove that $\delta_{d,k}$ applies as well,
i.e., that the minterm on this transition is triggered exactly when its positive conjuncts are triggered.
To prove this,
we need to show that the contents of the registers that a condition $\phi$ of the \nsra\ accesses are the same that this $\phi$ accesses in the \dsra\ when participating in a minterm.

As we said,
the condition on $\delta_{d,k}$ is a conjunct (minterm),
where all $\phi_{n,k}^{j}$ appear in their positive form and all other conditions in their negated form.
But note that the conditions in negated form are those that were not triggered in $\varrho_{n,k+1}$ when reading the $(k+1)^{th}$ tuple.
Additionally,
the arguments passed to each of the conditions of the minterm are the same (registers) as those passed to them in the non-deterministic run
(by the construction algorithm, line \ref{line:determinization:register_selection}). 
To make this point clearer,
consider the following simple example of a minterm:
\begin{equation*}
\phi = \phi_{1}(r_{1,1},\cdots,r_{1,k}) \wedge \neg \phi_{2}(r_{2,1},\cdots,r_{2,l}) \wedge \phi_{3}(r_{3,1},\cdots,r_{3,m})
\end{equation*}
This means that $\phi_{1}(r_{1,1},\cdots,r_{1,k})$,
with the exact same registers as arguments,
will be the formula of a transition of the \nsra that was triggered.
Similarly for $\phi_{3}$.
With respect to $\phi_{2}$,
it will be the condition of a transition that was not triggered.
If we can show that the contents of those registers are the same in the runs of the \nsra\ and \dsra\ when reading the last tuple,
then this will mean that $\delta_{d,k}.\phi$ is indeed triggered.
But this is the case by the induction hypothesis
($v_{n,k}(r) = v_{d,k}(r)$),
since all these registers appear in the run $\varrho_{n,k}$ up to $q_{n,k}$.

The second part of the proposition also holds,
since,
by the construction,
$\delta_{d,k}$ will write to all the registers that the various $\delta_{n,k}^{j}$ write
(see line \ref{line:determinization:register_writing} in the determinization algorithm).

The third part also holds.
This is the part that actually ensures that the contents of the registers are the same.
First, 
note that a register can appear only once in a run of $A_{n}$,
because of its tree-like structure.
Second,
by the construction,
we know that $\delta_{d,k}.W = \bigcup\limits_{j=1}^{m} { \delta_{n,k}^{j}.W }$
(see again line \ref{line:determinization:register_writing} in the algorithm).
Therefore, we know that $\delta_{d,k}$ will write only to registers that had not appeared before in the run of the \nsra\ and will leave every other register that had appeared unaffected.
This observation is critical.
We could not claim the same for non-windowed \sra,
as in Figure \ref{fig:determinization}.
If we attempted to determinize this \nsra,
without unrolling its cycles,
the resulting \sra\ could overwrite $r_{1}$.
Now, since $\delta_{d,k}$ and all the $\delta_{n,k}^{j}$ write the same element
and $\delta_{d,k}$ does not affect any previously appearing registers,
the proposition holds.
\end{proof}

Since the above proposition holds for accepting runs as well,
we can conclude that there exists an accepting run of $A_{n}$ iff there exists an accepting run of $A_{d}$.
According to the above proposition,
the union of the last states over all $\varrho_{n}$ is equal to the last state of $\varrho_{d}$.
Thus, if $\varrho_{n}$ reaches a final state,
then the last state of $\varrho_{d}$ will contain this final state and hence be itself a final state.
Conversely, if $\varrho_{d}$ reaches a final state of $A_{d}$,
it means that this state contains a final state of $A_{n}$.
Then, there must exist a $\varrho_{n}$ that reached this final state.

\end{proof}

%% file: algo_determinization.tex
\begin{algorithm}
\KwIn{Windowed \srem\ $e' := e^{[1..n]}$}
\KwOut{Deterministic \sra\ $A_{d}$ equivalent to $e'$}
$A_{n} \leftarrow \mathit{ConstructWSRA}(e')$; \tcp{{\footnotesize See Algorithm \ref{algorithm:wsrem2sra}}}\
$Q_{d} \leftarrow \mathit{ConstructPowerSet}(A_{n}.Q)$\;
$\Delta_{d} \leftarrow \emptyset$;	$Q_{f,d} \leftarrow \emptyset$\;
\ForEach{$q_{d} \in Q_{d}$}{
	\If{$q_{d} \cap A_{n}.Q_{f} \neq \emptyset$}{
		$Q_{f,d} \leftarrow Q_{f,d} \cup \{ q_{d} \}$\;
	}
	$\mathit{Conditions} \leftarrow ()$;	$rs_{d} \leftarrow ()$\;
	\ForEach{$q_{n} \in q_{d}$}{
		\ForEach{$\delta_{n} \in A_{n}.\Delta: \delta_{n}.\mathit{source} = q_{n}$ }{
			$\mathit{Conditions} \leftarrow \mathit{Conditions} :: \delta_{n}.\phi$\;
			$rs_{d} \leftarrow rs_{d} :: \delta_{n}.\phi.rs$\; \label{line:determinization:register_selection}
		}
	}
	\tcc{$\mathit{ConstructMinTerms}$ returns the min-terms from a set of conditions. 
	For example, if $\mathit{Conditions} = (\phi_{1},\phi_{2})$, then $\mathit{MinTerms} = (\phi_{1} \wedge \phi_{2}, \neg \phi_{1} \wedge \phi_{2}, \phi_{1} \wedge \neg \phi_{2}, \neg \phi_{1} \wedge \neg \phi_{2})$}\
	$\mathit{MinTerms} \leftarrow \mathit{ConstructMinTerms}(\mathit{Conditions})$\; 
	\ForEach{$mt \in \mathit{MinTerms}$}{
		$p_{d} \leftarrow \emptyset$;	$W_{d} \leftarrow \emptyset$\;
		\ForEach{$q_{n} \in q_{d}$}{
			\ForEach{$\delta_{n} \in A_{n}.\Delta: \delta_{n}.\mathit{source} = q_{n}$ }{
				\tcc{$\phi \vDash \psi$ denotes entailment, i.e., if $\phi$ is true then $\psi$ is necessarily also true. 
				For example, $\phi_{1} \wedge \neg \phi_{2} \vDash \phi_{1}$.}\
				\If{$mt \vDash \delta_{n}.\phi$}{
					$p_{d} \leftarrow p_{d} \cup \{\delta_{n}.\mathit{target}\}$\;
					$W_{d} \leftarrow W_{d} \cup \{\delta_{n}.W\}$\; \label{line:determinization:register_writing}
				}
			}
		}
		$\delta_{d} \leftarrow \mathit{CreateNewTransition}(q_{d},mt(rs_{d}) \downarrow W_{d} \rightarrow p_{d})$\;
		$\Delta_{d} \leftarrow \Delta_{d} \cup \{\delta_{d}\}$\;
	}
}
$q_{d,s} \leftarrow \{A_{n}.q_{s}\}$\; \label{line:determinization:start_state}
$A_{d} \leftarrow (Q_{d},q_{d,s},Q_{f,d},A_{N}.R,\Delta_{d})$\;
$\mathtt{return}\ A_{d}$\;
\caption{Determinization.}
\label{algorithm:determinization}
\end{algorithm}

%% file: proofs_wsra_complement.tex
\ifdefined\FI
\begin{corollary}
Windowed \sra\ are closed under complement.
\end{corollary}
\else
\begin{corollary*}
Windowed \sra\ are closed under complement.
\end{corollary*}
\fi

\begin{proof}
Let $A$ be a windowed \sra.
We first determinize it to obtain $A_{d}$.
Although $A_{d}$ is deterministic,
it might still be incomplete,
i.e., there might be states from which it might be impossible to move to another state.
This may happen if it is possible that the conditions on all of the outgoing transitions of such a state are not triggered.
As in classical automata,
such a behavior implies that the string provided to the automaton is not accepted by it.

\input{algo_complement}

We can make $A_{d}$ complete by adding a so-called ``dead'' state $q_{dead}$ (non-final) to $A_{d}$.
See Algorithm \ref{algorithm:complement}.
For each state $q$ of $A_{d}$,
we then gather all the conditions on its outgoing transitions.
Let $\Phi$ denote this set of conditions.
We can then create the conjunction of all the negated conditions in $\Phi$:
$\phi_{dead} := (\neg \phi_{1}) \wedge (\neg \phi_{2}) \wedge \cdots \wedge (\neg \phi_{n})$,
where $\phi_{i} \in \Phi$ and $\bigcup\limits_{i=1}^{n} \phi_{i} = \Phi$.
We then add a transition from $q$ to $q_{dead}$ with $\phi_{dead}$ as its condition and $\emptyset$ as its write registers.
If we do this for every state $q \in A_{d}.Q$,
we will have created a \sra\ that is equivalent to $A_{d}$,
since transitions to $q_{dead}$ are only triggered if none of the other conditions in $\Phi$ are triggered.
If there exists a condition $\phi_{i}$ that is triggered,
the new automaton will behave exactly as $A_{d}$ and if no $\phi$ is triggered it will go to $q_{dead}$.
Now, if we add a self-loop transition on $q_{dead}$ with $\top$ as its condition,
we also ensure that the new automaton will always stay in $q_{dead}$ once it enters it.
$q_{dead}$ thus acts as a sink state.
This new automaton $A_{d,c}$ will therefore be equivalent to $A_{d}$ and it will also be both deterministic and complete. 

The final move is to flip all the states of $A_{d,c}$,
i.e., make all of its final states non-final and all of its non-final states final, to obtain an automaton $A_{complement}$.
This then ensures that if a string $S$ is accepted by $A$ (or $A_{d}$),
it will not be accepted by $A_{complement}$ and if it is accepted by $A_{complement}$ it will not be accepted by $A$.
This is indeed possible because $A$ (and $A_{complement}$) is deterministic and complete.
Therefore, 
for every string $S$,
there exists exactly one run of $A$ (and $A_{complement}$) over $S$.
If $A$, after reading $S$, reaches a final state,
$A_{complement}$ necessarily reaches a non-final state and vice versa.
Therefore,
for every windowed \sra\ $A$ we can indeed construct a \sra\ which accepts the complement of the language of $A$. 

Notice that this trick of flipping the states would not be possible if $A$ were non-deterministic.
To see this,
assume that $A$ is non-deterministic and at the end of $S$ it reaches states $q_{1}$ and $q_{2}$,
where $q_{1}$ is non-final and $q_{2}$ is final.
This means that $S$ is accepted by $A$.
If we flip the states of the non-deterministic $A$ to get its complement $A_{complement}$,
we would again reach $q_{1}$ and $q_{2}$,
where,
in this case,
$q_{1}$ is final and $q_{2}$ is non-final.
$A_{complement}$ would thus again accept $S$,
which is not the desired behavior for $A_{complement}$.
\end{proof}

%% file: algo_complement.tex
\begin{algorithm}
\SetAlgoNoLine
\KwIn{Windowed \sra\ $A$}
\KwOut{\sra\ $A_{complement}$ accepting the complement of $A$'s language}
$A_{d} \leftarrow \mathit{Determinize}(A)$; \tcp{{\footnotesize See Algorithm \ref{algorithm:determinization}.}}\
$q_{dead} \leftarrow \mathit{CreateNewState}()$\;
$\Delta_{dead} \leftarrow \emptyset$\;
\ForEach{$q \in A_{d}.Q$}{
	$\Phi \leftarrow \emptyset$\;
	\ForEach{$\delta \in A_{d}.\Delta: \delta.\mathit{source} = q$}{
		$\Phi \leftarrow \Phi \cup \delta.\phi$\;
	}
	$\phi_{dead} \leftarrow \top$\;
	\ForEach{$\phi_{i} \in \Phi$}{
		$\phi_{dead} \leftarrow \phi_{dead} \wedge  (\neg \phi_{i})$\;
	} 
	$\delta_{dead} \leftarrow \mathit{CreateNewTransition}(q,\phi_{dead} \downarrow \emptyset \rightarrow q_{dead})$\;
	$\Delta_{dead} \leftarrow \Delta_{dead} \cup \delta_{dead}$\;
}
$\delta_{loop} \leftarrow \mathit{CreateNewTransition}(q_{dead},\top \downarrow \emptyset \rightarrow q_{dead})$\;
$\Delta_{dead} \leftarrow \Delta_{dead} \cup \delta_{loop}$\;
$Q_{comp} \leftarrow A.Q \cup \{q_{dead}\}$\;
$q_{comp,s} \leftarrow A.q_{s}$\;
$Q_{comp,f} \leftarrow A.Q \setminus A.Q_{f}$\;
$R_{comp} \leftarrow A.R$\;
$\Delta_{comp} \leftarrow A.\Delta \cup \Delta_{dead}$\;
$A_{complement} \leftarrow (Q_{comp},q_{comp,s},Q_{comp,f},R_{comp},\Delta_{comp})$\; 
$\mathtt{return}\ A_{complement}$\;
\caption{Constructing the complement of a \sra\ ($\mathit{Complement}$).}
\label{algorithm:complement}
\end{algorithm}

%% file: proofs_streaming_sra.tex
\ifdefined\FI
\begin{proposition}
If $S=t_{1},t_{2},\cdots$ is a stream of elements from a universe $\mathcal{U}$ of a $\mathcal{V}$-structure $\mathcal{M}$, 
where $t_{i} \in \mathcal{U}$, 
and $e$ is a \srem\ over $\mathcal{M}$,
then, for every $S_{m..k}$, $S_{m..k} \in \mathcal{L}(e)$ iff $S_{1..k} \in \mathcal{L}(e_{s})$ (and $S_{1..k} \in \mathcal{L}(A_{e_{s}})$).
\end{proposition}
\else
\begin{proposition*}
If $S=t_{1},t_{2},\cdots$ is a stream of elements from a universe $\mathcal{U}$ of a $\mathcal{V}$-structure $\mathcal{M}$, 
where $t_{i} \in \mathcal{U}$, 
and $e$ is a \srem\ over $\mathcal{M}$,
then, for every $S_{m..k}$, $S_{m..k} \in \mathcal{L}(e)$ iff $S_{1..k} \in \mathcal{L}(e_{s})$ (and $S_{1..k} \in \mathcal{L}(A_{e_{s}})$).
\end{proposition*}
\fi

\begin{proof}
First, assume that $S_{m..k} \in \mathcal{L}(e)$ for some $m, 1 \leq m \leq k$
(we set $S_{1..0} = \epsilon$). 
Then, for $S_{1..k} = S_{1..(m-1)} \cdot S_{m..k}$, $S_{1..(m-1)} \in \mathcal{L}(\top^{*})$, 
since $\top^{*}$ accepts every string (sub-stream),
including $\epsilon$. 
We know that $S_{m..k} \in \mathcal{L}(e)$, 
thus $S_{1..k} \in \mathcal{L}(\top^{*}) \cdot \mathcal{L}(e) = \mathcal{L}(\top^{*} \cdot e) = \mathcal{L}(e_{s})$.
Conversely, assume that $S_{1..k} \in \mathcal{L}(e_{s})$.
Therefore, $S_{1..k} \in \mathcal{L}(\top^{*} \cdot e) = \mathcal{L}(\top^{*}) \cdot \mathcal{L}(e)$.
As a result, 
$S_{1..k}$ may be split as $S_{1..k} = S_{1..(m-1)} \cdot S_{m..k}$ such that $S_{1..(m-1)} \in \mathcal{L}(\top^{*})$ and $S_{m..k} \in \mathcal{L}(e)$. 
Note that $S_{1..(m-1)} = \epsilon$ is also possible, 
in which case the result still holds, 
since $\epsilon \in \mathcal{L}(\top^{*})$.
\end{proof}